\documentclass[english, 11pt]{article}
\usepackage{base}
\usepackage{macros}
\usepackage{minimacros}
\usepackage[utf8]{inputenc} 
\usepackage{lipsum}         
\usepackage{geometry}       
\usepackage{hyperref}       
\usepackage{graphicx}       
\usepackage{booktabs}       
\usepackage{amssymb}
\usepackage{amsfonts}
\usepackage{amsthm}
\usepackage{amsmath}
\usepackage{thm-restate}
\usepackage{thmtools}

\newcommand{\cCF}{\cP_{\text{Factors}}}
\newcommand{\var}{\textrm{var}}

\setlist[enumerate]{label=(\arabic*)}

\makeatletter
\@ifundefined{c@theorem}{}{\renewcommand*{\theHtheorem}{\thesection.\arabic{theorem}}}
\@ifundefined{c@lemma}{}{\renewcommand*{\theHlemma}{\thesection.\arabic{lemma}}}
\@ifundefined{c@corollary}{}{\renewcommand*{\theHcorollary}{\thesection.\arabic{corollary}}}
\@ifundefined{c@proposition}{}{\renewcommand*{\theHproposition}{\thesection.\arabic{proposition}}}
\@ifundefined{c@claim}{}{\renewcommand*{\theHclaim}{\thesection.\arabic{claim}}}
\@ifundefined{c@definition}{}{\renewcommand*{\theHdefinition}{\thesection.\arabic{definition}}}
\@ifundefined{c@remark}{}{\renewcommand*{\theHremark}{\thesection.\arabic{remark}}}
\@ifundefined{c@fact}{}{\renewcommand*{\theHfact}{\thesection.\arabic{fact}}}
\@ifundefined{c@equation}{}{\renewcommand*{\theHequation}{\thesection.\arabic{equation}}}
\makeatother

\title{On Factorization of Sparse Polynomials of Bounded Individual Degree}
\author{Aminadav Chuyoon\thanks{School of Computer Science and AI, Tel Aviv University. This research was funded by the European Union (ERC, EACTP, 101142020). Views and opinions expressed are however those of the author(s) only and do not necessarily reflect those of the European Union or the European Research Council Executive Agency. Neither the European Union nor the granting authority can be held responsible for them.}\\
\texttt{chuyoon1@mail.tau.ac.il}
\and
Amir Shpilka\footnotemark[1]\\
\texttt{shpilka@tauex.tau.ac.il}}
\date{\today}
\date{}

\begin{document}

\maketitle

\begin{abstract}
    We study sparse polynomials with bounded individual degree and the class of their factors.
    In particular we obtain the following algorithmic and structural results:
    \begin{enumerate}
        \item A deterministic polynomial-time algorithm for finding \emph{all the sparse divisors} of a
        sparse polynomial with bounded individual degree. As part of this, we establish the first upper bound on the number of non-monomial irreducible factors of such polynomials. 
        
        \item  \sloppy A $\poly(n,s^{d\log \ell})$-time algorithm for recovering $\ell$ irreducible
        $s$-sparse polynomials of bounded individual degree $d$ from blackbox access to their product (which is not necessarily sparse). This partially resolves
        a question posed in \cite{DuttaST24}. 
        In particular, when $\ell=O(1)$, the algorithm runs in polynomial time.
        
        \item \sloppy Deterministic algorithms for factoring a \emph{product} of $s$-sparse polynomials of bounded individual degree $d$ from blackbox access. Over fields of characteristic zero or sufficiently large, the algorithm runs in $\poly(n,s^{d^3\log n})$-time; over arbitrary fields it runs in  $\poly(n,{(d^2)!},s^{d^5\log n})$-time. This improves upon the algorithm of \cite{bhargava2020deterministic},
        which runs in $\poly(n,s^{d^7\log n})$-time and applies only to a single sparse polynomial of bounded individual degree. In the case where the input is a single sparse polynomial,
        we give an algorithm that runs in $\poly(n,s^{d^2\log n})$-time.

        \item Given blackbox access to a \emph{product}
        of (not necessarily sparse or irreducible) \emph{factors of sparse polynomials of bounded individual degree}, we give a deterministic polynomial-time algorithm for finding all irreducible sparse multiquadratic factors of it (along with their multiplicities). This generalizes
        the algorithms of \cite{volkovich2015deterministically} and \cite{volkovich2017some}. We also show how to decide whether such a product is a complete power (in case it is defined over a field of zero or large enough characteristic), extending the algorithm of \cite{bisht2025solving}.
    \end{enumerate}
    
    Our algorithms most naturally apply over fields of zero or sufficiently large characteristic. To handle arbitrary fields, we introduce the notion of \emph{primitive divisors} for a class of polynomials, which may be of independent interest. This notion enables us to adapt ideas of \cite{bisht2025solving} and remove characteristic assumptions from most of our algorithms.
\end{abstract}

\thispagestyle{empty}
\newpage

\tableofcontents

\thispagestyle{empty}
\newpage

\clearpage
\pagenumbering{arabic}

\section{Introduction}
\subsection{Background and Motivation}\label{sec:background}
This paper is motivated by numerous open problems concerning the factorization and complexity of factors of multivariate sparse polynomials. These are $n$-variate polynomials of degree $\poly(n)$ whose number of monomials is also $\poly(n)$. In particular, this is the class of polynomials that have polynomial size $\Sigma\Pi$ circuits, which is one of the simplest classes from the algebraic complexity point of view.

Despite being structurally simple, sparse polynomials form a rich and non-trivial class of polynomials that has been extensively studied. Efficient deterministic algorithms for polynomial identity testing and reconstruction were developed for this class \cite{grigoriev1987matching,ben1988deterministic,GrigorievKS90,clausen1991zero,KS01}.
In particular, the work of Klivans and Spielman \cite{KS01} introduced a special reduction from multivariate sparse polynomials to univariate polynomials of polynomial degree, over any field, that has since found many applications, mainly for questions concerning polynomial identity testing and reconstruction. See the surveys \cite{saxena2009progress,SY10,saxena2014progress,DuttaGhosh-survey}. 

In this work we are mostly interested in understanding factors of sparse polynomials. The seminal work of Kaltofen \cite{kaltofen1985computing} shows that factors of sparse polynomials have polynomial-size algebraic circuits. A recent breakthrough result of  \cite{bhattacharjee2025closure} implies that they even have polynomial-size bounded-depth circuits. 

However, an intriguing  question that remains unsolved for over forty years concerns the complexity of factors of sparse polynomials in terms of their number of monomials. In other words, how many monomials can a factor of a sparse polynomial have.  
The following example of von zur Gathen and Kaltofen demonstrates that the sparsity of irreducible factors of polynomials may be super-polynomial.

\begin{example}[Example of \cite{von1985factoring}]\label{ex:qp-sparse}
Consider 	
\[
g(\vx) = \prod_{i=1}^{ n} (x_i - 1)
\quad\text{and}\quad
h(\vx) = \prod_{i=1}^{ n} (1 + x_i + \cdots + x_i^{n-1})+n.
\]	
It is not hard to prove that $h$ is irreducible over $\Q$.
Let
\[
f(\vx) = g\cdot h = \prod_{i=1}^{  n} (x_i^n - 1) + n \prod_{i=1}^{ n} (x_i - 1).
\]
Then, $f$ has $(2^{n+1}-1)$ monomials, and an irreducible factor $h$ of sparsity $n^{n}$, which is super-polynomial in the sparsity of $f$. 
\end{example}

This led them to pose the following question.
\begin{question}[{\cite{von1985factoring}}]\label{q:when-factors-sparse}
    Can the output size for the factoring problem  actually be more than quasi-polynomial in the sparsity of the input?
\end{question}

\sloppy This question is still wide open with only several special cases solved.
In 
\cite{volkovich2015deterministically} Volkovich proved that multilinear factors of sparse polynomials are sparse.  In a follow-up work \cite{volkovich2017some} he proved that all factors of multiquadratic sparse polynomials are themselves sparse. This result was extended to higher individual degrees by Bhargava, Saraf, and Volkovich \cite{bhargava2020deterministic}, who
considered sparse polynomials with bounded individual degree $d$, and proved a quasi-polynomial upper bound on the sparsity of factors of such polynomials. In particular, if $f$ is an $s$-sparse $n$-variate polynomial with individual degrees bounded by $d$, then every factor of $f$ has at most $s^{O(d^2 \log n)}$ monomials (see \Cref{thm: BSV bound}). This is the first nontrivial sparsity bound for factors in the regime $d>2$. 
While the result of \cite{bhargava2020deterministic} gives a quasi-polynomial upper bound on sparsity of factors of polynomials with bounded individual degree, it is believed that the "real" bound should be polynomial. 
In \cite{bisht2025solving}, the authors provide some
evidence for that by giving polynomial-time algorithms for various problems for which the existence of such algorithms is implied by the conjectured polynomial upper bound on sparsity of factors.

In their work von zur Gathen and Kaltofen gave a randomized algorithm for factorization of sparse polynomials and noted another obstacle towards obtaining deterministic factorization algorithms: 
\begin{description}
    \item ``Also note that there does not seem to be a feasible way of checking deterministically whether the output is correct, i.e., whether $f=\prod_i {f_i}^{e_i}$''  \cite{von1985factoring}.
\end{description} 
In particular, the following question is open.

\begin{question}[{\cite[Question 1]{DuttaST24}}]\label{q:factoring-sparse-irred-product}
    Can we find the factorization of polynomials whose irreducible factors are sparse?
\end{question}

In fact, Dutta et al. also asked the following special case of \Cref{q:factoring-sparse-irred-product}.

\begin{question}[{\cite[Section 5]{DuttaST24}}]\label{q:dst-factor-product-nsd}
    Given a blackbox computing the product of sparse irreducible polynomials $f_i$ with bounded individual degree, find the $f_i$'s in deterministic polynomial time.
\end{question}

The aforementioned results of \cite{volkovich2015deterministically,volkovich2017some,bhargava2020deterministic} also have consequences for \Cref{q:factoring-sparse-irred-product} and \Cref{q:dst-factor-product-nsd}. 
In \cite{volkovich2015deterministically} Volkovich gave a deterministic polynomial-time algorithm for factorization of sparse polynomials
whose irreducible factors are all multilinear (hence sparse) polynomials. In \cite{volkovich2017some} he gave an efficient deterministic algorithm for factoring multiquadratic sparse polynomials.  In \cite{bhargava2020deterministic} the authors combined
their sparsity bound with techniques similar to \cite{volkovich2015deterministically,volkovich2017some} to obtain a deterministic quasi-polynomial-time algorithm for factoring  sparse polynomials of bounded individual degree.
Dutta, Sinhababu, and Thierauf \cite{DuttaST24} gave efficient deterministic reconstruction of constant degree factors of
sparse polynomials.

One notable limitation of the algorithm of \cite{bhargava2020deterministic} is that it does not give an efficient way to find the \emph{sparse factors} of the input polynomial, even if such factors exist.
This problem is closely related to the following question asked by Dutta, Sinhababu and Thierhauf.

\begin{question}[{\cite[Question 2]{DuttaST24}}]\label{q:find-sparse-factors}
    Design a better-than-exponential time deterministic algorithm that outputs all the sparse irreducible factors of a sparse polynomial.
\end{question}

Another outstanding open problem concerns the reconstruction of rational functions whose numerator and denominator are sparse polynomials. Surprisingly, the best known deterministic algorithm for this question runs in exponential time \cite{grigoriev1994computational}. While seemingly unrelated, special cases of this question lie at the heart of our algorithms, and we solve it in two special cases where more information is provided.

Many more open problems were asked for sparse polynomials and 
the reader is referred to \cite{DuttaST24} and \cite{forbes2015complexity}.

\subsection{Our Results}\label{sec:our results}
Before  presenting our results, let us introduce the following notation.

\begin{definition}\label{def:notation}
We denote the class of $n$-variate, $s$-sparse polynomials of bounded individual degree $d$, by $\cP(n,s,d)$. We refer to polynomials in this class as $(n,s,d)$-sparse polynomials. It turns out that in many cases it makes sense to talk about the class of \emph{divisors} of $(n,s,d)$-sparse polynomials; we denote this class of polynomials by $\cCF(n,s,d)$.    
\end{definition}

We first summarize our main contributions informally as follows:
\begin{enumerate}
    \item We prove a sharp upper bound on the number of irreducible factors of $(n,s,d)$-sparse polynomials (\Cref{thm:bound on the number of irreducible factors of sparse polynomials}).
    \item Using this bound, we give the first deterministic polynomial-time algorithm that constructs \emph{all} sparse divisors of $(n,s,d)$-sparse polynomials (\Cref{finding sparse divisors of sparse polynomials}).
    \item We give a deterministic blackbox factorization algorithm for products of sparse $(n,s,d)$-sparse polynomials in quasi-polynomial time in the number of terms. This solves \Cref{q:dst-factor-product-nsd} when the number of terms is constant (\Cref{thm:factor-of-product-of-nsd}). We note that the running time of this algorithm is polynomial in the number of variables $n$.
\end{enumerate}

\begin{remark}
In stating our results we suppress the precise dependence on the field (beyond its characteristic) and on bit complexity. As usual in factorization algorithms, an additional factor reflecting the complexity of univariate factorization over the field should be included. These bounds are stated in \Cref{factorization background}.    
\end{remark}

We now give the formal statement of our results.
Our first main result gives an upper bound for the number of sparse divisors of $(n,s,d)$-sparse polynomials. While its proof is not difficult, it is instrumental in our algorithms and demonstrates the strong structural limitations that sparsity imposes.

\begin{restatable}[{\bf Divisor bound for $(n,s,d)$-sparse polynomials}]{theorem}{NumDivisors}
    Let $f\in \cP(n,s,d)$.
    The number of non-monomial irreducible factors of $f$, counted with multiplicities, is at most $d\log s$. In particular, the number of divisors of an $s$-sparse polynomial $f$ with bounded individual degree $d$, that are not divisible by any monomial, is at most $2^{d\log s}=s^{d}$.
    \label{thm:bound on the number of irreducible factors of sparse polynomials}
\end{restatable}

Our first main algorithmic result shows how to deterministically output all sparse divisors of such polynomials in polynomial time. Earlier, only a quasi-polynomial-time algorithm was known \cite{bhargava2020deterministic}.

\begin{restatable}[{\bf Finding all sparse divisors of $(n,s,d)$-sparse polynomials}]{theorem}{SparseDivisorFinding}
    There is a deterministic, $\poly(n,d!,s^{d})$-time algorithm which, given $f\in \cP(n,s,d)$, outputs all the
    $s$-sparse divisors of $f$ that are not divisible by any monomial, as well as the monomial of maximal degree dividing $f$.
    \label{finding sparse divisors of sparse polynomials}
\end{restatable}

Our second main algorithmic result answers \Cref{q:dst-factor-product-nsd} when there are constantly many polynomials in the product, and in the general case it provides a quasi-polynomial-time algorithm.

\begin{restatable}[{\bf Factoring products of $(n,s,d)$-sparse polynomials}]{theorem}{OpenProblemTheorem}
    Let $\mathbb{F}$ be a field of characteristic zero or larger than $2d$.
    There is a deterministic $\poly(n,d^d,s^{d\log \ell},\ell^d)$-time algorithm which, given blackbox access to a product $f=\prod_{i=1}^{\ell}{\phi_i}$, where
    $\phi_1,\ldots,\phi_\ell\in \cP(n,s,d)$ are all (not necessarily different) irreducible, returns the $\phi_i$-s.
    Moreover, when the assumption on the field is removed, there is an algorithm that solves this problem in $\poly(n,{(d^2)!},s^{d\log \ell+d^3},\ell^d)$-time.
    \label{thm:factor-of-product-of-nsd}
\end{restatable}

This theorem will follow as a consequence of the next, more general, result.

\begin{restatable}[{\bf Factoring products of factors of sparse polynomials}]{theorem}{generalfactorizationalgorithm}\label{General factorization algorithm}
    Let $\mathbb{F}$ be a field with $\Char(\F)=0$ or $\Char(\F)>2d$.
    There is a deterministic $\poly(n,d^d,s^{d\log \ell},\ell^d)$-time algorithm which, given blackbox access to a product $f$ of $\ell$ polynomials in $\cCF(n,s,d)$ over the field $\mathbb{F}$, returns all the $s$-sparse divisors of bounded individual degree $d$ of $f$ that are not divisible by any monomial, as well as the multiplicity $k_i$ of each $x_i$ as a factor of $f$. Moreover, if $f$ was in fact a product of polynomials in $\cP(n,s,d)$, there is such an algorithm that works over arbitrary fields, whose running time is $\poly(n,{(d^2)!},s^{d\log \ell+d^3},\ell^d)$.
\end{restatable}

Observe that this result is slightly weaker for fields of small characteristic. Indeed, in \Cref{Section 4} we will see that some of our methods only work for zero characteristic fields, or for fields with large enough characteristic.

In addition to our main results we provide several improvement and generalization of earlier work. 
We start by stating  results concerning factorization of sparse polynomials of bounded individual degree into irreducible factors. The first is an improvement of the best known factorization algorithm of such polynomials \cite{bhargava2020deterministic}.

\begin{restatable}{theorem}{SparseFactorization}
    There is a deterministic $\poly(n,s^{d^2\log n})$-time algorithm which, given an $s$-sparse polynomial of bounded individual degree $d$ on $n$-variables, returns its factorization into  irreducible factors.
    \label{improved BSV}
\end{restatable}

The second theorem generalizes the result of \cite{bhargava2020deterministic} to a significantly larger class of polynomials, while achieving improved running time  compared to their algorithm.
\begin{restatable}{theorem}{SparseProductFactorization}
    Let $\mathbb{F}$ be a field with $\Char(\F)=0$ or $\Char(\F)>2d$. 
    There is a deterministic $\poly(n,s^{d^3\log n},(d\ell)^d)$-time algorithm which, given blackbox access to a product $f$ of $\ell$ polynomials in $\cCF(n,s,d)$ over $\mathbb{F}$, returns its factorization into  irreducible factors.
    Moreover, there is an algorithm that works over arbitrary fields, whose running time is $\poly(n,{(d^2)!}, s^{d^5\log n},\ell^d)$.
    \label{improved BSV for products}
\end{restatable}

Note that this last result considers the class of \emph{factors of sparse polynomials of bounded individual degree}.
Indeed, we will see that it is more natural to consider this class in many cases, and several results of ours actually hold for this class rather than just for the class of sparse polynomials.

Finally, we generalize results of \cite{volkovich2015deterministically,volkovich2017some,bisht2025solving} in a similar manner to our generalization of the results of \cite{bhargava2020deterministic}.

Our next algorithm finds all the multiquadratic irreducible factors of a product of polynomials from $\cCF(n,s,d)$. Earlier, only factorization of polynomials whose irreducible factors are all multiquadratic and $s$-sparse was known \cite[Theorem 47]{volkovich2017some}.

\begin{restatable}{theorem}{MultiquadraticRecovery}
    Let $\mathbb{F}$ be a field of characteristic zero or larger than $2d$.
    There is a deterministic $\poly(n,d!,s^{d},\ell)$-time algorithm that, given blackbox access to a product $f$ of $\ell$ elements of $\cCF(n,s,d)$ defined over $\mathbb{F}$, returns all the $s$-sparse multiquadratic irreducible factors of $f$. In particular, the algorithm returns all the multilinear factors of $f$ (as these are $s$-sparse by \Cref{multilinear factors of sparse polynomials are sparse}).
    Moreover, there exists such an algorithm that works
    over arbitrary fields, whose running time is $\poly(n,(d^2)!,s^{d^2},\ell)$.
    \label{finding multiquadratic factors of a product of sparse polynomials}
\end{restatable}

We give a similar improvement to the perfect power testing of \cite{bisht2025solving}.

\begin{restatable}{theorem}{GeneralPower}
    Let $\mathbb{F}$ be a field of characteristic zero or larger than $2d$. There is a deterministic $\poly(n,d!,s^d,\ell)$-time algorithm that, given blackbox access to a polynomial $f$ that is a product of $\ell$ polynomials in $\cCF(n,s,d)$ (defined over $\mathbb{F}$) and a positive integer $e$, returns if $f$ is a complete $e$-th power.
    \label{general complete power testing}
\end{restatable}

Finally, we extend the polynomial identity testing result of \cite{bisht2025solving}.
\begin{restatable}{theorem}{GeneralDivisibility}
    Let $\mathbb{F}$ be a field of characteristic zero or larger than $2d$. There exists a deterministic $\poly(n,d!,s^{d},\ell)$-time algorithm that, given blackbox access to two polynomials (over $\mathbb{F}$) $f$ and $g$ that are products of $\ell$ polynomials in $\cCF(n,s,d)$, decides if $f|g$.
    Moreover, if $f$ and $g$ are actually products of $\ell$ polynomials in $\cP(n,s,d)$, there is
    an algorithm that  works over arbitrary fields, whose running time is $\poly(n,{(d^2)}!,s^{d^3},\ell)$
    \label{general divisibility testing}
\end{restatable}

As mentioned in \Cref{sec:background}, we will have to solve some rational interpolation problems. Somewhat surprisingly, it is not known
how to deterministically reconstruct two coprime $s$-sparse polynomials $f$ and $g$ from blackbox access to their quotient $g/f$ in polynomial time.
We shall see that with some additional information on $f$, this is in fact possible.
Without this extra information, only exponential deterministic algorithms
are known for solving this problem  \cite{grigoriev1994computational}.
We note that efficient randomized algorithms for solving this problem are known, see e.g., \cite{cuyt2011sparse,HoevenLecerf21}.


\subsection{Related Work}
The questions considered in this work have been studied extensively in recent years. We organize the prior results around three themes that are central to this paper: sparsity bounds, factorization and reconstruction algorithms, and divisibility and identity testing.

\paragraph{Sparsity of Factors of Sparse Polynomials.} Volkovich \cite{volkovich2015deterministically} showed that multilinear factors
of $s$-sparse polynomials are again $s$-sparse. In \cite{volkovich2017some} he proved that  factors of an $s$-sparse multiquadratic
polynomial are $s$-sparse as well. For general polynomials of bounded individual degree $d$, Bhargava, Saraf and Volkovich \cite{bhargava2020deterministic}, obtained the first non-trivial bound on the sparsity of factors of $s$-sparse polynomials with bounded individual degree $d$, namely $s^{O(d^{2}\log n)}$. Bisht and Saxena  \cite{bisht2025derandomization} improved this bound 
to $s^{\poly(d)}$ for the class of \emph{symmetric} $s$-sparse polynomials of bounded individual degree $d$. Bisht and Volkovich \cite{bisht2025solving} obtained
an  $s^{d\log s}$ bound on the sparsity of a quotient of an $s$-sparse polynomial of bounded individual degree $d$ by a multilinear factor. Finally, the work of Forbes \cite{forbes2015deterministic} and  results in \cite{bisht2025solving} imply that the sparsity of the quotient of $f\in\cP(n,s,d)$ by a degree $r$ polynomial is bounded by $(ns)^{O(dr)}$.

\paragraph{Deterministic factorization and reconstruction.} Factorization results have largely paralleled the development of sparsity bounds. Volkovich showed in \cite{volkovich2015deterministically} how to efficiently factorize a
sparse polynomial whose factors are all multilinear. In \cite{volkovich2017some} he demonstrated
how to efficiently compute the irreducible factors of a polynomial from a blackbox access, given that all of them are multiquadratic and sparse.
In  \Cref{finding multiquadratic factors of a product of sparse polynomials} we further improve this result
by relaxing the condition on the input polynomial - we allow it to be a product of factors of some sparse polynomial of bounded individual degree, and not simply a product of sparse multiquadratic polynomials.

Using their sparsity bound, Bhargava, Saraf, and Volkovich \cite{bhargava2020deterministic} derived a deterministic quasi-polynomial algorithm for factoring sparse polynomials of bounded individual degree.  However, their methods do not yield a polynomial-time algorithm for finding the sparse divisors
of the input polynomial, so they do not imply \Cref{finding sparse divisors of sparse polynomials}. Furthermore, their algorithm does not extend to factorization of products of sparse, bounded degree polynomials, let alone products of factors of such polynomials.

Kumar, Ramanathan and Saptharishi \cite{kumar2023deterministicalgorithmslowdegree} and
Dutta, Sinhababu, and Thierauf \cite{DuttaST24},  gave an algorithm that outputs all the irreducible factors of constant degree of a sparse polynomial
in quasi-polynomial time. As their result applies to general sparse polynomials without restrictions on the individual degree, it does not follow from the methods of this paper.

In \cite{bisht2025solving}, an algorithm for deciding if a sparse polynomial of bounded individual degree is a complete power is given;
in \Cref{sec:char-0} we use our stronger methods to obtain this result when the input polynomial is a product of polynomials in $\cCF(n,s,d)$, and
we have blackbox access to it. However, this generalization does not  work when the underlying field is of small characteristic.

\paragraph{Divisibility and identity testing.}

In \cite{volkovich2017some} it is also shown how to efficiently check if two $s$-sparse polynomials of bounded individual degree $d$ divide each other;
this is improved in \Cref{general divisibility testing}, where we show a divisibility testing result for two products of factors of
sparse polynomials of bounded individual degree, given blackbox access to them. 
In the same work,  Volkovich also gave a  deterministic algorithm for testing a factorization of sparse polynomials, with
constant individual degrees, into sparse irreducible factors. That is, testing if $f=\prod g_i$
when $f$ has constant individual degrees and $g_i$-s are irreducible and sparse. We note that by  \Cref{cor:factoring-product-sparse}  we can efficiently find the $g_i$'s and verify that $f=\prod g_i$.

Forbes \cite{forbes2015deterministic} gave a quasi-polynomial-time algorithm for checking if a constant
degree polynomial divides a sparse polynomial. Note that while our result allows the dividing polynomial to be
much more complicated, Forbes' result requires no restriction on the individual degree of the given sparse polynomial. Therefore,
our methods fall short in proving his result, or any other result on sparse polynomials with no restriction on their individual degree.

In \cite{bisht2025solving}, the authors design an efficient blackbox PIT algorithm for the class $\Sigma^{[2]}\Pi\Sigma\Pi^{[\text{ind-deg d
}]}$; while doing so, they obtain some technical results about the class $\cCF(n,s,d)$. In \Cref{sec:arb-char} we show that
their methods imply a stronger result than this PIT. Showing PIT for this class is essentially equivalent to giving an algorithm for determining if two products $f,g$ of $(n,s,d)$-sparse polynomials are equal, given blackbox access to them. We show that the methods of \cite{bisht2025solving} can be used for determining
if $f$ divides $g$. In \Cref{sec:char-0} we further generalize and strengthen the technical results of \cite{bisht2025solving}, but only for fields of zero or sufficiently large characteristic.

\subsection{Proof Overview}

We start by describing some technical results about the class of sparse polynomials of bounded individual degree needed for our proofs. We then turn to give an overview of our results regarding factorization of sparse polynomials. We first discuss our result on the recovery of all irreducible sparse multiquadratic factors of a product of $(n,s,d)$-sparse polynomials, as the proof contains some of the important ideas used in the proofs of our main results. We then move on to describe the recovery of all sparse divisors of sparse polynomials of bounded individual degree and our partial solution to \Cref{q:dst-factor-product-nsd}. Finally, we briefly describe how to use our methods for factoring general $(n,s,d)$-sparse polynomials and products of them, making use of the bound of \cite{bhargava2020deterministic}.   

\subsubsection{Sparse Polynomials of Bounded Individual Degree and Their Factors}\label{sec:methods outline}

The main tool we use for obtaining our results is a generator $G:\mathbb{F}^6 \rightarrow \mathbb{F}^n$ that preserves important properties of polynomials in $\cCF(n,s,d)$. Roughly speaking, we want this map $G$ to satisfy the following two properties:
\begin{enumerate}
    \item Its image $\textrm{Im}(G)$ is an interpolation set for polynomials in $\cP(n,s,d)$.
    \item If $\phi$ and $\psi$ are non-associate\footnote{Two polynomials are associate if one is a scalar product of the other.} irreducible elements of $\cCF(n,s,d)$, then $\phi(G_i)$ and $\psi(G_i)$ are coprime as polynomials in $x_i$, where  $G_i$ should be thought of as ``$G$ when reviving the $i$-th variable'' (i.e., keeping $x_i$ alive and substituting the remaining variables according to $G$).
\end{enumerate}

The first property is easily achieved by the basic properties of the Klivans-Speilman generator (KS-generator), introduced in \cite{KS01}. We show that the second property follows from the first, together with the fact that resultants of sparse polynomials of bounded individual degree are again sparse. However, we are able to do so only when the underlying field $\mathbb{F}$ is of zero characteristic, or if its characteristic is larger than twice the bound on the individual degree. If this is not the case then we are only able to prove something a little weaker - we prove that the second property holds if one replaces the notion of irreducible polynomials with the new notion of \emph{primitive divisors}, which we first define in this paper (we describe this notion below).

Before sketching the proof that the first property of $G$ implies the second, we note that proving that will automatically imply some algorithms for sparse polynomials. For example, it is straightforward to see that by factoring the composition of the input polynomials with $G_i$, the following two results can be obtained from the (general) second property of $G$. 
\begin{enumerate}
    \item Given blackbox access to two products of elements of $\cCF(n,s,d)$,
    we can decide if one divides the other in deterministic polynomial time (\Cref{general divisibility testing}).
    \item Given blackbox access to a product of elements of $\cCF(n,s,d)$,
    we can decide if it is a complete power in deterministic polynomial time (\Cref{general complete power testing}).
\end{enumerate}
This already generalizes results from \cite{volkovich2017some}
and \cite{bisht2025solving}.
As we are able to obtain the general second property of $G$ only for some fields, these results do not apply to arbitrary fields. It turns out, however, that the aforementioned weaker version of the second property of $G$ is strong enough to obtain \Cref{general divisibility testing} over arbitrary fields, in the case we replace $\cCF(n,s,d)$ by $\cP(n,s,d)$ (see \Cref{general characteristic divisibility testing}).

In the rest of this subsection we describe our proofs of the second property, both for arbitrary fields  and for fields with zero or large enough characteristic.

\paragraph{Arbitrary Fields.}\label{arbitraty field sec:methods outline}
In \cite[Lemmas 4.5,4.6]{bisht2025solving} it is shown that the second property of $G$ holds for each pair of irreducible polynomials $\phi$ and $\psi$ that satisfy a certain non-degeneracy condition: there exist $(n,s,d)$-sparse polynomials $f,g$ for which $\phi|f, \psi|g$ and the matrix
$\begin{pmatrix}
    v_{\phi}(f) & v_{\phi}(g) \\
    v_{\psi}(f) & v_{\psi}(g)
\end{pmatrix}$
is invertible, where $v_{\phi}(f)$ is the multiplicity of $f$ as a factor of $\phi$, and so are the other entries, respectively.
That is, a pair of polynomials $\phi, \psi\in \cCF(n,s,d)$ does not 
satisfy this condition if the multiplicity of $\phi$ as a factor of an element $f$ of $\cP(n,s,d)$ determines the multiplicity of $\psi$ as a factor of $f$. In this sense, $\phi$ and $\psi$ are "inseparable" from the point of view of polynomials in $\cP(n,s,d)$. 

For each irreducible $\phi\in \cCF(n,s,d)$, consider all $\psi$s that are inseparable from it, in the sense that $\phi, \psi$ do not satisfy the aforementioned non-degeneracy condition. For each $f\in \cP(n,s,d)$, consider the vector of multiplicities of these $\psi$s as factors of $f$ - $(v_{\psi}(f))_{\psi}$. By the degeneracy assumption, these vectors are all proportional. Consider the $\mathbb{Z}$-span of these vectors, and let $(a_\psi)_{\psi}$ be a generator for this one dimensional lattice. The polynomial $\prod_{\psi}{\psi^{a_{\psi}}}$ is called a primitive divisor for the class $\cP(n,s,d)$.
Note that by construction, as the property of being "inseparable" is an equivalence relation on the class of irreducible polynomials in $\cCF(n,s,d)$, it must hold that every two primitive divisors are coprime.

The notion of a primitive divisor is a general notion, that applies to a general class of polynomials $\cP$, and therefore may be of independent interest; a formal definition of it is given in \Cref{def:primitive divisor}. Primitive divisors for a class of polynomials $\cP$ are for the class $\cP$ what irreducible polynomials are for the class of all polynomials, in the sense that every element of the class $\cP$ can be written as a product of primitive divisors for $\cP$ in a unique way; this is not too difficult to see by the construction given in the previous paragraph, and is formally done in \Cref{factorization to primitive divisors}. Therefore, it makes sense to talk about the multiplicity $v_{\phi}(f)$ of a primitive divisor $\phi$  as a factor of  $f\in\cP$.

With this definition in hand, one can show that every two non-associate primitive divisors $\phi$ and $\psi$ for a class $\cP$ satisfy the non-degeneracy condition of \cite{bisht2025solving}. As a result, we are able to apply the proof of the aforementioned result of \cite{bisht2025solving} to obtain the weaker version of the second property of $G$, namely, for every two non-associate primitive divisors $\phi, \psi$ for $\cP(n,s,d)$ it  
holds that $\gcd_{x_i}(\phi(G_i), \psi(G_i))=1$.

We further note that for every two products $f,g$ of elements of $\cP(n,s,d)$ we have that $f|g$ if and only if $v_\phi(f)\leq v_\phi(g)$ for every primitive divisor $\phi$ for $\cP(n,s,d)$; this is a simple corollary of our unique-factorization theorem. As a result, the weaker second property of $G$  is enough for giving a divisibility testing for $f$ and $g$: just
check if $f(G_i)|g(G_i)$ as polynomials in $x_i$ for every $i$. This is formally done in \Cref{general characteristic divisibility testing}.

\paragraph{Zero or Large Characteristic.}\label{special fields sec:methods outline}

In the proof of \cite[Lemma 4.5]{bisht2025solving}, Bisht and Volkovich  consider the Sylvester matrix $M_{x_i}(f^{\alpha},g^{\beta})$ of the some suitably chosen small powers of the polynomials $f,g$, viewed as polynomials in $x_i$. They then use the first property of $G$ to prove that the rank of this matrix does not change when composing with the generator $G_i$, that is,
$\rank(M_{x_i}(f^{\alpha},g^{\beta}))=\rank(M_{x_i}(f(G_i)^{\alpha},g(G_i)^{\beta}))$.
From that, they deduce that it must be the case that $\phi(G_i)$ and $\psi(G_i)$ do not share factors, as the existence of such a factor would result in rank decrease.

To obtain our result, instead of considering the Sylvester matrix $M_{x_i}(f^\alpha, g^\beta)$, we consider the matrix defining the discriminant $\Delta(fg)$
of the polynomial $fg$. That is, we consider the Sylvester matrix $M_{x_i}(fg, (fg)')$ of $fg$ and its formal derivative with respect to $x_i$, $(fg)'$. We next explain the idea behind it.

Suppose $\phi|f, \psi|g$ are two irreducible elements of $\cCF(n,s,d)$. View all polynomials as polynomials in some variable $x_i$, and
consider the Sylvester matrix $M_{x_i}(fg,(fg)')$. It is not hard to prove that its rank essentially counts the number of unique irreducible factors of $fg$ (when appropriately counted; for the formal statement, see \Cref{cla:sum of multiplicities captured by rank of the Sylvester matrix of f and f'}).
As all the determinants of minors of $M_{x_i}(fg,(fg)')$ are sparse polynomials of low degree, we are  able to show, using the first property of $G$,
that the rank of this matrix does not  change when composing with $G_i$. In particular,
this means that $\phi(G_i), \psi(G_i)$ cannot  share a common factor, as otherwise the number of unique factors of $fg(G_i)$ would be smaller than that of $fg$, in contradiction
to the fact that the rank remains unchanged.
With this result in hand, we obtain \Cref{general divisibility testing} and \Cref{general complete power testing} 
by considering the factorization of the inputs composed with the generators $G_i$.

\subsubsection{Finding Multiquadratic Factors of Products of Sparse Polynomials}\label{multiquadratic outline}
 
We next show how to use the results of \Cref{sec:methods outline} to obtain \Cref{finding multiquadratic factors of a product of sparse polynomials}. 
For the sake of presentation we assume here that $\operatorname{char}(\mathbb{F})=0$.

We first consider the following general setting. Let $f$ be a product of polynomials in $\cCF(n,s,d)$; that is, $f$ is a product of factors of $s$-sparse polynomials of bounded individual degree $d$. Consider some irreducible $\phi \in \cCF(n,s,d)$. By the second property of $G$ it follows that for every $i$ for which $x_i\in \textrm{var}(\phi)$, the highest power of $\phi(G_i)$ that divides $f(G_i)$ is precisely the multiplicity of $\phi$ as a factor of $f$. In particular, this is true whenever $\phi$ is an irreducible $s$-sparse polynomial of bounded individual degree $d$. 
With this observation in mind, we are now ready to describe the algorithm for finding multilinear factors.
Suppose $f=\prod_{i=1}^{\ell}{\phi_i^{e_i}}\in \mathbb{F}[x_0,\vx]$ where $\phi_i$ are factors of $(n+1,s,d)$-sparse polynomials. Suppose further that
we have blackbox access to $f$, and we want to recover all of its irreducible multilinear factors. \footnote{The general algorithm finds the multiquadratic, sparse, irreducible factors of the input. To keep the presentation here as simple as possible, we consider only the multilinear factors.}
Our algorithm follows the following steps:
\paragraph{Step 1 - Normalization and Recursion.}
In this step, we would like to replace our input polynomial $f$ with another polynomial $\tilde{f}$ that still captures all the information about the factors of $f$,
but is normalized (in the sense that its free term is 1 as a polynomial in $x_0$).
To do that, we use a standard method, employed for example in \cite{bisht2025solving,bhargava2020deterministic}.
Let $f=\sum_{i=0}^{\ell d}{f_ix_0^i}$ where $f_i\in \mathbb{F}[\vx]$ and let $\tilde{f}(y,x_1,\ldots,x_n) = f(f_0y,x_1,\ldots,x_n)/f_0$.\footnote{
It is easy to reduce to the case $f_0\neq 0$, so to keep the presentation as simple as possible
we assume it to be the case.}
Note that as we have blackbox access to $f$, we also have blackbox access to $f_0$ and therefore to $\tilde{f}$.
It is straightforward to see that each multilinear factor $f_i(\vx)=a_i(\vx)x_0+b_i(\vx)$ of $f$ now corresponds
to a multilinear factor $\frac{a_if_0}{b_i}y + 1$ of $\tilde{f}$.
Thus, one sees that:
\begin{enumerate}
    \item If $a_i \neq 0$, then all the information about $a_i$ and $b_i$ is captured in $\tilde{f}$.
    \item If $a_i=0$, then the factor $b_i$ cannot  be studied from $\tilde{f}$. However, note that in this case $b_i|f_0$, and therefore it
        can be studied from a factorization of $f_0$.
\end{enumerate}

Note that $f_0$ has the same structure of $f$ - it is itself a product of factors of $s$-sparse polynomials of individual degree $d$. Therefore,
by applying one recursive call, we can find all irreducible multilinear factors of $f_0$.
Thus, to understand the factors with $a_i=0$, we just need to understand, for any multilinear irreducible factor
$\phi$ of $f_0$, how many times it divides $f$.
We will show a way to do so in the next step. For now, the main benefit from this step is the normalization $\tilde{f}$ and the
calculation of the multilinear factors $\phi_i$ of $f_0$.
\paragraph{Step 2 - Obtaining the factorization.}
As we saw in the first step, there are two types of factors - those with $a_i=0$ and those with $a_i\neq 0$.
Thus, we split the algorithm into two parts:
\begin{enumerate}
    \item Recovering the factors for which $a_i = 0$. In this case, it holds that $b_i | f_0$. Therefore, as we already saw, it suffices
    to determine which irreducible factors of $f_0$ divide $f$ (and to what power). To that end, note that for a given
    sparse irreducible factor $\phi$ of $f_0$, it follows from the second property
    of $G$ and from the fact that $f$ is a product of elements of $\cCF(n+1,s,d)$, that to calculate the multiplicity of $\phi$ as a factor of $f$ it suffices
    to check how many times $\phi(G_j)$ divides $f(G_j)$,  where $x_j \in \text{var}(\phi)$. This can be easily done, for example, by factoring both polynomials,
    which is possible as they have a constant number of variables and polynomial degree.
    \item Recovering the factors for which $a_i \neq 0$. In this case, using the first step, we can
    get access to $\frac{a_i(G)f_0(G)}{b_i(G)}$, by factoring $\tilde{f}(y, G)$. Recovering $a_i$ from such an expression is the rational interpolation problem that we earlier alluded to. Here we have the additional information that $b_i$ is a multilinear factor of $f_0$. Thus, to find $b_i$, it suffices to understand,
    for each multilinear irreducible factor $\phi$ of $f_0$, what is the maximal power of $\phi$ that divides $b_i$.
    For any such factor $\phi$, using the fact that $a_i$ and $b_i$ are coprime (as otherwise $a_ix_0+b_i$ will not  be irreducible),  we have
    that $\phi^{k} | b_i$ if and only if $\phi^{t-k+1} \nmid \frac{a_if_0}{b_i}$, where $t$ is the multiplicity of
    $\phi$ as a factor of $f_0$. It follows that we just have to understand the multiplicity of $\phi$ as
    a factor of $\frac{a_if_0}{b_i}$.
    Since $a_i$ is sparse and $f_0$ is a product of polynomials in $\cCF(n,s,d)$, we have that
    $\frac{a_if_0}{b_i}$ is a product of polynomials in $\cCF(n,s,d)$ as well, so we
    can finish by computing the multiplicity of $\phi(G_j)$ as a factor of $\frac{a_i(G_j)f_0(G_j)}{b_i(G_j)}$ (where
    $x_j \in \text{var}(\phi)$). This recovers $b_i$ and subsequently $a_i$ as well, completing the algorithm.
    \end{enumerate}

In the above, we actually solved a rational interpolation problem. That is, we had access to $\frac{af}{b}$, where $b|f$, and needed to recover $a$ and $b$ from this information. We claim that in fact, the algorithm described above works in greater generality.
The exact formulation of this generalization and its proof are given in \Cref{rational interpolation theorem}
and \Cref{thm:multi rational interpolation theorem}.

We conclude with the following remarks. First, we note that it is possible to modify this algorithm to recover all the multiquadratic irreducible factors of $f$, rather than just the multilinear ones. We further note that the algorithm described above still works when the assumption on the field $\mathbb{F}$ is removed, albeit with a slightly worse running time.
These claims are proven in full detail in \Cref{sec:multiquadratic}.

\subsubsection{Finding Sparse Divisors of Sparse Polynomials}

We now describe the ideas behind our main result (\Cref{finding sparse divisors of sparse polynomials}).
Recall that our goal is to output all the sparse divisors of a given $(n,s,d)$-sparse polynomial $f$. That is, we wish to efficiently find those divisors of $f$ that can be ``written down'' efficiently.

To obtain such an algorithm, a natural thing to do is to follow the lines of the algorithm presented in \Cref{multiquadratic outline}.
The first step of this algorithm is performing a normalization and then applying a recursive call for factoring the free term $f_0$.
But now, a problem arises - this recursive call only returns the irreducible sparse factors of $f_0$, while
the free term of each factor of $f$ ($b_i$ in the language of \Cref{multiquadratic outline}) might be a \emph{reducible} sparse divisor of $f_0$ whose irreducible
factors are not sparse; in this case, the recursive call did not provide us with any valuable information about $b_i$.
In particular, we cannot  reconstruct $b_i$ by studying how many times each sparse irreducible factor of $f_0$ divides it as we did in \Cref{multiquadratic outline}. In fact,
the true reason the algorithm presented in \Cref{multiquadratic outline} works is that in that context, $b_i$ is sparse multilinear (or sparse multiquadratic in the more general case),
and therefore so are its irreducible factors; hence, we indeed get information about $b_i$ from the recursive call. In other words, the free terms there belong to a class of sparse polynomials that is closed under factorization.

This motivates considering the problem of finding all the \emph{sparse divisors} of a sparse polynomial. This way, we will  still
have information about the free terms of the factors of $f$ - the recursive call will actually compute them.
Indeed, as sparse divisors of sparse divisors of $(n,s,d)$-sparse polynomials are themselves sparse divisors of $(n,s,d)$-sparse polynomials, so this class is closed under factorization. This enables us to recover all these sparse divisors in
deterministic $\poly(n,s^{\poly(d)})$-time. 
Indeed, the recursive call returns all the sparse divisors of the free term $f_0$ rather than just the set of irreducible factors of it, so we can
directly brute-force over the list of these factors to find $b_i$. This yields \Cref{finding sparse divisors of sparse polynomials}.

At first glance, this approach may seem surprising, as the existence of such an algorithm suggests that the
possible number of outputs is bounded above by $\poly(n,s^{\poly(d)})$. Indeed,  \Cref{thm:bound on the number of irreducible factors of sparse polynomials} ensures that $s$-sparse polynomials
of bounded individual degree $d$ have no more than $d\log s$ irreducible factors (that are not monomials), and therefore, no more than $s^d$ divisors
(not divisible by a monomial) overall.
To prove \Cref{thm:bound on the number of irreducible factors of sparse polynomials}, we show that if $f\in \cP(n,s,d)$ is not divisible by any monomial,
then one can identify a set of $\log s$ variables that satisfy the property that every irreducible factor of $f$
depends on at least one of them. From this claim it easily follows that $f$ has at most $d\log s$ irreducible factors, and therefore
at most $2^{d\log s}=s^d$ divisors overall. To prove that such a set of variables exist, we analyze the Newton polytope of $f$, and
show that if the minimal set of variables satisfying the property is of size $k$, then the Newton polytope of $f$ has at least $2^k$ vertices. This is somehow intuitive, as each additional variable in the set effectively adds a dimension to the polytope and therefore doubles
the number of its vertices. 

We note that \Cref{finding sparse divisors of sparse polynomials} is not implied by the methods of \cite{bhargava2020deterministic}; indeed, 
there the bound on the sparsity of \emph{all irreducible factors} of $f$ is needed to hit all resultants of different factors; therefore,
applying their methods to this problem must result in a running time of at least $s^{O(d^2\log n)}$, which is their bound on the sparsity of the factors of $f$.

One main drawback of our algorithm is that it does not  tell us what are the \emph{irreducible} sparse divisors of the
input polynomial, as we are  not familiar with a way of efficiently identifying the irreducible polynomials from the outputs we get. On the other hand, if we are guaranteed that $f$ is a product of irreducible sparse factors then we can easily find the irreducible factors among all the sparse divisors. This establishes \Cref{cor:factoring-product-sparse}, thereby solving \Cref{q:factoring-sparse-irred-product} for the special case of $(n,s,d)$-sparse polynomials.

These ideas are also used in the proof of \Cref{thm:factor-of-product-of-nsd}.
Recall that \Cref{q:dst-factor-product-nsd} asks for a deterministic, polynomial-time algorithm that given
blackbox access to a product of $(n,s,d)$-sparse polynomials returns the factors.
While we do not know  a polynomial-time algorithm when the number of terms is unbounded, we show an algorithm that runs in quasi-polynomial time in their number, making the first
step towards a resolution of the problem. 
The algorithm
is the natural, straightforward generalization of the algorithm described above. Unfortunately, it requires
quasi-polynomial time as we actually find all the $s$-sparse divisors of bounded individual degree $d$
of $f$, and there are (possibly) quasi-polynomially many of them (as shown by \Cref{general bound on the number of sparse divisors}). 

\subsubsection{Factorization for Sparse Polynomials}\label{sec:BSV overview}

Lastly, we discuss our algorithms for factoring $(n,s,d)$-sparse polynomials and their products. These algorithms are essentially the same as the algorithms we already presented; the only difference is that now, when dealing with elements of $\cCF(n,s,d)$, we view them as elements of $\cP(n,s^{O(d^2\log n)},d)$. 

Consider for example our algorithm for finding all the sparse divisors of an $(n,s,d)$-sparse polynomial. With a slight modification, we can make it output all the $(n,s^{O(d^2\log n)}, d)$-sparse divisors of the input polynomial in $\poly(n,s^{d^2\log n})$-time. By \cite{bhargava2020deterministic}, these are all the possible divisors. Identifying the irreducible factors among these can be easily done. For example, we can check for any two divisors $\phi$ and $\psi$ if $\psi|\phi$ by trying to sparse-interpolate their quotient - which should be an $s^{O(d^2\log n)}$-sparse polynomial. 

As for our algorithm for factoring products of $\cCF(n,s,d)$ elements, it is identical to the algorithm described in \Cref{multiquadratic outline}, now taking into account that all irreducible factors of interest are $s^{O(d^2\log n)}$-sparse, rather than $s$-sparse and multiquadratic. 

\subsection{Organization}
In \Cref{sec:prelim} we give all the background needed for this paper. In \Cref{Section 3} we construct the generator $G$, which
is the main technical tool of this paper. 
In \Cref{Section 4} we prove the basic properties of $G$ we need, both for arbitrary fields and when the characteristic is zero or large enough, and apply them to obtain our divisibility testing and exact power testing results (\Cref{general divisibility testing}, \Cref{general complete power testing}). In \Cref{Section 5} we prove our rational interpolation result.

We give a \emph{meta-algorithm}, according to which all of our algorithms are modeled, in \Cref{sec:meta}. We then demonstrate how to apply it to recover the multiquadratic sparse irreducible factors of a product of polynomials in $\cCF(n,s,d)$ (\Cref{finding multiquadratic factors of a product of sparse polynomials}) in \Cref{sec:multiquadratic}.
In \Cref{sec:sparse-div-of-sparse} we bound the number of divisors of $(n,s,d)$-sparse polynomials (\Cref{thm:bound on the number of irreducible factors of sparse polynomials}), and use it to 
find the all the sparse divisors of such polynomials (\Cref{finding sparse divisors of sparse polynomials}), as well as to give one partial solution to \Cref{q:dst-factor-product-nsd} (\Cref{thm:factor-of-product-of-nsd}).
In \Cref{sec:BSV based results} we give our algorithms for factoring $(n,s,d)$-sparse polynomials and products of them, based on the bound of \cite{bhargava2020deterministic} (\Cref{improved BSV},\Cref{improved BSV for products}). The second result gives another partial solution to \Cref{q:dst-factor-product-nsd}.
We conclude with a short discussion in \Cref{Section:open}.
In \Cref{appendixA} we give different, yet weaker, proofs for some technical claims from \Cref{Section 4}. Since the proof technique is different, we decided to give the proof in spite of the results being slightly weaker.

\section{Preliminaries}\label{sec:prelim}

We use $\vx$ to denote the $n$-variate vector of variables, $\vx=(x_1,\ldots,x_n)$. In general, bold-face letters such as $\vx,\ve$ will be used to denote $n$-tuples, where the domain of the $n$-tuple will always be clear from the context.

We say that two nonzero polynomials $f,g$ are associate  if one is a scalar product of the other. We also state the following special case of Gauss' lemma.

\begin{lemma}[Gauss's Lemma]\label{lem:gauss}
A 
polynomial $f(x_0,\vx) = \sum_{i=0}^{d} c_i(\vx) x_0^i \in \F[\vx][x_0]\setminus \F[\vx]$ is irreducible if and only if it is both
irreducible in $\F(\vx)[x_0]$ and $\gcd(c_0, \dots, c_d)=1$.
\end{lemma}

\subsection{Reverse Monic Polynomials and Normalization}

Throughout the paper we will use the following definition.
\begin{definition}\label{def:reverse-monic}
    A polynomial $f\in \mathbb{F}[x_0,\vx]$ is said to be reversed-monic with respect to $x_0$ if its free term
    as a polynomial in $\mathbb{F}(\vx)[x_0]$ is 1. 
    Equivalently, $f$ is reversed-monic with respect to $x_0$ if it is of the form $f=1+\sum_{i\ge1}{f_ix_0^i}$ where $f_i\in \mathbb{F}[\vx]$. 
\end{definition}

We prove the following simple claim about reversed-monic polynomials.
\begin{claim}\label{cla:reversed-monic bound}
    Let $f\in \mathbb{F}[x_0,\vx]$ be a reversed-monic polynomial with respect to $x_0$. Then, the number of irreducible factors of $f$ is at most
    $\deg_{x_0}(f)$.
\end{claim}

\begin{proof}
    We first note that all irreducible factors of $f$ must depend on $x_0$. Indeed, write $f=\sum_{i=0}^{\deg_{x_0}(f)}{f_ix_0^i}$ where $f_i\in \mathbb{F}[\vx]$. By definition of reversed-monic polynomials, $f_0=1$. Clearly, any divisor of $f$ that does not depend on $x_0$ must divide $f_0$, and hence no such non-trivial divisors exist. 
    The claim follows, since there are at most $\deg_{x_0}(f)$ irreducible factors of $f$ depending on $x_0$. 
\end{proof}

We next define a \emph{normalization} operation that turns a polynomial into a reverse monic one.

\begin{definition}[Normalization]\label{def:normalization}
Let $f(x_0,\vx)=\sum_{i=0}^df_i(\vx)x_0^i\in\F[x_0,\vx]$, such that $x_0\nmid f$. The normalization of $f$ with respect to $x_0$ is the reverse monic polynomial
\begin{equation}\label{definition of normalization}
  \tilde{f}(y,\vx) = f(yf_0,x_1,\ldots,x_n) / f_0=1+\sum_{i=1}^{d}f_i(\vx)f_0^{i-1}(\vx)y^i.
\end{equation}
\end{definition}

Normalization is a standard and well-known technique, and different versions of our next claim can be also found in \cite{bhargava2020deterministic,bisht2025solving}. Our version of the claim describes the structure of the factorization of $\tilde{f}$, which will play an important role in our proofs and reconstruction algorithms.

\begin{claim}\label{cla:normalization of a polynomial}
    Let $f$ as in \Cref{def:normalization}.
    Suppose $f = \prod_{i=0}^{N}{h_i}\in \mathbb{F}[x_0,\vx]$ is the factorization of $f$ into irreducible factors $h_i$,
    and denote each factor by $h_i=\sum_{j=0}^{d_i}{c_{i,j}{x_0}^{j}}$ (where $c_{i,j} \in \mathbb{F}[\vx]$).
    Then, the factorization of $\tilde{f}$ is given by 
    \[\tilde{f} = \prod_{i=0}^{N}{\tilde{h}_i}\ , \textrm{where}\quad \tilde{h}_i = \sum_{j=0}^{d_i}{\frac{c_{i,j}}{c_{i,0}}{f_0}^{j}y^{j}} \ .\]
    Moreover, every factor (not necessarily irreducible) of $f$ of the form $\sum_{j=0}^{d}{\zeta_{j}{x_0}^{j}}$ (where $\zeta_{j} \in \mathbb{F}[\vx]$)
    corresponds to a factor of $\tilde{f}$ of the form $\sum_{j=0}^{d}{\frac{\zeta_{j}}{\zeta_{0}}{f_0}^{j}y^{j}}$.
    In particular, every multiquadratic factor of $f$ of the form $a_2x_0^2+a_1x_0+b$ corresponds
    to a factor of $\tilde{f}$ of the form $1+\frac{a_1f_0}{b}y+\frac{a_2f_0^2}{b}y^2$. In addition, we have blackbox access to $\tilde{f}$ from blackbox access to $f$.
\end{claim}

\begin{proof}
    First, note that $f_0 = \prod_{i=0}^{N}{c_{i,0}}$. Next, let
    \begin{equation*}
        \tilde{f}(y,\vx)=f(yf_0,\vx) /f_0 = \prod_{i=0}^{N}{\frac{h_i(yf_0,\vx)}{c_{i,0}}} \ .
    \end{equation*}
    Define
    \begin{equation*}
    \tilde{h}_i(y,\vx) := \frac{h_i(yf_0,\vx)}{c_{i,0}} = \sum_{j=0}^{d_i}{\frac{c_{i,j}{(yf_0)}^{j}}{c_{i,0}}} = 
    \sum_{j=0}^{d_i}{\frac{c_{i,j}}{c_{i,0}}{f_0}^{j}{y}^{j}} \in \F[y,\vx] \ .    
    \end{equation*}

    We claim that $\tilde{h}_i$ is irreducible. If $d_i=0$, then $\tilde{h}_i=1$, so this is clear.
    Assume $d_i>0$. By definition, ${h_i}$ is an irreducible polynomial as an element of $\mathbb{F}[x_0,\vx]$,
    and therefore it is also irreducible as an element of $\mathbb{F}(\vx)[x_0]$. In this ring it holds that
    $\tilde{h}_i(x_0) = h_i(f_0x_0)/c_{i,0}$. In particular, $\tilde{h}_i$ differs from $h_i$ by a linear change of
    variables and a multiplication by scalar (in the field $\F(\vx)$), so it is irreducible as well. Since the free term of $\tilde{h}_i$ is $1$, it follows
    from Gauss' lemma (\Cref{lem:gauss}) that $\tilde{h}_i$ is irreducible as an element of $\mathbb{F}[x_0,\vx]$, as we wanted.

   The "moreover" part follows from the fact that the argument on the form of $\tilde{h}_i$ holds regardless of the irreducibility of $h_i$.

   The claim regarding blackbox access to $\tilde{f}$ follows from simple interpolation: given an evaluation point $(y_0, \valpha)$ we interpolate the polynomial $f(x_0,\valpha)$ as a polynomial in $x_0$ to get access to each $f_i(\valpha)$. One can then obtain $\tilde{f}(y_0,\valpha)$ simply by using the right hand side of \eqref{definition of normalization}.
   In particular, every query to $\tilde{f}$ requires at most $d+1$ queries to $f$. 
\end{proof}

\subsection{Factorization and Interpolation of Polynomials}\label{factorization background}

We start by stating a result on factorization of univariate polynomials. 
The first is the classical algorithm of Lenstra, Lenstra and Lovasz \cite{LLL}.

\begin{theorem}[Univariate factorization over $\Q$]\label{thm:lll}
    Let $f\in \Q[x]$ have degree $d$ and bit complexity $B$. Then, there is a deterministic algorithm that outputs all irreducible factors of $f$ whose running time is $\poly(d,B)$.
\end{theorem}

For the following theorem see e.g., \cite[Section 9]{GathenShoup1992-deterministic-factor}.

\begin{theorem}[Univariate factorization over Finite Fields]\label{thm:finite-factor}
    Let $\F_q$ be a field of size $q=p^k$ and characteristic $p>0$.
    Let $f(x)\in \F_q[x]$ be a polynomial of degree $d$. Then, there is a deterministic algorithm that outputs all irreducible factors of $f$ in time $\poly(p,\log q,d)$.
\end{theorem}

We next state a theorem concerning deterministic multivariate factorization in the multivariate case. See e.g., {\cite{Lenstra85det-factorization,Kaltofen85reduction}}. 
These results determine the dependence on the field, its characteristic and the bit complexity of the input polynomial, in all our algorithms
that invoke polynomial factorization.

\begin{theorem}
    \label{thm:polynomial factorization}
    There exists a deterministic algorithm that given an $n$-variate polynomial $f\in\F[\vx]$ of total degree $d$, and bit complexity $B$, outputs its factorization
    into irreducible polynomials. Its running time is $\poly(d^n)\cdot T(d,B,\F)$, where $T(d,B,\F)$ is  the running time of univariate factorization of degree $d$ polynomials, with bit complexity $B$, over $\F$.
\end{theorem}

Next, we state the following folklore lemma on interpolation of polynomials on a small numbers of variables.
\begin{lemma}\label{Interpolation of polynomials}
    Let $f\in \mathbb{F}[x_1,\ldots,x_n]$ be an $n$-variate polynomial of total degree $d$. Then, there is a
    deterministic polynomial-time algorithm that given blackbox access to $f$ returns its monomial representation.
    Its running time is $\poly((d+1)^n)$.
\end{lemma}

Finally, we conclude this section by recalling the definition of Lagrange interpolation.
\begin{definition}\label{lagrange interpolation}
    Let $\alpha_1,\ldots,\alpha_r\in \mathbb{F}$ be different field elements. Then the polynomials  \begin{equation}\label{explicit lagrange interplation}A_i(x)=\prod_{1\leq j\leq r, j\neq i}{\frac{x-\alpha_j}{\alpha_i-\alpha_j}}
    \end{equation}
    satisfy $\deg(A_i)=r-1$ and $A_i(\alpha_j)=\delta_{i,j}$, where $\delta_{i,j}$ is Kronecker's delta. The $A_i$-s are called the
    Lagrange interpolating polynomials of $\alpha_1,\ldots,\alpha_r$.
\end{definition}

\subsection{Sparse Polynomials}
We give some definitions and state several known results concerning sparse polynomials.

\begin{definition}\label{sparse polynomial definition}
    An $s$-sparse polynomial $f\in \mathbb{F}[x_1,\ldots,x_n]$ is a polynomial which is a sum of at most $s$ monomials.
    It is said to be of bounded individual degree $d$ if $f$ is of degree at most $d$ as an element of
    $\mathbb{F}(x_1,\ldots,x_{i-1},x_{i+1},\ldots,x_n)[x_i]$ for every $1\leq i \leq n$.
    We denote the class of $n$-variate, $s$-sparse polynomials of bounded individual degree $d$ by $\cP(n,s,d)$, and refer to its elements as $(n,s,d)$-sparse polynomials.
    Furthermore, we denote the class of factors of elements of $\cP(n,s,d)$ by $\cCF(n,s,d)$. Note that by factors we mean all possible factors of $(n,s,d)$-sparse polynomials, not just the irreducible ones. 
\end{definition}

The next two lemmas give two special cases in which upper bounds on sparsity of factors of sparse polynomials is known.

\begin{lemma}[{\cite[Lemma 10]{GuptaKL12}}]\label{multilinear factors of sparse polynomials are sparse}
    A multilinear factor of an $s$-sparse polynomial is again $s$-sparse.
\end{lemma}

\begin{lemma}[{\cite[Lemma 32]{volkovich2017some}}]\label{factors of multiquadratic sparse polynomials are sparse}
    Let $g$ be an $s$-sparse multiquadratic polynomial (that is, $g$ is of individual degree at most two) and suppose that $f|g$. Then
    $f$ is also $s$-sparse.
\end{lemma}

In 
\cite{bhargava2020deterministic}, Bhargava, Saraf and Volkovich upper bounded the sparsity of factors of $(n,s,d)$-sparse polynomials, and obtained the best known factorization algorithm for $(n,s,d)$-sparse polynomials. We state these results next.

\begin{theorem}[{\cite[Theorem 1]{bhargava2020deterministic}}]\label{thm: BSV bound}
    $\cCF(n,s,d)\subset \cP(n,s^{O(d^2\log n)},d)$. In other words, every factor of an $s$-sparse polynomial of bounded individual degree $d$ in $n$-variables is $s^{O(d^2\log n)}$-sparse.
\end{theorem}

In the next result, one should add a factor that depends on the field according to \Cref{thm:polynomial factorization}.

\begin{theorem}[{\cite[Theorem 2]{bhargava2020deterministic}}]\label{thm: BSV factorization algorithm}
    There is a deterministic $\poly(n,d^d,s^{d^7\log n})$-time algorithm, that given an $(n,s,d)$-sparse polynomial, returns its factorization into  irreducible factors.
\end{theorem}

To obtain \Cref{thm: BSV bound}, the authors study sparse polynomials via the study of their Newton polytopes.
As we will use the language of Newton polytopes as well, we give the definition of Newton polytopes here.

\begin{definition}\label{Newton polytope}
    The Newton polytope of a polynomial $f=\sum_{e}{\gamma_ex^e}$, denoted by $P_f$, is the convex hull of the points $e\in \mathbb{R}^n$.
    The set of vertices of the polytope $P_f$ is denoted by $V(P_f)$.
\end{definition}

The following well known fact can be found in \cite[Corollary 3.5]{bhargava2020deterministic}.

\begin{lemma}\label{number of vertices is monotone}
    Let $f|g$ be two polynomials, and suppose that $g$ is $s$-sparse. Then, $|V(f)|\leq |V(g)|\leq s$.
\end{lemma}

\subsection{Resultants and Discriminants}
We next discuss resultants and discriminants. See  \cite{GKZ-book} for a  detailed treatment.

\begin{definition} \label{Sylvester matrix and resultant}
    Let $f, g\in \mathbb{F}[x]$ be two univariate polynomials of degrees $d_1 ,d_2$ respectively.
    Define $\phi: \mathbb{F}[x]_{\leq (d_1-1)}\times \mathbb{F}[x]_{\leq (d_2-1)} \rightarrow \mathbb{F}[x]_{\leq (d_1+d_2-1)}$ by
    \begin{equation} \label{Sylvester map}
        \phi(h_1, h_2) = h_1g+h_2f.
    \end{equation}
    Here, $\mathbb{F}[x]_{\leq (d)}$ is the $(d+1)$-dimensional $\mathbb{F}$-vector space of univariate polynomials over $\mathbb{F}$
    of degree at most $d$. This space has a standard monomial basis given by the monomials $1, x, \ldots.,x^d$.
    
    The Sylvester matrix $M(f,g)$ of $f$ and $g$ is the $(d_1+d_2)\times (d_1+d_2)$ matrix representing $\phi$ with respect to the standard
    monomial basis. The determinant of this matrix is called the resultant of $f$ and $g$, and is denoted by $\res (f,g)\in \mathbb{F}$.
    If $f,g\in \mathbb{F}[x_0,x_1,\ldots,x_n]$ are two multivariate polynomials,
    we denote by $M_{x_0}(f,g)$ and $\res_{x_0}(f,g)\in \mathbb{F}[x_1,\ldots,x_n]$ the Sylvester matrix and
    resultant of $f$ and $g$ when viewed as elements of $\mathbb{F}(x_1,\ldots,x_n)[x_0]$.
\end{definition}

\begin{definition} \label{discriminant}
    Let $f\in \mathbb{F}[x]$ be a univariate polynomial. The discriminant of $f$ is given by
    \begin{equation}
        \Delta(f)=\res(f,f'),
    \end{equation}
    where $f'$ is the standard formal derivative of $f$.
    For a multivariate polynomial $f\in \mathbb{F}[\vx]$, denote by  $\Delta_{x_i}(f)\in \mathbb{F}[x_1,\ldots,x_{i-1},x_{i+1},\ldots, x_n]$ the discriminant of $f$ as an element
    of $\mathbb{F}(x_1,\ldots,x_{i-1},x_{i+1},\ldots,x_n)[x_i]$.
\end{definition}

\begin{lemma}\label{basic resultant properties}
    Let $\mathbb{F}$ be a field.
    \begin{enumerate}
        \item Let $f, g\in \mathbb{F}[x]$ be two monic polynomials. Then, $\res(f,g)=\prod_{\alpha: f(\alpha)=0}{g(\alpha)}$,
        where the product is over all the roots $\alpha$ of $f$ over $\overline{\mathbb{F}}$.
        \item $f, g \in \mathbb{F}[x]$ have a common factor in $\mathbb{F}[x]$ if and only if their resultant vanishes.
        \item If $\text{char}(\mathbb{F})=0$, a polynomial $f$ is square free if and only if $\Delta(f) \neq 0$. Moreover,
        the last statement is also true if $\text{deg}(f)<\text{char}(\mathbb{F})$.
        \item \label{item:dim=deg} Let $f,g \in \mathbb{F}[x_0,\vx]$ be two $s$-sparse polynomials of bounded individual degree $d$. Then,
        their Sylvester matrix $M_{x_0}(f,g)$ is a $2d\times 2d$ matrix whose entries are $s$-sparse elements of $\mathbb{F}[\vx]$,
        and moreover, $\dim(\ker {M_{x_0}(f,g)})=\deg_{x_0}(\gcd(f,g))$. In particular,
        \begin{enumerate}
            \item If $f, g \in \mathbb{F}[x_0,\vx]$ are two $s$-sparse polynomials of bounded individual degree $d$, 
            then $\res_{x_0}(f, g) \in \mathbb{F}[\vx]$ is a $(2d)!s^{2d}$-sparse polynomial of bounded individual
            degree $2d^2$.
            \item If $f\in \mathbb{F}[x_0,\vx]$ is an $s$-sparse polynomial of bounded individual degree $d$,
            then $\Delta_{x_0}{f} \in \mathbb{F}[\vx]$ is a $(2d)!s^{2d}$-sparse polynomial of bounded individual
            degree $2d^2$. 
        \end{enumerate}
    \end{enumerate}
\end{lemma}
    
\begin{proof}
    The first three statements are standard. The last statement is also well known, and its explanation is given for example in \cite{volkovich2017some},
    in the language of subresultants. We give the full proof of the "moreover" part here as well,
    as it plays an important role in the proof of \Cref{lem:G always preserves coprimality of factors}.
    
    Suppose $f, g \in \mathbb{F}[x_0,\vx]$ are two $s$-sparse polynomials of bounded individual degree $d$. Think of both
    as elements of $\mathbb{F}(\vx)[x_0]$, and write $f=\prod_{i}{h_i^{a_i}}, g=\prod_{i}{h_i^{b_i}}$, where $h_i$ are
    all the irreducible factors of $fg$, and $a_i,b_i\ge 0$ (for example, if $f, g$ are coprime, then
    $a_ib_i=0$ for every $i$). Note that $\gcd (f,g)=\prod_{i}{h_i^{\min {(a_i, b_i)}}}$ and $\operatorname{lcm} (f,g)=\prod_{i}{h_i^{\max {(a_i, b_i)}}}$. 
    
    Consider an element $(w_1, w_2)\in \ker {M_{x_0}(f,g)}$. Then, by definition,
    \begin{equation}\label{being in the kernel of sylvester}
        w_1g=-w_2f.
    \end{equation}
    Therefore, by considering the factorization into  irreducible factors of both sides,
    it must be the case that $\prod_i{h_i^{\max(a_i-b_i,0)}}|w_1$ and $\prod_i{h_i^{\max(b_i-a_i,0)}}|w_2$.
    Thus, we can write 
    \begin{equation}\label{gcd decomposition}
        w_1=\prod_i{h_i^{\max(a_i-b_i,0)}\tilde{w}_1}, \quad w_2=\prod_i{h_i^{\max(b_i-a_i,0)}\tilde{w}_2} \ .
    \end{equation}
    It follows that we have
    \begin{equation}
    \begin{aligned}\label{connection to lcm}
        w_1g=\prod_i{h_i^{b_i+\min(a_i-b_i, 0)}}\tilde{w}_1=\prod_i{h_i^{\max(a_i,b_i)}}\tilde{w}_1=\operatorname{lcm} (f,g)\tilde{w}_1 \\
        w_2f=\prod_i{h_i^{a_i+\min(b_i-a_i, 0)}}\tilde{w}_2=\prod_i{h_i^{\max(a_i,b_i)}}\tilde{w}_2=\operatorname{lcm} (f,g)\tilde{w}_2  \ .
    \end{aligned}
    \end{equation}
    Combining \eqref{being in the kernel of sylvester} and \eqref{connection to lcm}, we get that $\tilde{w}_1 =-\tilde{w}_2$. Since $\deg w_1,\deg w_2\leq d-1$, we get from \eqref{gcd decomposition} that an element in $\ker {M_{x_0}(f,g)}$ corresponds to a polynomial $\tilde{w}$ of degree at most
    \begin{equation}
    \begin{aligned}
        d-1-\sum_i{\max(a_i-b_i,0)\deg{h_i}} &= \sum_i{(a_i-\max(a_i-b_i,0))\deg(h_i)}-1 \\
        &=\sum_i{\min(a_i,b_i)\deg(h_i)}-1=\deg_{x_0}(\gcd (f,g))-1 \ .
    \end{aligned}
    \end{equation}
    Therefore, $\ker {M_{x_0}(f,g)}$ is isomorphic to $\mathbb{F}[x]_{\leq (\deg_{x_0} \gcd(f,g)-1)}$, and hence
    is of dimension $\deg_{x_0} \gcd(f,g)$, as claimed.
\end{proof}

\section{Generator Construction}\label{Section 3}

In this section, we construct our generator, which is a map $G: \mathbb{F}^{6} \rightarrow \mathbb{F}^n$ that satisfies the properties described in \Cref{sec:methods outline}.

In what follows we let $m$ be some parameter (in our setting we will usually set it to $m=(nsd)^2$ when dealing when $(n,s,d)$-sparse polynomials). To define our generator $G_{(m)}$ we will use the following notational conventions. Let  $q= \Theta(m)$ be a prime (for concreteness $q$ can be the first prime larger than $m$). 
We shall also assume that $|\F|>ndq$. If this is not the case then we work over an extension field $\F\subset\mathbb{E}$  of size $dnq<|\mathbb{E}|\leq (dnq)^2$.

 Fix the following sets:
\begin{enumerate}
    \item A set of integers $\cI = \{1,\ldots,m\}$ of size $m$, that contains $m$ different residues mod $q$. 
    \item A set of field elements $\cA = \{\alpha_1,\ldots,\alpha_m\} \subset \mathbb{F}$ of size $m$. For every $1 \leq i \leq m$ let
    $A_i$ be the Lagrange interpolating polynomials for $\cA$.
    \item A set of field elements $\cB = \{\beta_1,\ldots,\beta_n\} \subset \mathbb{F}$ of size $n$. For every $1 \leq i \leq n$ let
    $B_i$ be the Lagrange interpolating polynomials for $\cB$.
    \item A set of field elements $\cC = \{\gamma_1,\ldots,\gamma_n\} \subset \mathbb{F}$ of size $n$. For every $1 \leq i \leq n$ let
    $C_i$ be the Lagrange interpolating polynomials for $\cC$.
    \item A set of field elements ${\cT} = \{\delta_1,\ldots,\delta_{ndq}\} \subset \mathbb{F}$ of size $ndq$. 
\end{enumerate}

To define our generator we first define the generator of Klivans and Spielman and state some of its properties.

\begin{definition}[The KS-generator]\label{def:KS}
In the notations above, the KS-generator is defined as  
\begin{equation}
    {G}_{(m)}^{\text{KS}}(x,y,z,w) = \sum_{i=1}^{m} {A_i(y)\left((1 + B_1(z)(w-1))x^{i\bmod q},\ldots,(1 + B_n(z)(w-1))x^{i^n \bmod q}\right)}\ .
\end{equation}    
\end{definition}

\begin{theorem}[{\cite[Theorem 11]{KS01}}]\label{thm:KS}
Let $f\in\F[\vx]$ be an $(n,s,d)$-sparse polynomial. Let $m=(nds)^2$.
Then, $f$ can be reconstructed in time $\poly(n,s,d)$ from the following set of evaluations:
\[
\left\{ f\left({G}_{(m)}^{\text{KS}}(x,y,z,w) \right) \mid   \; y\in \cA,\;  z\in \cB, \; x,w\in \cT\; \right\}.
\]
\end{theorem}
\begin{proof}
    Observe that when setting $y=a_i$ and $z=b_j$ we obtain
\[    {G}_{(m)}^{\text{KS}}(x,a_i,b_j,w) = \left(x^{i\bmod q},\ldots,x^{i^{j-1}\bmod q},wx^{i^j \bmod q},x^{i^{j+1}\bmod q},\ldots,x^{i^{n}\bmod q}\right)\ .\]
The claim follows from the analysis in \cite[Theorem 11]{KS01}.
\end{proof}

Another tool that we will use is the generator defined in \cite{shpilka2009improved}, which is given by the formula
\begin{equation}
    G^{\text{SV}}(u,v) = \left(C_1(v)u,\ldots,C_n(v)u\right) \ . 
\end{equation}
The basic property of this generator is that it makes it possible to "revive" a variable of our choice. E.g., by setting $v=\gamma_j$ we get the vector that has zeros in all coordinates except of the $j$-th coordinate where we have the variable $u$.

We are now ready to define our generator.

\begin{definition}[The generator $G_{(m)}$]\label{def:Gm}
    The generator $G_{(m)}:\mathbb{F}^{6} \rightarrow \mathbb{F}^{n}$ is defined as 
    \begin{equation}
    G_{(m)}(x,y,z,w,u,v) = {G}_{(m)}^{\text{KS}}(x,y,z,w) + G^{\text{SV}}(u,v)\ .
\end{equation}
    We further define the generators obtained from reviving a variable. Denote with ${G}_{(m)}^{\text{KS}}(x,y,z,w)_i$ the $i$-th coordinate of ${G}_{(m)}^{\text{KS}}(x,y,z,w)$.
    
    For any $1\leq i \leq n$, let $G_{(m,i)}$ be the generator obtained from $G_{(m)}$
    by restricting to $u=u-{G}_{(m)}^{\text{KS}}(x,y,z,w)_{i}$ and $v=\gamma_i$.
    That is,
    \begin{equation}\label{eq:reviving}
        G_{(m,i)}(x,y,z,w,u) = \left({G}_{(m)}^{\text{KS}}(x,y,z,w)_{1},\ldots,u,\ldots,{G}_{(m)}^{\text{KS}}(x,y,z,w)_{n}\right)\ .
    \end{equation}
\end{definition}

Let's explain this definition.
The first part ${G}_{(m)}^{\text{KS}}$ enables us to interpolate $(n,s,d)$-sparse polynomials (\Cref{thm:KS}). 
The second part allows us to “revive” a variable of our choice, a feature that will play an important role in the analysis.

The following lemma shows that our generator indeed satisfies the first property described in \Cref{sec:methods outline}.

\begin{lemma}\label{lem:first property of G}
    Let $f$ be an $(n,s,d)$-sparse polynomial and $m=(nds)^2$.
    Then  one can recover $f$ from $f(G_{(m)})$ in time $\poly(m)$.
    In particular, it holds that $f=0$ if and only if $f(G_{(m)})=0$.
    In fact, for every $1\leq i\leq n$, the statement is also true when replacing $G_{(m)}$ with $G_{(m,i)}$ and with ${G}_{(m)}^{\text{KS}}$.
\end{lemma}

\begin{proof}
    For  every $1\leq i\leq n$, the image of each one of $G_{(m)}$, $G_{(m,i)}$ and $G^{\text{KS}}_{(m)}$ includes image of $G^{\text{KS}}_{(m)}$, so the claim follows from \Cref{thm:KS}.
\end{proof}

We conclude this section with the following three corollaries of the previous lemma.

\begin{corollary}\label{G hits factors of sparse polynomials}
    Suppose $f\in \cCF(n,s,d)$.
    Then $f(G_{(m)}) \neq 0$ for $m=(nds)^2$.
    In fact, for every $1\leq i\leq n$, the statement is also true when replacing $G_{(m)}$ with $G_{(m,i)}$ and with ${G}_{(m)}^{\text{KS}}$.
\end{corollary}

\begin{proof}
    Let $g\in \cP(n,s,d)$ so that $f|g$, and let $G$ be any one of $G_{(m)}$, $G_{(m,i)}$ and ${G}_{(m)}^{\text{KS}}$.
    \Cref{lem:first property of G} implies that $g(G)\neq 0$. The claim follows from
    the fact that $f(G)|g(G)$.
\end{proof}

\begin{corollary}\label{G preserves degree}
    Let $f\in \cCF(n,s,d)$.
    Then, $\deg_{u}(f(G_{(m,i)}))=\deg_{x_i}(f)$ for every $1\leq i\leq n$ and $m=(nsd)^2$.
\end{corollary}

\begin{proof}
    Consider the leading coefficient\footnote{The leading coefficient of a polynomial in a variable $x$ is the coefficient of the highest power of $x$ appearing nontrivially in the polynomial.} 
$\operatorname{LC}_{x_i}(f)$ of $f$ as a polynomial in $x_i$, and let $g \in \cP(n,s,d)$ be any polynomial such that $f \mid g$.
Then the leading coefficient $\operatorname{LC}_{x_i}(g)$ of $g$ (viewed as a polynomial in $x_i$) is divisible by $\operatorname{LC}_{x_i}(f)$.
Since $g$ is $s$-sparse, its leading coefficient is also $s$-sparse.
Therefore, $\operatorname{LC}_{x_i}(f)$ is a factor of an $(n-1,s,d)$-sparse polynomial, and hence, by
\Cref{G hits factors of sparse polynomials}, we have
$\operatorname{LC}_{x_i}(f)(G^{\text{KS}}_{(m)}) \neq 0$ when substituting all variables except $x_i$.

By~\eqref{eq:reviving}, $\operatorname{LC}_{x_i}(f)(G^{\text{KS}}_{(m)})$ is precisely the leading coefficient of
$f(G_{(m,i)})$ as a polynomial in $u$.
Therefore, we must have $\deg_u f(G_{(m,i)}) = \deg_{x_i}(f)$.
\end{proof}

\begin{corollary}\label{reduction for nonzero constant term}
    There is a deterministic $\poly(n,d,s,\ell)$-time algorithm that, given blackbox access to a product $f$ of $\ell$ polynomials in $\cCF(n+1,s,d)\subset\mathbb{F}[x_0,\vx]$, returns the multiplicity $k$ of $x_0$ as a factor of $f$.
    Moreover, if $f={x_0}^{k}g$, one can obtain blackbox access to $g$ from blackbox access
    to $f$. 
\end{corollary}

\begin{proof}
    Let $f={x_0}^{k}\prod_{i}{f_i}$ be the factorization of $f$ into irreducible factors, where $f_i \neq x_0$. For each $i$, let $g_i$ be an $(n+1,d,s)$-sparse polynomial such that $f_i|g_i$. Let $m=((n+1)ds)^2$.
    Consider the polynomial 
    \[f(x_0, G_{(m)})={x_0}^{k}\prod_{i}{f_i(G_{(m)})}\ .\]
    Fix some $i$, and denote $f_i = \sum_{j=0}^{d}{f_{i,j}{x_0}^{j}}$ and $g_i = \sum_{j=0}^{d}{g_{i,j}{x_0}^{j}}$, where $f_{i,j},g_{i,j} \in \mathbb{F}[\vx]$.
    As $x_0 \nmid f_i$, it follows that $f_{i,0}\neq 0$ for every $i$. By dividing out powers of $x_0$ we may therefore assume that $x_0\nmid g_i$ (note that when doing so, we do not change the sparsity of $g_i$).
    
    Since $g_i$ is an $(n+1,d,s)$-sparse polynomial it follows that $g_{i,0}$ is $(n,s,d)$-sparse polynomial. Furthermore, for every $i$ we have $f_{i,0}|g_{i,0}$.  Therefore, it follows from \cref{G hits factors of sparse polynomials} that
    $f_{i,0}(G_{(m)})\neq 0$, and hence $x_0\nmid f_i(G_{(m)})$. That is, the multiplicity of $x_0$ as a factor
    of $f$ is the same as its multiplicity as a factor of $f(x_0, G_{(m)})$. As the latter is
    a polynomial on a constant number of variables of bounded individual degree $\poly(n,d,s,\ell)$,
    we can interpolate it using \Cref{Interpolation of polynomials} and directly find $k$ (which is the minimal power of $x_0$ dividing all monomials
    of $f(x_0, G_{(m)})$).
    
    For the "moreover" part, let $f={x_0}^{k}g$, and
    suppose $\valpha=(\alpha_0,\alpha_1,\ldots,\alpha_n)\in \mathbb{F}^{n+1}$ is a point we want to evaluate $g$ on.  Consider $h(x)=f(x,\alpha_1,\ldots,\alpha_n)=x^kg(x,\alpha_1,\ldots,\alpha_n)$ as a polynomial in $\mathbb{F}[x]$.
    One can use a standard one-variable interpolation (as stated in \Cref{Interpolation of polynomials}) to recover $h(x)$ from blackbox access to $f$,
    and therefore to recover the polynomial $g(x,\alpha_1,\ldots,\alpha_n)=h(x)/{x^k}$.
    By plugging $x=\alpha_0$ one can then recover $g(\valpha)$, as we wanted.
\end{proof}

\begin{corollary}[Finding monomial divisors]\label{cor:monomial-div}
    There is a deterministic $\poly(n,d,s,\ell)$-time algorithm that, given blackbox access to a product $f$ of $\ell$ polynomials in $\cCF(n,s,d)\subset\mathbb{F}[\vx]$, returns the multiplicity $k_i$ of $x_i$ as a factor of $f$. 
    Moreover, if $M$ is the largest monomial divisor of $f$, such that $f=Mg$, then we can obtain blackbox access to $g$ from blackbox access to $f$. 
\end{corollary}

\begin{proof}
    This is an immediate corollary of \Cref{reduction for nonzero constant term}. We find the maximal power of each variable dividing $f$. This gives us the monomial $M$. To get blackbox access to $g$ we do the following: Given a point $\valpha\in\F^n$, if all coordinates of $\valpha$ are nonzero then we can compute $g(\valpha)=f(\valpha)/M(\valpha)$. If some coordinates of $\valpha$ are zero, so we cannot divide by $M(\valpha)$, then we do the following. Consider the substitution $\vx\leftarrow \valpha + (t,\ldots,t)$, where $t$ is a new variable. By evaluating $\tilde{g}_\valpha(t):=(f/M)(\valpha + (t,\ldots,t))$ at more than $dn$ values of $t$ such that $t\not\in \{-\alpha_1,\ldots,-\alpha_n\}$, we can interpolate $\tilde{g}(t)$. Then we can evaluate $g(\valpha)=\tilde{g}_\valpha(0)$. Thus, any query to $g$ requires $O(dn)$ queries to $f$.
\end{proof}

\begin{remark}
    The first part of the previous corollary is an easy special case of the technical lemmas in \Cref{Section 4}. In particular, it is always possible to determine the multiplicity of an irreducible sparse polynomial as a factor of such an $f$ using the factorization of $f(x_0,G_{(m)})$. However, as we show in the next section, a larger value of $m$ is required in general.
\end{remark}


\section{Divisibility Testing and Exact Power testing}\label{Section 4}
In this section we prove our divisibility testing and exact power testing results. 
For that we first have to prove some technical results about our generator $G$.
We split the discussion on these technical results into two parts. In \Cref{sec:arb-char} we obtain some results that are valid over
arbitrary fields; in \Cref{sec:char-0} we prove a stronger version of these results, that is only valid over zero or large enough characteristic.

\subsection{Arbitrary Characteristic Results}\label{sec:arb-char}
In this subsection we build on the methods of \cite{bisht2025solving} to obtain our divisibility testing results over fields of arbitrary characteristic.
To do that, we define and study the notion of primitive divisors for a class of polynomials and develop the tools needed for our algorithms; this is done in \Cref{sec:arb-char tools}. In  \Cref{sec:arb-char algs} we use these tools to obtain our algorithms.

\subsubsection{Arbitrary Characteristic Tools}\label{sec:arb-char tools}
We start with the definition of a \emph{primitive divisor}.

\begin{definition}\label{def:primitive divisor}
Let $\cP \subset \mathbb{F}[\vx]$ be a class of polynomials. We say that a polynomial
$g \in \mathbb{F}[\vx]$ is a primitive divisor for $\cP$ if it satisfies the following two conditions.
\begin{enumerate}
    \item For every $f \in \cP$ and every $t \in \mathbb{N}$, the polynomial $\gcd(f, g^t)$ is a power of $g$.
    Equivalently, for every $f \in \cP$, there exists $k \ge 0$ and a polynomial $\tilde g$ coprime to $g$
    such that $f = g^k \cdot \tilde g$.
    \item No non-scalar multiple of $g$ satisfies the first condition.
\end{enumerate}
\end{definition}

\begin{example}\label{irreducible elements of C are primitive}
    An irreducible polynomial $f \in \cP$ is clearly a primitive divisor for $\cP$.
\end{example}

\begin{remark}
    Every primitive divisor for $\cP$ divides some polynomial $f \in \cP$.
Indeed, suppose that $g$ is a primitive divisor for $\cP$ that does not divide any polynomial in $\cP$.
Then, for every $f \in \cP$, we have $\gcd(f,g)=1$, and hence $g$ is coprime to every element of $\cP$.
It follows that $\gcd(f,g^2)=1$ for every $f \in \cP$, so $g^2$ also satisfies the first condition.
This contradicts the second condition for $g$.
\end{remark}

\begin{remark}
    It is not true that primitive divisors for $\cP$ are necessarily elements of $\cP$. Consider for example the class of polynomials $\cP=\{g^2,g^3\}$ for some non-constant polynomial $g\in \mathbb{F}[\vx]$. Then, the only primitive divisor (up to multiplication by scalars) for $\cP$ is $g$, which is not an element of $\cP$.
\end{remark}

In some sense, primitive divisors are to a class of polynomials $\cP$ what irreducible polynomials are to the class of all polynomials.
To make this statement precise, we prove that every $f\in\cP$ can be written uniquely as a product of primitive divisors. That is, primitive
divisors for $\cP$ are indeed the "minimal building blocks" for the polynomials in $\cP$.

To prove that, we start by showing that any two different primitive divisors for a class $\cP$ are coprime.
\begin{lemma}\label{lem:different primitive divisors are coprime}
    Let $\cP\subset \mathbb{F}[\vx]$ be a class of polynomials, and let $g$ and $h$ be two non-associate primitive divisors for it.
    Then $gcd(g,h)=1$.
\end{lemma}

\begin{proof}
    Suppose for the sake of contradiction that $\gcd(g, h)\neq 1$, and let $f_1,\ldots,f_N$ be the irreducible factors of it, for $N\ge 1$.
    Write:
    \begin{equation}
        g = \prod_{1\leq i\leq N}{{f_i}^{u_i}}\cdot \prod_j{g_j^{e_j}}\ , \quad h = \prod_{1\leq i\leq N}{{f_i}^{v_i}}\cdot \prod_j{h_j^{r_j}}
    \end{equation}
    where $f_i, g_j, h_j$ are all different irreducible polynomials, and all the exponents $u_i,v_i,e_j,r_j$ are positive integers.
    
    We claim that the vectors $(u_1,\ldots,u_N)$ and $(v_1,\ldots,v_N)$ are proportional. Indeed, consider
    some $\phi\in\cP$ for which $g|\phi$. Then, by the first condition of \Cref{def:primitive divisor}, we must have $k'$ and $t'$ for which 
    \[\phi=g^{k'}\cdot\ \tilde{g},\; \text{ where }\; \gcd(g, \tilde{g})=1\; \text{ and} \]
    \[\phi=h^{t'}\cdot\ \tilde{h}, \; \text{ where }\;\gcd(h, \tilde{h})=1\ .\]
    Therefore, by considering the multiplicities of the $f_i$'s as factors of $\phi$, we must have that    \begin{equation}\label{proportional vectors} k'(u_1,\ldots,u_n)=t'(v_1,\ldots,v_n).
    \end{equation}
    Let $(k,t):=(k'/\gcd(k',t'), t'/\gcd(k',t'))$.
    That is, $(k,t)$ is the smallest pair of positive integers satisfying \eqref{proportional vectors}.
    Consider the polynomial given by:
    \begin{equation}    \tilde{\phi}=\prod_{1\leq i\leq n}{f_i^{ku_i}}\cdot \prod_j{g_j^{ke_j}}\cdot \prod_j{h_j^{tr_j}}.
    \end{equation}
    Note that as $k,t\ge 1$ ($k\ge 1$ because $g|\phi$, and $t\ge 1$ by \eqref{proportional vectors}), $\phi$ is a non-trivial multiple of $g$ and $h$. 
    
    We claim that $\tilde{\phi}$ satisfies the first condition of \Cref{def:primitive divisor}. Indeed, let  $\psi\in\cP$.
    Since $g$ and $h$ are primitive divisors, we can write:
    \begin{align}\label{f tilde}
        \psi &=g^{k_0}\cdot\ \hat{g} \text{ where } \gcd(g, \hat{g})=1 . \\
        \psi & =h^{t_0}\cdot\ \hat{h} \text{ where } \gcd(h, \hat{h})=1.\notag
    \end{align}
    By examining the multiplicities of $f_i$ as factors of $\psi$, we get that $k_0(u_1,\ldots,u_N)=t_0(v_1,\ldots,v_N)$.
    As $(k,t)$ is the minimal pair of positive integers satisfying \eqref{proportional vectors}, we must have that $k_0=ck$ and $t_0=ct$
    for some non-negative integer $c$. Now, by plugging this information back to \eqref{f tilde}, it can be easily checked that
    $\psi={\tilde{\phi}}^c\cdot \hat{\psi}$ for some polynomial $\hat{\psi}$ coprime to $\tilde{\phi}$. This proves that $\tilde{\phi}$ satisfies the first condition
    of \Cref{def:primitive divisor}. However, this is a contradiction to the fact that $g$ and $h$ satisfy the second condition of it.
    This proves that $g$ and $h$ cannot share an irreducible factor, and completes the proof.
\end{proof}

We next state a lemma capturing the intuition that primitive divisors play a similar role to irreducible polynomials.

\begin{lemma}\label{factorization to primitive divisors}
    Let $\cP$ be some class of polynomials, and let $f\in\cP$. Then, $f$ can be uniquely written as a product of primitive divisors
    for $\cP$.
\end{lemma}

\begin{proof}
    Let $f\in\cP$. Consider some irreducible factor $f_0$ of $f$. Clearly, $f_0$ satisfies the first condition of \Cref{def:primitive divisor}.
    Let $g_0$ be a maximal divisor of $f$ that satisfies the property that it is divisible by $f_0$ and satisfies the first condition of \Cref{def:primitive divisor}. We claim that $g_0$ must be a primitive divisor.
    For this we must show that it satisfies the second condition of \Cref{def:primitive divisor}. Assume for a contradiction that there is some $h$, which is a non-trivial multiple of $g_0$, that satisfies the first condition of \Cref{def:primitive divisor}.
    By maximality of $g_0$, $h \nmid f$. From the first condition we obtain that $\gcd(f,h)=1$. However,  since $f_0\mid g_0$ and $g_0\mid h$, we get that $f_0|\gcd(f,h)$, in contradiction.
    
    Thus, we identified some primitive divisor $g_0|f$. By definition, we can write $f=g_0^k\tilde{g}_0$ with $\gcd(g_0,\tilde{g}_0)=1$. If $\tilde{g}_0=1$ then we are done. If not, fix some irreducible factor  $f_1| \tilde{g}_0$. By the same reasoning as before,
    there is a multiple of $f_1$, denoted $g_1$, which is a primitive divisor for $\cP$, and that divides $f$.
    By \Cref{lem:different primitive divisors are coprime}, and by the fact that $f_1$ divides $g_1$ but not $g_0$, it follows that
    we must have that $g_0$ and $g_1$ are coprime. Therefore, we can write $f=g_0^{k_0}g_1^{k_1}\tilde{g}_1$,
    where $\tilde{g}_1$ is coprime to both $g_0$ and $g_1$. Continuing this way, we can obtain the desired factorization of 
    $f$ to a product of primitive divisors.
    
    To end the proof, we note that the uniqueness of the factorization thus obtained follows from \Cref{lem:different primitive divisors are coprime}.
\end{proof}

We can therefore define the multiplicity with which a primitive divisor divides a polynomial in the class.

\begin{definition}\label{multiplicity of primitive divisor}
    Let $f\in \cP$, and let $g$ be a primitive divisor for $\cP$.
    The multiplicity of $g$ as a factor of $f$, denoted by $v_g(f)$, is the power of $g$ appearing in the factorization of $f$ to primitive divisors. 
\end{definition}

Next, we discuss the connection of the definition of primitive divisors and the non-degeneracy condition of \cite{bisht2025solving}. The proof of the next lemma follows from \cite[Lemma 4.5,4.6]{bisht2025solving}.

\begin{lemma}[{\cite[Lemma 4.5,4.6]{bisht2025solving}}]\label{lem:BV lemma}
    Let $\phi$ and $\psi$ be two non-associate \emph{irreducible} polynomials. Fix some $i$ for which $x_i\in \textrm{var}(\phi)$.
    Suppose that there are $s$-sparse polynomials $f$ and $g$
    of bounded individual degree $d$ so that $\phi|f$ and $\psi|g$. Denote with $v_{\phi}(f)$ the multiplicity of $\phi$ as a factor of $f$, and similarly for $\psi$ and $g$. Suppose that the matrix
    $\begin{pmatrix}
        v_{\phi}(f) & v_{\phi}(g) \\
        v_{\psi}(f) & v_{\psi}(g)
    \end{pmatrix}$
    is invertible.
    Then, the polynomials $\phi(G_{(m,i)})$ and $\psi(G_{(m,i)})$ are coprime as polynomials in $u$ for some explicit $m=\poly(n,{(d^2)}!,s^{d^3})$.
\end{lemma}

The condition of the lemma can be better understood in terms of primitive divisors, as demonstrated by the following claim.

\begin{claim}\label{BV in terms of primitive divisors}
    Let $\cP$ be a class of polynomials, and let $\phi$ and $\psi$ be two irreducible polynomials that divide some (not necessarily the same) polynomial in $\cP$.
    Then, the following are equivalent.
    \begin{enumerate}
        \item $\phi$ and $\psi$ divide the same primitive divisor of $\cP$.\label{item:same-prim-div}
        \item For every $f,g\in \cP$ it holds that the matrix
            $\begin{pmatrix}
            v_{\phi}(f) & v_{\phi}(g) \\
            v_{\psi}(f) & v_{\psi}(g)
        \end{pmatrix}$
        is not invertible.\label{item:degenerate}
    \end{enumerate}
\end{claim}

\begin{proof}
    First, note that the term "the primitive divisor that $\phi$ divides" is well defined, as by \Cref{lem:different primitive divisors are coprime} there is at most one (up to multiplication by scalars) such primitive divisor, and by the fact that $\phi$ divides some polynomial in $\cP$ there is at least one such divisor; clearly, the same is also true for $\psi$. 
    
    For the implication \eqref{item:same-prim-div}$\implies$ \eqref{item:degenerate} note that if $h$ is the primitive divisor that $\phi$ and $\psi$ divide, by \Cref{lem:different primitive divisors are coprime} and \Cref{factorization to primitive divisors} it follows that for every $f\in \cP$ we have that
    $(v_\phi(f), v_\psi(f))=v_h(f)(v_\phi(h), v_\psi(h))$. Therefore, it is clear that for every pair $f,g\in \cP$ the vectors $(v_\phi(f), v_\psi(f))$ and $(v_\phi(g), v_\psi(g))$ are proportional, and hence the matrix whose columns are these vectors is not invertible.

    The next claim  establishes the implication \eqref{item:degenerate}$\implies$ \eqref{item:same-prim-div}.
    \begin{claim}\label{cla:different primitive divisors are non-degenerate}
        Let $h_1$ and $h_2$ be two non-associate primitive divisors for a class $\cP$. Then, there are polynomials $f_1,f_2\in \cP$ for which the matrix
        $\begin{pmatrix}
            v_{h_1}(f_1) & v_{h_1}(f_2) \\
            v_{h_2}(f_1) & v_{h_2}(f_2)
        \end{pmatrix}$
        is invertible.
    \end{claim}
    \begin{proof}
        Consider the set of vectors $(v_{h_1}(f), v_{h_2}(f))_{f\in \cP}$.
        To prove the claim, it is enough to prove that there are two independent vectors in this set. Suppose for the sake of contradiction that this is not the case. Then, it follows that the $\mathbb{Z}$-span of this set is a one dimensional lattice, generated by some vector $(k,t)$. 
        
        We claim that $h:=h_1^{k}h_2^{t}$ must satisfy the first condition of \Cref{def:primitive divisor}. Indeed, consider some element $\tilde{f}\in \cP$. By \Cref{factorization to primitive divisors} we can write $\tilde{f}=h_1^{v_{h_1}(f)}h_2^{v_{h_2}(f)}g$, for some polynomial $g$ coprime to both $h_1$ and $h_2$. By definition of $(k,t)$, there is some non-negative integer $c$ for which $(v_{h_1}(f),v_{h_2}(f)) = c(k,t)$. Thus, $\tilde{f}=h_1^{ck}h_2^{ct}g=h^cg$, where $g$ is coprime to $h$.
        That is, $h$ indeed satisfies the first condition of \Cref{def:primitive divisor}. As $h_1$ is a primitive divisor, $v_{h_1}(\tilde{f})\ge 1$ for some $\tilde{f}\in 
        \cP$, and therefore we must have that $k\ge 1$. Similarly, we obtain $t\ge 1$. That is, $h$ is a non-trivial multiple of $h_1$ and $h_2$ satisfying the first condition of \Cref{def:primitive divisor}, in contradiction.
        This concludes the proof.
    \end{proof}
    Observe that \Cref{cla:different primitive divisors are non-degenerate} implies  $\neg$ \eqref{item:same-prim-div}$\implies$ $\neg$\eqref{item:degenerate}. Indeed,
    let $h_1$ and $h_2$ be the primitive divisors that $\phi$ and $\psi$ divide, respectively. Suppose that $h_1$ and $h_2$ are not associated. By \Cref{cla:different primitive divisors are non-degenerate}, there are $f_1, f_2\in \cP$ such that the matrix $\begin{pmatrix}
            v_{h_1}(f_1) & v_{h_1}(f_2) \\
            v_{h_2}(f_1) & v_{h_2}(f_2)
    \end{pmatrix}$
    is invertible. From divisibility we conclude
    \begin{equation*}
        \begin{aligned}
            \begin{pmatrix}
            v_{h_1}(f_1) & v_{h_1}(f_2) \\ v_{h_2}(f_1) & v_{h_2}(f_2)
            \end{pmatrix}
            =
            \begin{pmatrix}
            \frac{v_{\phi}(f_1)}{v_{\phi}(h_1)} & \frac{v_{\phi}(f_2)}{v_{\phi}(h_1)} \\ \frac{v_{\psi}(f_1)}{v_{\psi}(h_2)} & \frac{v_{\psi}(f_2)}{v_{\psi}(h_2)}
            \end{pmatrix}
        \end{aligned}
    \end{equation*}
    and therefore 
    \begin{equation*}
        \det \begin{pmatrix}
        v_{\phi}(f_1) & v_{\phi}(f_2) \\ v_{\psi}(f_1) & v_{\psi}(f_2) \end{pmatrix}
        =
        v_{\phi}(h_1)v_{\psi}(h_2)
        \det \begin{pmatrix}
        v_{h_1}(f_1) & v_{h_1}(f_2) \\ v_{h_2}(f_1) & v_{h_2}(f_2)\end{pmatrix}
        \neq 0.
    \end{equation*}
    Hence the matrix
    $\begin{pmatrix}
            v_{\phi}(f_1) & v_{\phi}(f_2) \\
            v_{\psi}(f_1) & v_{\psi}(f_2)
    \end{pmatrix}$
    is invertible. This concludes the proof.
\end{proof}

\Cref{BV in terms of primitive divisors} explains in a more natural way the condition from \Cref{lem:BV lemma}: if the matrix of multiplicities is degenerate, then $\phi$ and $\psi$ divide the same primitive divisor for $\cP(n,s,d)$, and are therefore "inseparable" from the point of view of this class.
This motivates us to formulate an analog of \Cref{lem:BV lemma} for every two primitive divisors of  $\cP(n,s,d)$.

\begin{lemma}\label{lem:G presereves coprimality of primitve divisors}
    Let $\phi$ and $\psi$ be two non-associate primitive divisors for $\cP(n,s,d)$.
    Fix some $1\leq i\leq n$ such that $x_i\in\textrm{var}(\phi)$. Then for some explicit value $m=\poly(n,{(d^2)}!,s^{d^3})$ we have that $\phi(G_{(m,i)})$
    is coprime to $\psi(G_{(m,i)})$, when both are viewed as polynomials in $u$.
\end{lemma}

The proof will be nearly identical to the proof of \Cref{lem:BV lemma}, which captures \cite[Lemma 4.6]{bisht2025solving}, with the main difference being  that \Cref{lem:BV lemma} concerns irreducible polynomials while \Cref{lem:G presereves coprimality of primitve divisors} considers primitive divisors. We give a sketch of it for the sake of completeness.

\begin{proof}
    In this proof, whenever we refer to a primitive divisor, we mean a primitive divisor for $\cP(n,s,d)$.

    By \Cref{BV in terms of primitive divisors}, there are $f_1,f_2\in \cP(n,s,d)$ for which the matrix
    \[\begin{pmatrix} k_1 & t_1 \\ k_2 & t_2 \end{pmatrix}:=
    \begin{pmatrix} v_{\phi}(f_1) & v_{\phi}(f_2) \\ v_{\psi}(f_1) & v_{\psi}(f_2) \end{pmatrix}\]
    is invertible.
    Let $h_1=f_1^{t_1+t_2}$ and $h_2=f_2^{k_1+k_2}$. By a simple calculation it can be shown that, up to switching $h_1$ and $h_2$,
    we have $h_1=\gcd(h_1,h_2)\phi^{a}H_1$ and $h_2=\gcd(h_1,h_2)\psi^{b}H_2$, where $a,b\ge 1$ and $H_1,H_2$ are
    coprime to $\phi,\psi$ (by unique factorization and properties of primitive divisors).
    
    Consider the Sylvester matrix of $M_{x_i}(h_1,h_2)$. As both $f_1$ and $f_2$ have individual degrees bounded by $d$ we have that $k_1+k_2,m_1+m_2\leq 2d$. Therefore, $h_1$ and $h_2$ are $s^{2d}$-sparse polynomials of bounded individual degree $2d^2$. Thus, by \Cref{basic resultant properties}, determinants of minors of $M_{x_i}(h_1,h_2)$ are $(4d^2)!s^{8d^3}$-sparse polynomials of individual degree at most
    $8d^4$. 
    \Cref{lem:first property of G} implies that for 
    \[m={(n(8d^4)((4d^2)!s^{8d^3}))}^{2}=\poly(n,{(d^2)}!,s^{d^3})\] 
    no nonzero determinant of a minor vanishes when composing with $G_{(m,i)}$. Therefore, it follows that 
    \[\operatorname{rank}(M_{x_i}(h_1,h_2))=\operatorname{rank}(M_{u}(h_1(G_{(m,i)}),h_2(G_{(m,i)})))\ .\] \Cref{basic resultant properties}\eqref{item:dim=deg}    then ensures that $\deg_{x_i}(\gcd(h_1,h_2))=\deg_u(\gcd(h_1(G_{(m,i)}), h_2(G_{(m,i)})))$. Hence, composition with $G_{(m,i)}$ did not change
    the degree in $u$. Consequently,  $\phi(G_{(m,i)})$ and $\psi(G_{(m,i)})$ are coprime as polynomials in $u$, since otherwise  their common factor would increase the degree of the gcd.

    This concludes the sketch of the proof. We note that  in \Cref{lem:G always preserves coprimality of factors},
    we prove a similar claim, where we explain in more detail the technical points in the argument above.
    These are also explained, of course, in \cite[Lemma 4.5,4.6]{bisht2025solving}.
\end{proof}

We next state an important corollary.
\begin{corollary} \label{G preserves coprimality of factors when one factor is sparse}
    Let $\phi\in\cP(n,s,d)$ be an irreducible polynomial,
    and let $\psi\in\cCF(n,s,d)$ be an irreducible factor of some $g\in \cP(n,s,d)$, which is not associate to $\phi$.
    Fix some $1\leq i\leq n$ such that $x_i\in\var(\phi)$. Then, for an explicit $m=\poly(n,{(d^2)}!,s^{d^3})$, we have that
    $\phi(G_{(m,i)})$ is coprime to $\psi(G_{(m,i)})$, when viewed as polynomials in $u$.
    Moreover, if $\phi$ is multiquadratic, we can take $m=\poly(n,{(d^2)}!,s^{d^2})$.
\end{corollary}

\begin{proof}
    Intuitively, this should follow from \Cref{lem:G presereves coprimality of primitve divisors} (recalling that an irreducible polynomial in $\cP(n,s,d)$ is also a primitive divisor for it \Cref{irreducible elements of C are primitive}).
    Technically speaking we have in hand an irreducible polynomial $\psi$, which is not necessarily a primitive divisor (as it might not belong to $\cP(n,s,d)$). We get around this issue in a similar manner to the proof of \Cref{BV in terms of primitive divisors}. We instead consider the (unique, up to multiplication by scalars)
    primitive divisor $\tilde{\psi}$ that $\psi$ divides.
    Clearly $\phi$ and $\tilde{\psi}$ are non-associate, so we can just apply \Cref{lem:G presereves coprimality of primitve divisors} with $\phi$ and $\tilde{\psi}$ and conclude the claim.
    
    For the "moreover" part, we note that the proof of \Cref{lem:G presereves coprimality of primitve divisors} actually gives a better result when $\phi$ is multiquadratic. We describe what changes should be made in the proof of \Cref{lem:G presereves coprimality of primitve divisors} to see it; from now on, we use the notations from that argument.
    
    In this case, as $\phi\in \cP(n,s,d)$ is irreducible, one can take $f_1=\phi$, and  $f_2$ to be some polynomial that $\tilde{\psi}$ divides. Therefore, we have that $k_1=1, k_2=0$.
    It follows that:
    \begin{itemize}
        \item $h_1=f_1^{t_1+t_2}=\phi^{t_1+t_2}$ is a polynomial of degree at most $2(t_1+t_2)\leq 4d$, instead of $2d^2$, as a polynomial in $x_i$. Its sparsity is bounded above by $s^{2d}$. 
        \item $h_2=f_2^{k_1+k_2}=f_2$ is a polynomial of degree at most $d$, instead of $2d^2$, as a polynomial in $x_i$. Its sparsity is bounded above by $s$.
    \end{itemize}
    As a result, it follows that any determinant of a minor of $M_{x_i}(h_1,h_2)$ is an $s^{O(d^2)}$-sparse polynomial, of bounded individual degree at most $O(d^2)$. This implies that it is sufficient to take
    $m=\poly(n,{(d^2)}!,s^{d^2})$ in this case.
\end{proof}

\subsubsection{Arbitrary Characteristic Algorithms}\label{sec:arb-char algs}
Our next goal is to obtain a divisibility testing result for products of sparse polynomials of bounded individual degree.
To do that, we need the following lemma, which is the heart of the divisibility testing algorithm. This lemma will play
a crucial part in \Cref{Section 5}, when we prove our rational interpolation result.

\begin{lemma}\label{lem:Recovering multiplicity of primitive divisors}
    There is a deterministic $\poly(n,{(d^2)}!,s^{d^3},\ell)$-time algorithm that,
    given blackbox access to a product $f$ of $\ell$ $(n,s,d)$-sparse polynomials, and blackbox access to a primitive divisor $\phi$ for $\cP(n,s,d)$, returns the multiplicity of $\phi$ as a factor of $f$ (in the sense of \Cref{multiplicity of primitive divisor}). In fact, the algorithm only needs to know the values of $f$ on the image of the generator $G_{(m)}$, for some explicit value $m=\poly(n,{(d^2)}!,s^{d^3})$. 
    Moreover, if $\phi$ is an irreducible $s$-sparse polynomial, this algorithm works if $f$ is a product of polynomials in $\cCF(n,s,d)$ (rather than $\cP(n,s,d)$); if furthermore $\phi$ is multiquadratic, taking $m=\poly(n,{(d^2)}!,s^{d^2})$ suffices, and the running time reduces to $\poly(n,{(d^2)}!,s^{d^2},\ell)$.
\end{lemma}

\begin{proof}
    Fix some $i$ for which $x_i \in \text{var}(\phi)$, and let $m=\poly(n,{(d^2)}!,s^{d^3})$ so that the conclusion of \Cref{lem:G presereves coprimality of primitve divisors} holds.  \Cref{G preserves degree} implies that $\deg_u(\phi(G_{(m,i)})) = \deg_{x_i}(\phi)\ge 1$.
    
    By \Cref{factorization to primitive divisors} we can write $f=\phi^k\cdot \prod{\psi}_j$,
    where $\psi_j$ are  primitive divisors for $\cP(n,s,d)$ that are not associate to $\phi$. From \cref{lem:G presereves coprimality of primitve divisors} it follows that for every $\psi_j$
    we have that $\phi(G_{(m,i)})$ is coprime to $\psi_j(G_{(m,i)})$, when both are viewed as polynomials in $u$. In particular, we have that  $\gcd_u(\phi(G_{(m,i)}), \prod{\psi_j(G_{(m,i)})})=1$. As $\deg_u(\phi(G_{(m,i)}))\ge 1$, we must have that the highest power of $\phi(G_{(m,i)})$ that divides $f$ as a polynomial in $u$ is exactly $k$.

    Thus, to recover the desired multiplicity, we will calculate, for every $1\leq i\leq n$, 
    the highest power of $\phi(G_{(m,i)})$ that divides $f$ as polynomials in $u$, and return the minimal value.
    This power can be calculated by interpolating both $\phi(G_{(m,i)})$ and $f(G_{(m,i)})$ using
    \cref{Interpolation of polynomials} and  computing the highest power of $\phi(G_{(m,i)})$ dividing $f(G_{(m,i)})$ as polynomial in $u$ (just by
    regular polynomial division, there is no need for polynomial factorization here).
    Note that the degree of both polynomials is at most $\poly(n,d,s,m,\ell)$ so the claim on the running time holds.
    This ends the description of the algorithm.

    As for the "moreover" part, it follows from the same proof, by considering the factorization into irreducible polynomials (rather than to primitive divisors), and by replacing
    the use of \Cref{lem:G presereves coprimality of primitve divisors} with
    \Cref{G preserves coprimality of factors when one factor is sparse}.
\end{proof}

We are finally ready to obtain our divisibility testing result. This is a generalization of the divisibility testing result 
of \cite[Theorem 7]{volkovich2017some}.

\begin{theorem}\label{general characteristic divisibility testing}
    There exists a deterministic $\poly(n,{(d^2)}!,s^{d^3},\ell)$-time algorithm that, given blackbox access to two products $f,g$ of $\ell$ $(n,s,d)$-sparse polynomials, decides if $f|g$.
    Moreover, if all irreducible factors of $f$ are in fact sparse (that is, $f$ is a product of irreducible $(n,s,d)$-sparse polynomials), this algorithm works even if $g$ is a product of polynomials in $\cCF(n,s,d)$ (rather than in $\cP(n,s,d)$); if furthermore all factors of $f$ are multiquadratic, then the running time of this algorithm can be reduced to $\poly(n,{(d^2)}!,s^{d^2},\ell)$.
\end{theorem}

\begin{proof}
    Fix some $m=\poly(n,{(d^2)}!,s^{d^3})$ such that \Cref{lem:Recovering multiplicity of primitive divisors} applies. 
    For every $1\leq i\leq n$ check if $f(G_{(m,i)})$ divides $g(G_{(m,i)})$, and return that $f|g$ if all these checks passed.
    Clearly, if $f|g$ we return the correct answer. Suppose now that $f\nmid g$. By \Cref{factorization to primitive divisors}, we know that
    there is a primitive divisor $\phi$ whose multiplicity as a factor of $f$ exceeds its multiplicity in $g$.
    Then, by the proof of \Cref{lem:Recovering multiplicity of primitive divisors} it follows that for every $i$ for which
    $x_i\in \textrm{var}(\phi)$, the highest power $k$ for which $\phi(G_{(m,i)})^k$ divides $f(G_{(m,i)})$ 
    actually equals the multiplicity of $\phi$ as a factor of $f$, and analogously for $g(G_{(m,i)})$. In particular, for such values of $i$ it cannot 
    be the case that $f(G_{(m,i)})|g(G_{(m,i)})$, so we will indeed return that $f\nmid g$.

    As for the "moreover" part, it follows from the same proof, when considering the factorization of $f$ and $g$ into irreducible polynomials (rather than to primitive divisors), and by using the "moreover" part of \Cref{lem:Recovering multiplicity of primitive divisors}.
\end{proof}

As we saw in this subsection, the method of \cite{bisht2025solving} is pretty strong for dealing with problems concerning
the factorization of products of sparse polynomials of bounded individual degree. 
Indeed, we will see in \Cref{Section 5} that \Cref{lem:Recovering multiplicity of primitive divisors}
is what we need for our rational interpolation result, and as a consequence, for our algorithm that recovers multiquadratic sparse
irreducible factors of a product of sparse polynomials of bounded individual degree. 

However, this method has some limitations. Consider for example the problem of deciding if a given product of sparse polynomials
of bounded individual degree is a complete power. A natural thing to do is to apply an algorithm that is similar to the one
presented in \Cref{general characteristic divisibility testing} - check if $f(G_{(m,i)})$ are complete powers for every $i$.
This algorithm, however, does not work. To see why, suppose our input polynomial is a primitive divisor of the form $\phi=\phi_1\phi_2$
where $\phi_1,\phi_2$ are two non-associate irreducible polynomials.
The methods of this subsection do not say anything about the interaction of $\phi_1(G_{(m,i)})$ and $\phi_2(G_{(m,i)})$.
In particular, it might be the case that for every $i$, $\phi_1(G_{(m,i)})=\phi_2(G_{(m,i)})$, so that $\phi(G_{(m,i)})$ is
a perfect square.

The goal of the next subsection is to generalize the results of this subsection from dealing only with primitive divisors to
dealing with irreducible polynomials. This allows us, for example, to make the suggested complete power testing
work. Moreover, it enables us to generalize all results from the class of products of sparse polynomials to the class
of products of factors of sparse polynomials.

Our generalization, however, only works for fields of characteristic zero, or when the characteristic is larger
than twice the bound on the individual degree. In particular, the results of the present subsection are the best we can do for small characteristic fields.  

\subsection{Zero or Large Characteristic Results}\label{sec:char-0}
In this subsection we prove stronger versions of the results of \Cref{sec:arb-char},
both in terms of generality and   running times. 
However, the results of this subsection are valid only when
$\operatorname{char}(\mathbb{F})=0$ or when $\operatorname{char}(\mathbb{F})>2d$ (where $d$ is the bound on the individual degree).

Similar to \Cref{sec:arb-char}, in \Cref{sec:zero-char tools} we build the tools that we will  use in our algorithms, and in \Cref{sec:zero-char algs} we give our algorithms.

\subsubsection{Zero or Large Characteristic Tools}\label{sec:zero-char tools}

\sloppy The main technical tool we will need is a generalization of \Cref{lem:G presereves coprimality of primitve divisors}
and \Cref{lem:Recovering multiplicity of primitive divisors} for the case of irreducible factors rather than primitive divisors.
We prove the following analogue of \Cref{lem:G presereves coprimality of primitve divisors}.

\begin{lemma}\label{lem:G always preserves coprimality of factors}
    Let $\F$ be a field of characteristic zero, or of characteristic larger than $2d$.
    Let $\phi,\psi$ be non-associate irreducible polynomials in $\cCF(n,s,d)$ that divide the $(n,s,d)$-sparse polynomials $f$, and $g$ (possibly identical to $f$), respectively. Fix some $i$ for which $x_i \in \text{var}(\phi)$. Then, for 
    some explicit value $m=\poly(n,d!,s^d)$, we have that
    $\phi(G_{(m,i)})$ is coprime to $\psi(G_{(m,i)})$, when viewed as polynomials in $u$.
    Moreover, $\phi(G_{(m,i)})$ and $\psi(G_{(m,i)})$ are both square free as polynomials in $u$.
\end{lemma}

\begin{proof}
    Fix some $i$ for which $x_i \in \text{var}(\phi)$. From now on, view all polynomials as polynomials in $x_i$.
    Consider $M_{x_i}(fg,(fg)')$, the Sylvester matrix of $fg$ and its formal derivative with respect to $x_i$.
    Let $fg=\prod_{j}{h_j^{e_j}}$ be the factorization of $fg$ into irreducible factors (as a polynomial in $x_i$ over $\F(\vx\setminus x_i)$).

    \begin{claim}\label{cla:sum of multiplicities captured by rank of the Sylvester matrix of f and f'}
        Let $\F$ as in \Cref{lem:G always preserves coprimality of factors}.
        Let $f\in \mathbb{F}[x]$ be a non-constant univariate polynomial of degree at most $2d$.
        Suppose $f=\prod_{i}{h_i^{e_i}}$ is the factorization of $f$ into irreducible polynomials.
        Then $\dim(\ker {M(f,f')}=\sum_i{\deg(h_i)(e_i-1)}$.
    \end{claim}
    \begin{proof}
        Clearly, $f'=\sum_{i}{e_ih_i^{e_i-1}h_i'\hat{h}_i}$, where we denote by $\hat{h}_i$ the
        product of the factors of $f$ different than $h_i$.
        As $h_i$ is a non-constant polynomial of degree at most $2d$, the conditions on $\operatorname{char}(\mathbb{F})$ imply $h_i'\neq 0$. Since $h_i$ is irreducible, it must be the case that it is coprime to $h_i'$.
        Moreover, as $h_i$ and $\hat{h}_i$ are clearly coprime, and as $e_i\neq 0$
        (because of the condition on $\operatorname{char}(\mathbb{F})$), we have that the multiplicity of $h_i$ as a factor of $f'$
        is precisely $e_i-1$. As $h_i$ are the only irreducible factors of $f$, it follows that
        $\gcd(f,f')=\prod_{i}{h_i^{e_i-1}}$. \Cref{basic resultant properties} then implies that
        $\dim(\ker {M(f,f')}=\deg(\gcd(f,f')) = \deg(\prod_{i}{h_i^{e_i-1}}) = \sum_i{\deg(h_i)(e_i-1)}$, as claimed.
    \end{proof}
    Continuing with the proof, we clearly have that $\deg_{x_i}(fg)\leq 2d$. Therefore,
    by \Cref{cla:sum of multiplicities captured by rank of the Sylvester matrix of f and f'} $r:=\dim(\ker {M_{x_i}(fg,(fg)')}) = \sum_j{\deg_{x_i}(h_j)(e_j-1)}$. As $M_{x_i}(fg,(fg)')$ is
    an $(4d-1)\times (4d-1)$ matrix, it contains an  invertible $(4d-1-r)\times (4d-1-r)$ minor $M$.
    As $fg$ and $(fg)'$ are $s^2$-sparse, it follows from  \Cref{basic resultant properties}
    that the entries of $M$ are $s^2$-sparse polynomials of bounded individual degree $2d$, and hence $\det M$
    is an $(4d)!s^{8d}$-sparse polynomial of individual degree $8d^2$. Therefore,  \Cref{G hits factors of sparse polynomials} implies that for $m={(n(8d^2)({(4d)}!s^{8d}))}^2=\poly(n,d!,s^d)$, it holds that $\det M({G}_{(m,i)}) \neq 0$ (where we substitute $G_{(m,i)}$ to all coordinates except the $i$-th).
    Note that $\det M({G}_{(m,i)})$ is precisely the determinant
    of a $(4d-1-r)\times (4d-1-r)$ minor of the matrix $M_{u}(fg(G_{(m,i)}),(fg')(G_{(m,i)}))$, and therefore it follows that
    \begin{equation}\label{upper bound on rank}
        \dim(\ker {M_{u}(fg(G_{(m,i)}),(fg)'(G_{(m,i)}))} \leq r.
    \end{equation}
    Now, we know that $fg(G_{(m,i)})=\prod_j{{h_j(G_{(m,i)})}^{e_j}}$; let
    $fg(G_{(m,i)})=\prod_k{{\tilde{h}_k}^{r_k}}$ be the factorization into  irreducible factors of $fg(G_{(m,i)})$.
    We now claim the following.
    \begin{claim}\label{score of factorization always maxinal for irreducible factorization}
        \begin{equation*}            \sum_k{\deg{{\tilde{h}_k}}(r_k-1)}\ge \sum_j{\deg {h_j(G_{(m,i)})}}(e_j-1)   
        \end{equation*}
        and equality holds if and only if the set $\{h_j(G_{(m,i)})\}$ is a set of square free, pairwise coprime polynomials.
    \end{claim}
    \begin{proof}
        By considering the factorization into irreducible polynomials, we can write $h_j(G_{(m,i)})=\prod_k{{\tilde{h}_k}^{e_{j,k}}}$.
        Now, for every $k$ we have that $\sum_j{e_je_{j,k}}=r_k$. Therefore:
        \begin{align*}
             \sum_k{\deg{{\tilde{h}_k}}(r_k-1)} &= \sum_k{\deg{{\tilde{h}_k}}(\sum_j{e_je_{j,k}}-1)} \\
             &\ge \sum_k{\deg{{\tilde{h}_k}}(\sum_j{e_{j,k}}(e_j-1))} = \sum_j{(e_j-1)\sum_k{e_{j,k}}\deg{{\tilde{h}_k}}} \\
             &= \sum_j{\deg{h_j(G_{(m,i)})}(e_j-1)}.
        \end{align*}
        This proves the inequality. Note that equality holds if and only if, for every $k$, $\sum_j{e_{j,k}=1}$. But this clearly
        holds if and only if the set $\{h_j(G_{(m,i)})\}$ is a set of square free, pairwise coprime polynomials, as claimed.
    \end{proof}
    Thus, from  \Cref{score of factorization always maxinal for irreducible factorization}, 
    \Cref{cla:sum of multiplicities captured by rank of the Sylvester matrix of f and f'} and \eqref{upper bound on rank} we deduce that:
    \begin{align*}
        r &\ge \dim(\ker {M_{u}(fg(G_{(m,i)}),(fg(G_{(m,i)}))')} = \sum_k{\deg{{\tilde{h}_k}}(r_k-1)} \\
        &\ge \sum_j{\deg {h_j(G_{(m,i)})}}(e_j-1)=\sum_j{\deg {h_j}}(e_j-1) = r
    \end{align*}
    Here we used that $\deg h_j(G_{(m,i)}) = \deg h_j$, where the degree is with respect to $u$ in the LHS and the with respect to $x_i$ in the RHS. 
    Indeed, note that when viewed as an element of $\mathbb{F}[\vx]$, we have that $h_j\in \cCF(n,s^2,2d)$, therefore the equality follows from
    \Cref{G preserves degree}.
    
    Thus, it follows that we are in the equality case of 
    \Cref{score of factorization always maxinal for irreducible factorization}, and therefore the set
    $\{h_j(G_{(m,i)})\}$ is a set of square free, pairwise coprime polynomials. This ends the proof, as $\{h_j\}$ contains,
    by definition, $\phi$. If $\psi$ is also there then we are done. Otherwise, in the case where $x_i\notin \textrm{var}(\psi)$, the claim is trivial.
\end{proof}


\begin{remark}
    In fact, we can use the results of \Cref{sec:arb-char} to say something even if $d<\operatorname{char}(\mathbb{F})\leq 2d$: Consider the primitive divisors $h_1, h_2$ that $\phi, \psi$ divide. There are two cases:
    \begin{itemize}
        \item If $h_1$ and $h_2$ are not associates, then the claim follows from \Cref{lem:G presereves coprimality of primitve divisors}, but now we must take $m=\poly(n,{(d^2)}!,s^{d^3})$.
        \item If $h_1=h_2$, consider some $\tilde{h}\in \cP(n,s,d)$ so that $h_1=h_2|\tilde{h}$.
        Then, we can replace the polynomial $fg$ in the proof above by the polynomial $\tilde{h}$ and get the result for $m=\poly(n,d!,s^{d})$. Note that now $\deg_{x_i}(\tilde{h})\leq d$, so indeed we only need that $\operatorname{char}(\mathbb{F})\ge d+1$ (or $\operatorname{char}(\mathbb{F})=0$).
    \end{itemize}
\end{remark}

\subsubsection{Zero or Large Characteristic Algorithms}\label{sec:zero-char algs}

From \Cref{lem:G always preserves coprimality of factors}, we are  able to obtain the following generalization of \Cref{lem:Recovering multiplicity of primitive divisors}.
\begin{lemma}\label{Recovering multiplicity of general factors}
    Let $\F$ be a field of characteristic zero, or of characteristic larger than $2d$. 
    There is a deterministic $\poly(n,d!,s^d,\ell)$-time algorithm that, given blackbox access to a product, $f$, of $\ell$ polynomials in $\cCF(n,d,s)$ and  blackbox access to an irreducible polynomial $\phi$ in $\cCF(n,s,d)$, both defined over $\mathbb{F}$, returns the multiplicity of $\phi$ as a factor of $f$. In fact, the algorithm only needs to know the values of $f$ and $\phi$ on the image
    of the generator $G_{(m)}$ for some explicit $m=\poly(n,d!,s^d)$.
\end{lemma}

\begin{proof}
    Let $m=\poly(n,d!,s^d)$ as required by \Cref{lem:G always preserves coprimality of factors}.
    For every $1\leq i\leq n$, using \Cref{Interpolation of polynomials}, we can
    recover $\phi(G_{(m,i)})$ and $f(G_{(m,i)})$. Therefore, using regular division of univariate polynomials, we can find the largest $k$
    for which for every $i$ it holds that $\phi(G_{(m,i)})^k|f(G_{(m,i)})$ as polynomials in $u$.

    Denote with $t$ the multiplicity of $\phi$ as a factor of $f$.
    We claim that $k=t$.
    Indeed, for every $i$ it holds that $\phi(G_{(m,i)})^k|f(G_{(m,i)})$, so clearly $t\leq k$. For the inequality in the other direction,
    fix some $i$ for which $x_i \in \textrm{var}(\phi)$, and write $f=\phi^k\cdot \prod_{j}{\psi_j}$
    where $\psi_j$ are all irreducible and non-associate to $\phi$.
    Then, by the \Cref{lem:G always preserves coprimality of factors}, $\phi(G_{(m,i)})$ is coprime to $\psi_j(G_{(m,i)})$ for any $j$ (as polynomials in $u$), and therefore
    $\phi(G_{(m,i)})\nmid \prod_{j}{\psi_j(G_{(m,i)})}$. Note that here we used that $x_i\in \textrm{var}(\phi)$, as in this case we know that $\deg_{u}(\phi(G_{(m,i)}))=\deg_{x_i}(\phi)\ge 1$ by \Cref{G preserves degree}.
    Thus, $\phi(G_{(m,i)})$ divides $f(G_{(m,i)})$ precisely $k$ times as polynomials in $u$, and
    the lemma follows.
\end{proof}

As we did in the previous subsection, we can use the proof of \Cref{Recovering multiplicity of general factors} to obtain a divisibility testing
result. Since the results of the present subsection apply to general irreducible factors of sparse polynomials, the resulting
divisibility testing algorithms are valid not just for products of sparse polynomials, but for products of factors of sparse polynomials. In contrast, the results in \Cref{sec:arb-char algs} apply only to products of polynomials from $\cP(n,s,d)$.

\begin{theorem}\label{thm:div-text-l-product-char-0
}
    Let $\F$ be a field of 
    characteristic
    zero or larger than $2d$.
    There exists a deterministic $\poly(n,d!,s^d,\ell)$-time algorithm that given blackbox access
    to polynomials $f$ and $g$, such that each is a product of $\ell$ polynomials in $\cCF(n,s,d)$,  decides if $f|g$.
\end{theorem}

\begin{proof}
    Fix some $m=\poly(n,d!,s^{d})$, satisfying the conclusion of \Cref{Recovering multiplicity of general factors}. 
    For every $1\leq i\leq n$ check if $f(G_{(m,i)})$ divides $g(G_{(m,i)})$ (as polynomials in $u$; throughout this argument, we view them in this way), and return that $f|g$ if all these checks passed. Note that no polynomial factorization is needed to do that; regular univariate polynomial
    division actually suffices.
    Clearly, if $f|g$ we return the correct answer. Suppose now that $f\nmid g$; let $\phi$ be an irreducible factor whose multiplicity as a factor of $f$ is larger than its multiplicity as a factor of $g$.
    Then, by the proof of \Cref{Recovering multiplicity of general factors} it follows that for every $i$ for which
    $x_i\in \textrm{var}(\phi)$, the highest powers $k_1,k_2$ for which $\phi(G_{(m,i)})^{k_1}$ divides $f(G_{(m,i)})$ and $\phi(G_{(m,i)})^{k_2}$ divides $g(G_{(m,i)})$,
    actually equal the multiplicity of $\phi$ as a factor of $f$ and $g$, respectively. In particular, for these values of $i$ it cannot 
    be the case that $f(G_{(m,i)})|g(G_{(m,i)})$, so we will indeed return that $f\nmid g$.
\end{proof}

We can now prove \Cref{general divisibility testing}. We repeat its statement for convenience.

\GeneralDivisibility*

\begin{proof}
    The proof follows from combining \Cref{general characteristic divisibility testing} and \Cref{thm:div-text-l-product-char-0 }.
\end{proof}

Taking a closer look at the proof of \Cref{thm:div-text-l-product-char-0
}, we are  able to obtain \Cref{general complete power testing} that generalizes a result of \cite{bisht2025solving}. As mentioned at the end of the previous subsection, we currently do not know a version that works
over arbitrary fields, as we need the full power of \Cref{lem:G always preserves coprimality of factors}.

\GeneralPower*

\begin{proof}
    As before, fix some large enough $m=\poly(n,d!,s^d)$ so that \Cref{lem:G always preserves coprimality of factors} holds. For every $1\leq i\leq n$ compute the factorization of $f(G_{(m,i)})$, using \Cref{Interpolation of polynomials}
    and \Cref{thm:polynomial factorization}.
    From these factorizations, it is easy to check, for every $1\leq i\leq n$, if $f(G_{(m,i)})$ is a complete $e$-th power.
    Indeed, a polynomial is an $e$-th power if and only if every irreducible factor of it divides it with multiplicity divisible by $e$, and its leading coefficient is an $e$-th power as an element of $\mathbb{F}$. The former can be easily checked for $f(G_{(m,i)})$ from its factorization into  irreducible factors. As for the latter, it can be directly checked for every polynomial $h\neq 0$. For instance:
    \begin{itemize}
        \item If $\mathbb{F}=\mathbb{F}_{q}$ for some prime power $q$, check if the leading coefficient $\alpha$ of $h$ satisfies that $\alpha^{\frac{q-1}{e}}=1$.
        \item If $\mathbb{F}=\mathbb{Q}$, use \Cref{thm:lll} to check if $x^{e}-\alpha$ has a rational root, where $\alpha$ is the leading coefficient of $h$.
    \end{itemize}
    We return that $f$ is a complete $e$-th power if $f(G_{(m,i)})$ is a complete $e$-th power for every $1\leq i\leq n$.
    
    Clearly, if $f$ is a complete $e$-th power, the algorithm returns the right answer.
    So, suppose that $f$ is not a complete $e$-th power. Then, by the discussion above, one of the following must hold.
    \begin{itemize}
        \item There's an irreducible factor $\phi|f$ so that
        its multiplicity as a factor of $f$ is some integer $k$, not divisible by $e$.
        Consider some value $i$ for which $x_i\in \textrm{var}(\phi)$, and consider some irreducible factor
        $\tilde{\phi}$ of $\phi(G_{(m,i)})$. We note that such an irreducible factor of positive degree in $u$ exists, by \Cref{G preserves degree}.
        We claim that the multiplicity of $\tilde{\phi}$ as a factor of $f(G_{(m,i)})$
        is $k$. Indeed, by \Cref{lem:G always preserves coprimality of factors} it is coprime to $\psi(G_{(m,i)})$ for every irreducible $\psi|f$
        that is not associate to $\phi$. In addition, it divides $\phi(G_{(m,i)})$ precisely once (by the "moreover" part of the lemma).
        Therefore $\tilde{\phi}$ is an irreducible factor of $f(G_{(m,i)})$ that divides it exactly $k$ times, hence $f(G_{(m,i)})$ is not
        a complete $e$-th power. Thus, we will indeed return that $f$ is not an $e$-th power.
        \item $f=\alpha g^e$ for some $g\in \mathbb{F}[\vx]$ and $\alpha \in \mathbb{F}$, so that $\alpha$ is not an $e$-th power in $\mathbb{F}$. In this case,
        for every $1\leq i\leq n$ have that $f(G_{(m,i)})=\alpha g(G_{(m,i)})^e$ is not an $e$-th power as well, hence we indeed return that $f$ is not an $e$-th power.
    \end{itemize}
    This completes the proof.
\end{proof}

\section{A Rational Interpolation Theorem}\label{Section 5}
In this section we prove our rational interpolation result.

We state our rational interpolation result. The result requires a lot of information about the functions involved, but in our applications we show how to obtain that information.

\begin{theorem}\label{rational interpolation theorem}
    Suppose we have  blackbox access to some rational function of the form $\frac{af}{b}$, defined over a field $\F$, where:
    \begin{enumerate}
        \item $a$ is an $(n,s,d)$-sparse polynomial.
        \item We have  blackbox access $f$, which is a product of $\ell$ polynomials in $\cCF(n,s,d)$. 
        \item We have a set $\cF$ of $s$-sparse irreducible factors of $f$, call them $\phi_i$.
        \item $b$ is an $s$-sparse factor of $f$ whose irreducible factors are all elements of $\cF$.
        \item $a$ and $b$ are coprime.
    \end{enumerate}
    Then, there is an algorithm that recovers $a$ and $b$ from  evaluating $\frac{af}{b}$ on the image of the generator $G_{(m)}$ for some explicit $m$ (depending on $\chr(\F)$), such that the running time is $\poly(m,\ell)$, where:
    \begin{itemize}
        \item If $\chr({\F})=0$ or $\chr(\F)>2d$, one can take $m=\poly(n,d!,s^{d})$.
        \item Over arbitrary fields, $m=\poly(n,{(d^2)}!,s^{d^3})$.   If we further assume that all elements of $\cF$ are multiquadratic, then we can take $m=\poly(n,{(d^2)}!,s^{d^2})$. 
    \end{itemize} 
    \end{theorem}

\begin{proof}
    Assume first that the characteristic of the field is zero or larger than $2d$.
    Fix some $m=\poly(n,d!,s^{d})$ for which the conclusion of \Cref{Recovering multiplicity of general factors} holds.
    We start by recovering $b$.  As all irreducible factors of $b$ belong to $\cF$, it suffices
    to understand the multiplicity of each $\phi \in \cF$ as a factor of $b$.
    Fix such an irreducible factor $\phi\in\cF$. Using \cref{Recovering multiplicity of general factors} we can compute the multiplicity of $\phi$ as a factor of $f$, which we denote by $t$.
    
    As $a$ and $b$ are coprime, it follows that, for any $k$, $\phi^k|b$ if and only if $\phi^{t-k+1} \nmid \frac{af}{b}$.
    Therefore, it is sufficient to understand the multiplicity of $\phi$ as a factor of $\frac{af}{b}$.
    As $f$ is a product of factors of $s$-sparse polynomials, $b| f$, and $a$ is $s$-sparse, $\frac{af}{b}$ is also
    a product of factors of $s$-sparse polynomials. We can therefore apply \cref{Recovering multiplicity of general factors} to learn the multiplicity of $\phi$ as a factor of $\frac{af}{b}$. We conclude that we can learn the factorization of $b$ into irreducible polynomials.
    In particular, we can explicitly recover $b$ using \cref{lem:first property of G}.
    To learn $a$, we can simply compute $a(G_{(m)})=f(G_{(m)})/b(G_{(m)})$ and apply \cref{lem:first property of G}.

    For fields of small characteristic, and the multiquadratic case, we repeat the same argument, replacing the use of \Cref{Recovering multiplicity of general factors} by the "moreover" part of \Cref{lem:Recovering multiplicity of primitive divisors}.
\end{proof}

We next state a slight generalization of the theorem above that we will later use.

\begin{theorem}\label{thm:multi rational interpolation theorem}
    Suppose we have blackbox access to a set of rational functions, defined over a field $\F$, of the form $\frac{a_jf^j}{b}$, for $1\leq j \leq D$, where:
    \begin{enumerate}
        \item $a_j$ is an $(n,s,d)$-sparse polynomial for every $1\leq j \leq D$.
        \item $f$ is a product of $\ell$ polynomials in $\cCF(n,s,d)$
        to which we have blackbox access.
        \item We have a set $\cF$ of $s$-sparse irreducible factors of $f$, call them $\phi_i$.
        \item $b$ is an $s$-sparse factor of $f$ whose irreducible factors are all elements of $\cF$.
        \item $\gcd(a_1,\ldots,a_D, b) = 1$.
    \end{enumerate}
    Then, there is an algorithm that recovers the $a_j$-s and $b$ from evaluations of $\frac{a_jf^j}{b}$ on the image of the generator $G_{(m)}$, for some explicit $m$ (depending on $\chr(\F)$), such that the running time is $\poly(m,\ell,D)$, where:
    \begin{itemize}
        \item If $\chr(\F)=0$ or $\chr(\F)>2d$, one can take $m=\poly(n,d!,s^{d},D)$.
        \item For arbitrary fields we have  $m=\poly(n,{(d^2)}!,s^{d^3},D)$. If we further assume that all elements of $\cF$ are multiquadratic, then we can take $m=\poly(n,{(d^2)}!,s^{d^2},D)$. 
    \end{itemize}
\end{theorem}

\begin{proof}
    The proof goes along the same lines as the proof of \cref{rational interpolation theorem}. The only difference is that now,
    when recovering $b$, it is does not necessarily hold that  $\gcd(a_j,b)=1$. To overcome this obstacle,
    observe that for every $\phi \in \cF$ of multiplicity $t$ as a factor of $f$, it holds that,
    for every $k\ge 1$, $\phi^k|b$ if and only if there is $1\leq j\leq D$ for which $\phi^{tj-k+1} \nmid \frac{a_jf^j}{b}$.
    This is true by the same reasoning as before. Indeed, there is some $1\leq j\leq D$ for which $\phi \nmid a_j$, as otherwise $\phi | \gcd(a_1,\ldots,a_D,b)$.
    The algorithm then proceeds in the exact same way as before.
\end{proof}

\makeatletter
\@ifundefined{c@theorem}{}{\renewcommand*{\theHtheorem}{\thesection.\arabic{theorem}}}
\@ifundefined{c@lemma}{}{\renewcommand*{\theHlemma}{\thesection.\arabic{lemma}}}
\@ifundefined{c@corollary}{}{\renewcommand*{\theHcorollary}{\thesection.\arabic{corollary}}}
\@ifundefined{c@proposition}{}{\renewcommand*{\theHproposition}{\thesection.\arabic{proposition}}}
\@ifundefined{c@claim}{}{\renewcommand*{\theHclaim}{\thesection.\arabic{claim}}}
\@ifundefined{c@definition}{}{\renewcommand*{\theHdefinition}{\thesection.\arabic{definition}}}
\@ifundefined{c@remark}{}{\renewcommand*{\theHremark}{\thesection.\arabic{remark}}}
\@ifundefined{c@fact}{}{\renewcommand*{\theHfact}{\thesection.\arabic{fact}}}
\@ifundefined{c@equation}{}{\renewcommand*{\theHequation}{\thesection.\arabic{equation}}}
\makeatother

\section{Meta-Algorithm}\label{sec:meta}

In this section we present a meta-algorithm that will be used, with some variations, in the proofs of Theorems~\ref{finding sparse divisors of sparse polynomials}, \ref{thm:factor-of-product-of-nsd}, \ref{General factorization algorithm}, \ref{improved BSV}, \ref{improved BSV for products} and \ref{finding multiquadratic factors of a product of sparse polynomials}, which state our algorithmic results. Since \Cref{thm:factor-of-product-of-nsd} follows as a corollary of \Cref{General factorization algorithm}, we will not discuss it here.   

Each of these algorithms follows the outline of the meta-algorithm, and in their proofs we will refer to steps of the meta-algorithm and describe how they are implemented in the relevant setting.


\begin{enumerate}
\item{\bf Setting:}\label{meta:setting}

We assume that we want to learn properties of a polynomial $f$ such that $f$ is either 
\begin{enumerate}
    \item $(n,s,d)$-sparse (\Cref{finding sparse divisors of sparse polynomials}, \Cref{improved BSV})
    \item a product of polynomials in $\cCF(n,s,d)$ (\Cref{General factorization algorithm}, \Cref{improved BSV for products}, \Cref{finding multiquadratic factors of a product of sparse polynomials})
\end{enumerate}

and we want to find a list $\cF(f)$ of divisors of $f$ having certain properties. For every variable $x_i$, the polynomial $f|_{x_i=0}$ is also of the same type as $f$. This is clearly true when $f$ is $(n,s,d)$-sparse. When $f$ is a product of polynomials in $\cCF(n,s,d)$, this follows from \Cref{lem:closure under taking free terms}, whose formulation and proof are given at the end of this section.

In what follows, for the sake of convenience, we assume that $f$ is an $n+1$-variate polynomial and is an element of $\mathbb{F}[x_0,\vx]$.

\item{\bf Preprocessing:}\label{meta:preprocess}

\begin{enumerate}
    \item Remove from $f$ its largest monomial divisor (using \Cref{cor:monomial-div}).
    \item Normalize $f$ to get $\tilde{f}(y,\vx) = f(yf_0,x_1,\ldots,x_n) / f_0$. Set $\hat{f}=f(y,G)$ for an appropriate generator $G$. Observe that $\tilde{f}(y,\vx)$ is reverse monic with respect to $y$, and hence so is $\hat{f}$. By \Cref{cla:reversed-monic bound} it then follows that $\hat{f}$ has at most $\deg_{y}(\hat{f})\leq d$ irreducible factors if $f$ is $(n,s,d)$-sparse, and at most $\deg_{y}(\hat{f})\leq \ell d$ if $f$ is a product of $\ell$ polynomials in $\cCF(n,s,d)$.
    \item Use \Cref{Interpolation of polynomials} to interpolate $\hat{f}$ as a polynomial in $O(1)$ many variables. Then apply \Cref{thm:polynomial factorization} to obtain its factorization into irreducible factors. 
\end{enumerate}


    \item{\bf Recursive call to solve for $f_0$:}\label{meta:recursive} 
    Recursively solve the problem for $f_0=f|_{x_0=0}$, obtaining a list $\cF(f_0)$ of divisors not divisible by a monomial or irreducible factors, according to the problem.
    
\item{\bf Generating a list of candidate solutions:}\label{meta:list}

By \Cref{cla:normalization of a polynomial}, every divisor $h=\sum_{i=0}^{d}{c_ix_0^i}$ of $f$ gives rise to a divisor $\sum_{i=0}^{d}{\frac{c_i}{c_0}f_0^iy^i}$ of $\tilde{f}$ and hence to a divisor $\hat{h}=\sum_{i=0}^{d}{\frac{c_i(G)}{c_0(G)}f_0^iy^i}$ of $\hat{f}$. We would like to get access to $c_i(G)$ in order to interpolate them, and through that obtain access to $h$. This leads us to the rational interpolation problem, where we have some information about $c_0$ (or even $c_0$ itself) coming from $\cF(f_0)$.

\begin{enumerate}
    \item In  \Cref{finding sparse divisors of sparse polynomials}, 
    \Cref{General factorization algorithm} we will have that $c_0\in \cF(f_0)$. In \Cref{improved BSV} $c_0$ is a product of elements of $\cF(f_0)$, but we have enough time to enumerate all relevant products so we can pretend that $c_0\in \cF(f_0)$.

    In fact, the above statement is not entirely accurate. The recursive call returns the set $\cF(f_0)$ of all divisors of $f_0$ that are not divisible by a monomial, together with the largest monomial dividing $f_0$. It might be the case that $c_0$ is a divisor of $f_0$ that \emph{is} divisible by a monomial. In this case, $c_0$ will not belong to $\cF(f_0)$. However, $c_0$ does lie in $\cF(f_0)$ \emph{up to a multiplication by a monomial}. As we shall see, this suffices for our purposes. 
    
    \item In \Cref{improved BSV for products} and \Cref{finding multiquadratic factors of a product of sparse polynomials} the situation is a bit different. Here, $c_0$ may be a product of elements from $\cF(f_0)$, and we cannot go over all possible ways to obtain $c_0$. However, in this case we will have that $\cF(f_0)$ contains all the irreducible factors of $c_0$, which is what \Cref{thm:multi rational interpolation theorem} requires.
\end{enumerate}

Accordingly, for each factor $\hat{h}=\sum_{i=0}^{d}e_i y^i$ we compose a list $\cF_{\hat{h}}$ as follows:

\begin{enumerate}
    \item When we have a list containing (up to multiplication by a monomial) the correct $c_0$ (\Cref{finding sparse divisors of sparse polynomials}, \Cref{General factorization algorithm}, \Cref{improved BSV}), we iterate over all possible candidates $c_0$ in that list. For each such candidate, we interpolate the polynomials of the form $\bigl((x_1\cdots x_n)^d c_0 / f_0^i\bigr)(G)\, e_i$ in order to recover the coefficients $c_i$. We then add to $\cF_{\hat h}$ those polynomials $\sum_{i=0}^{d} c_i x_0^i$ that are valid candidates for the problem, after dividing by their largest monomial divisor:
    \begin{enumerate}
        \item In \Cref{finding sparse divisors of sparse polynomials},  \Cref{General factorization algorithm} we keep all $(n,s,d)$-sparse polynomials.
        \item In \Cref{improved BSV} we keep all $(n,S,d)$-sparse polynomials, where $S=s^{O(d^2\log n)}$ is the upper bound on sparsity of factors of $(n,s,d)$-sparse polynomials obtained in \Cref{thm: BSV bound}.
    \end{enumerate}
    The term $(x_1\cdots x_n)^d$ is needed to deal with the fact that $c_0$ is in $\cF(f_0)$ only up to division by a monomial. We shall explain this point in full detail in the proof of \Cref{finding sparse divisors of sparse polynomials}.

    \item  If we only have a list of irreducible factors of $c_0$ (\Cref{improved BSV for products} and \Cref{finding multiquadratic factors of a product of sparse polynomials}), apply \Cref{thm:multi rational interpolation theorem} on the coefficients of $\hat{h}$ to recover potential $c_0$ and $c_1,\ldots,c_d$. We put the resulting polynomial $\sum_{i=0}^{d}{c_ix_0^i}$ in the list $\cF_{\hat{h}}$ if:
    \begin{enumerate}
        \item For  \Cref{improved BSV for products} it is $(n,S,d)$-sparse. 
        \item For \Cref{finding multiquadratic factors of a product of sparse polynomials}, it is multiquadratic and $(n,s,d)$-sparse.
    \end{enumerate}
    We note that, in both cases, we set $\cF_{\hat{h}}=\cF(f_0)$ for $\hat{h}=1$.
\end{enumerate}

\item{\bf Pruning the list:}\label{meta:prune}

Depending on the application we either have to keep divisors or irreducible divisors. We do that in two steps - we first find the divisors in the list of candidates, and then, if needed, identify the irreducible polynomials among them. The implementation of the two steps vary in the different algorithms; we give here a solution that works for all cases, though not always optimal, and provide a better implementation in the proofs of the algorithms.

In all cases we apply \Cref{general divisibility testing} to remove from the list those polynomials that do not divide $f$. This already suffices for \Cref{finding sparse divisors of sparse polynomials} and \Cref{General factorization algorithm}. 

To obtain irreducible divisors as required in \Cref{improved BSV}, \cref{improved BSV for products}, \cref{finding multiquadratic factors of a product of sparse polynomials}, we apply \Cref{general divisibility testing} to each two polynomials $g,h$ in the list of divisors to check whether $g$ divides $h$, and keep only those divisors that are not divisible by any other divisor.


\item{\bf Computing multiplicities:}\label{meta:multiplicities}

To compute multiplicities, where needed, we apply \Cref{Recovering multiplicity of general factors} or \Cref{lem:Recovering multiplicity of primitive divisors}, depending on the characteristic. Here as well there are better solutions in some cases; these are described in the proofs of the algorithms.

\end{enumerate}

The analysis of the running time of the algorithm depends on the exact setting. However, in all settings we perform exactly one recursive call so the overall complexity is determined by the cost of the different steps of the algorithm, and it does not blow up because of the recursive call.

\begin{remark}
    The described meta-algorithm provides the outline that all our algorithms follow. It is, however, \emph{not} an algorithm that \emph{generalizes} all of them. It is possible to formulate two concrete algorithms  which all our results follow - one generalizing Theorems \ref{finding sparse divisors of sparse polynomials}, \ref{General factorization algorithm} and \ref{improved BSV}, and one generalizing Theorems \ref{improved BSV for products} and \ref{finding multiquadratic factors of a product of sparse polynomials}. However, the precise formulations of these general algorithms are rather complicated and not particularly intuitive; hence, we have chosen not to include them in the present paper.
\end{remark}

We close this section by proving the following lemma.

\begin{lemma}\label{lem:closure under taking free terms}
    Let $f$ be a product of $\ell$ polynomials in $\cCF(n+1,s,d)\subset \mathbb{F}[x_0,\vx]$. Suppose further that $x_0\nmid f$, and let $f_0$ be the free term of $f$ as a polynomial in $x_0$. Then, $f_0$ is a product of $\ell$ elements of $\cCF(n,s,d)\subset \mathbb{F}[\vx]$.
\end{lemma}

\begin{proof}
    Write $f=\prod_{i=1}^{\ell}{g_i}$, where $g_i\in \cCF(n+1,s,d)$. Denote further $g_i=\sum_{j=0}^{d}{g_{i,j}x_0^j}$, where $g_{i,j}\in \mathbb{F}[\vx]$, and let $\hat{g}_i$ be an $(n+1,s,d)$-sparse polynomial that $g_i$ divides. First, note that $f_0=\prod_{i=1}^{\ell}{g_{i,0}}$. Next, observe that $g_{i,0}$ divides the trailing monomial of $\hat{g}_i$ as a polynomial in $x_0$. As this trailing monomial is clearly $(n,s,d)$-sparse, it follows that $g_{i,0}\in \cCF(n,s,d)$. Consequently, $f_0$ is indeed a product of $\ell$ polynomials in $\cCF(n,s,d)$. 
\end{proof}


\section{Finding Multiquadratic Factors}\label{sec:multiquadratic}

In this section we prove \Cref{finding multiquadratic factors of a product of sparse polynomials}. 

\MultiquadraticRecovery*

\begin{proof}

    We follow the outline of the meta-algorithm described in \Cref{sec:meta}. As this is the first time we implement the meta-algorithm, we give the full details of the implementation. We describe the algorithm for the case  $\chr(\F)=0$ or $\chr(\F)>2d$ and then explain which components change when we work over arbitrary fields.

    To simplify the presentation we shall assume that the input $f$ is an element of $\mathbb{F}[x_0,\vx]$, and is a product of $\ell$ elements of $\cCF(n+1,s,d)$.
    As described in \Cref{sec:meta}\ref{meta:preprocess}, we begin with a preprocessing step.

    \paragraph{Preprocessing:}

    By \Cref{cor:monomial-div} we may assume that $f$ is not divisible by any monomial. In particular, we may assume that $x_0\nmid f$. Write $f=\sum_{i=0}^{\ell d}{f_ix_0^i}$,
    where $f_i\in \mathbb{F}[\vx]$. Let $\tilde{f}$ be the normalization of $f$ as in \Cref{def:normalization}. By \Cref{cla:normalization of a polynomial} $\tilde{f}$ is reverse monic and we have blackbox access to it. Denote \[\hat{f}:=\tilde{f}(y,G),\]
    for $G=G_{(m)}$, with $m$ as given in \Cref{thm:multi rational interpolation theorem}. 
    Using \Cref{Interpolation of polynomials} and \Cref{thm:polynomial factorization}, obtain the factorization of $\hat{f}$.
    Note that, as $\tilde{f}$ is reversed-monic with respect to $y$, so is $\hat{f}$. Thus, by \Cref{cla:reversed-monic bound},
        $\hat{f}$ has at most $\ell d$ irreducible factors. Consequently, it has at most $O({(\ell d)}^{2})$ (not necessarily irreducible) factors of degree at most two in $y$ overall. We note that $\hat{f}$ is a polynomial in $O(1)$ variables of degree $\poly(m, \ell)$, therefore the running time of this step is at most $\poly(m, \ell)$.

\paragraph{Recursive call:}  

    We would like to first solve the problem for $f_0$. By \Cref{lem:closure under taking free terms}, $f_0$ is also a product of $\ell$ elements of $\cCF(n,s,d)$. Hence, we can, by applying one recursive call, find the set of all multiquadratic, $s$-irreducible factors of it. Denote this set by $\cF$.

\paragraph{Generating a list of candidate solutions:}

        Recall that \Cref{cla:normalization of a polynomial} shows that for every
        irreducible factor $\phi=b+a_1x_0+a_2x_0^2$ of ${f}$, there is a corresponding irreducible factor $\tilde{\phi}=1+\frac{a_1f_0}{b}y+\frac{a_2f_0^2}{b}y^2$ of $\tilde{f}$, which corresponds to some (not necessarily irreducible) quadratic factor of $\hat{f}$, which we denote by $\hat{\phi}=\tilde{\phi}(y,G)$. Our plan is to learn the $(a_1,a_2,b)$'s, corresponding to the different factors, from the factorization of from $\hat{f}$ and $\cF$. 
        
        For that, we run over all factors $\hat{\phi}$ of $\hat{f}$, that have degree at most two as polynomials in $y$. For each such $\hat{\phi}$ we run the algorithm of \Cref{thm:multi rational interpolation theorem} to try and recover corresponding $a_1,a_2$ and $b$ (or $(a_1,b)$ if $\deg_y(\hat{\phi})=1$). If $b+a_1x_0+a_2x_0^2$ is an $s$-sparse multiquadratic polynomial then we add it to the list of candidate factors. 
        
        \begin{claim}
            In the notation above, assume that $\phi=b+a_1x_0+a_2x_0^2$  corresponds to a multiquadratic (or multilinear) irreducible factor of $f$. Then the procedure described above will add $\phi$ to the list of candidate multiquadratic factors when executed on $\hat{\phi}$.
        \end{claim}
        \begin{proof}
            We need to show that we have all the required information to run \Cref{thm:multi rational interpolation theorem} successfully. Indeed:
            \begin{enumerate}
            \item As $\phi$ is $s$-sparse, $a_1,a_2$ are $s$-sparse as well.
            \item $f_0$ is a product of $\ell$ elements of $\cCF(n,s,d)$ to which we have a blackbox access.
            \item By the recursive call, we have a set $\cF$ containing all multiquadratic, $s$-sparse irreducible factors of $f_0$.
            \item As $\phi$ is multiquadratic then so is $b$. Therefore, by \Cref{factors of multiquadratic sparse polynomials are sparse}, all irreducible factors of $b$ are $s$-sparse and multiquadratic as well. Since $b|f_0$, these factors must be elements of $\cF$.
            \item From the irreducibility of   $\phi$ we obtain $\gcd(a_1,a_2,b)=1$.
        \end{enumerate}
        Thus, \Cref{thm:multi rational interpolation theorem} will return $a_1,a_2$ and $b$.
        \end{proof}
        Therefore, by running over all factors of $\hat{f}$ of degree at most $2$ in $y$,
        we are guaranteed to have included all irreducible multiquadratic factors of $f$ in the list of candidates. 
        Since $f$ may have multiquadratic factors that do not depend on $x_0$, we also add the elements of $\cF$ to the list of candidates for factors of $f$.

        Our next step is finding all the ``true'' factors. Indeed, our list of candidates 
        may contain "fake" candidates as well. For example, these may come from a successful interpolation process
        to some factor $\phi$ that has degree higher than two in some variable $x_i$, a property that was not maintained in the composition, factorization and interpolation process.

\paragraph{Pruning the list:}        
        
  We now have a set of $O((\ell d)^2)$ candidates for divisors of $f$, all of which are multiquadratic and $s$-sparse, and we are left with the problem of identifying the real irreducible divisors among them and computing their multiplicities. 
  
  By running the algorithm of \Cref{general divisibility testing} we decide, for every candidate $\phi\in \cF$, if it really divides $f$. We now have a set of candidates, all of which are multiquadratic and $s$-sparse, that divide $f$, and we have to identify the irreducible ones. 

   \Cref{factors of multiquadratic sparse polynomials are sparse} implies that since each factor $\phi$ is multiquadratic and $s$-sparse, then so are all its factors. In particular, it is divisible by some multiquadratic irreducible factor of $f$, which must be in the list. Therefore, for every pair $\phi, \psi$, we check if $\psi|\phi$ (using \Cref{general divisibility testing}). The candidates $\phi$ for which there are no other candidates that divide them are the desired irreducible factors.

\paragraph{Computing the multiplicities:}
   
    We have found all the multiquadratic, $s$-sparse irreducible factors of $f$. To compute their multiplicities as factors of $f$ we invoke \Cref{Recovering multiplicity of general factors}.

\paragraph{Arbitrary fields:}

    For arbitrary fields we let $m$ be as in \Cref{thm:multi rational interpolation theorem} for the case of multiquadratic polynomials. Instead of using \Cref{general divisibility testing} we apply the "moreover" part of \Cref{general characteristic divisibility testing}. Similarly, instead of \Cref{Recovering multiplicity of general factors} we rely on  the "moreover" part of \Cref{lem:Recovering multiplicity of primitive divisors}.

\paragraph{Running time:}

    Note that for fields of characteristic zero or larger than $2d$, each step runs in time  $\poly(n, d!, s^d, \ell)$. This is the running time guaranteed by \Cref{thm:multi rational interpolation theorem}, \Cref{general divisibility testing}, and  \Cref{Recovering multiplicity of general factors}.
    Since we have one recursive call the total running time is also $\poly(n, d!, s^d, \ell)$. For arbitrary fields, the running time is $\poly(n, (d^2)!, s^{d^2}, \ell)$, as this is what we obtain from \Cref{thm:multi rational interpolation theorem}, the "moreover" part of \Cref{general characteristic divisibility testing} and the "moreover" part of \Cref{lem:Recovering multiplicity of primitive divisors}.

    \medskip
    
    This completes the description, and analysis, of the algorithm.
\end{proof}

\section{Computing all Sparse Divisors of a Sparse Polynomial}\label{sec:sparse-div-of-sparse}

In this section, we analyze the set of sparse divisors of a sparse polynomial of bounded individual degree, and give a polynomial time
algorithm for constructing it. We then apply our methods to give a partial solution to the problem
raised in \cite{DuttaST24} regarding the recovery of sparse irreducible polynomials of bounded
individual degree from blackbox access to their product (\Cref{q:dst-factor-product-nsd}).

\subsection{Number of Divisors of Sparse Polynomials of Bounded Individual Degree}
We first prove an upper bound on the number of divisors (that are not divisible by any monomial) of sparse polynomials with bounded individual degree.
\NumDivisors*

\begin{proof}
    By dividing by a monomial, we may assume that $f$ is not divisible by any $x_i$.
    We choose a sequence of variables in the following iterative way. First, choose some $x_{i_1}\in \textrm{var}(f)$.
    Suppose we have   already chosen $x_{i_1},\ldots,x_{i_k}$, and consider the factorization of $f$ into irreducible polynomials.
    If there is an irreducible factor $\phi$ that contains none of the variables chosen thus far, we choose some $x_{i_{k+1}} \in \textrm{var}(\phi)$. 
    If no such factor exists, then the process terminates.
    
    Assume that in the process we chose $k$ variables (in sequence), which we denote without loss of generality by $x_1,x_2,\ldots,x_k$ (i.e., we assume that $i_1=1$ etc.).
    The sequence of chosen variables naturally corresponds to a sequence of \emph{multisets} of irreducible factors of $f$.
    The first multiset $I_1$ is the multiset of irreducible factors that depend on $x_1$ (counted with multiplicities); the second multiset $I_2$ consists of irreducible factors that do not depend on $x_1$
    but depend on $x_2$; the $j$-th multiset $I_j$ is the multiset of irreducible factors that do not  depend on $x_1,\ldots,x_{j-1}$ but
    depend on $x_j$.
    It is easy to see that \(\{I_j\}_{1\leq j\leq k}\) form a partition of the multiset of irreducible factors of $f$.

    As $f$ is of bounded individual degree $d$ and every factor of $I_j$ depends on $x_j$, it must be the case
    that $|I_j|\leq d$ as a multiset. Thus, the total number of irreducible factors of $f$, counted with multiplicities, is at most $kd$. Therefore, to
    prove our theorem, it is sufficient to show that $k\leq \log s$.

    To this end, let $h_j$ be the product of all the polynomials in $I_j$, so that $f=\prod_{j=1}^{k}{h_j}$.
    \begin{claim}\label{polytopes argument}
        Denote $f_j=\prod_{t=k-j+1}^{k}{h_t}$. Then, $2^{j}\leq |V(P_{f_j})|$, where $P_{f_j}$ is the Newton polytope
        corresponding to $f_j$ and $V(P_{f_j})$ is the set of its vertices (as defined in \Cref{Newton polytope}).
    \end{claim}
    We first conclude the proof of the theorem and then prove the claim.
    By the claim, it follows that $2^{k}\leq |V(P_{f_k})|=|V(P_f)|\leq s$, and thus $k\leq \log s$, as we wanted to show.
    Note that we used here that $f$ is $s$-sparse, to get, using \Cref{number of vertices is monotone}, that $|V(P_f)|\leq s$. Note that as $d\log s$ is an upper bound on the number of factors, including multiplicities, the number of divisors of $f$ is at most $2^{d\log s}=s^d$.
\end{proof}

    \begin{proof}[Proof of \Cref{polytopes argument}]
        The proof is by induction of $j$. For $j=1$, the claim follows from  the fact that $f_1=h_k$ is not
        a monomial (as we removed all monomial divisors of $f$).
        
        For the inductive step, suppose that $2^{j}\leq |V(P_{f_j})|$. We will show that $2^{j+1}\leq |V(P_{f_{j+1}})|$.
        To this end, express $h_{k-j}$ as a polynomial in $x_{k-j}$. That is, $h_{k-j}=\sum_{i=0}^{d_{k-j}}{g_ix_{k-j}^{i}}$ where $g_i\in \mathbb{F}[x_{k-j+1},\ldots,x_n]$. By construction, $x_{k-j}\in \textrm{var}(h_{k-j})$,
        so $d_{k-j}=\deg_{x_{k-j}}(h_{k-j})>0$. 
        Moreover, since $x_{k-j} \nmid h_{k-j}$, it must be the case that $g_0 \neq 0$.
        Observe that $f_{j+1}=f_jh_{k-j}=\sum_{i=0}^{d_{k-j}}{f_{j}g_ix_{k-j}^{i}}$, where $f_jg_i\in \mathbb{F}[x_{k-j+1},\ldots,x_n]$.
        We claim the following.
        \begin{enumerate}
            \item Any vertex $u\in V(P_{f_{j}g_{0}})\subset \mathbb{R}^{n+j-k}$ corresponds to a vertex $(0,u)\in \mathbb{R}^{n+j-k+1}$
            of $V(P_{f_{j+1}})$.
            \item Any vertex $v\in V(P_{f_{j}g_{d_{k-j}}})\subset \mathbb{R}^{n+j-k}$ corresponds to a vertex $(d_{k-j},v)\in \mathbb{R}^{n+j-k+1}$
            of $V(P_{f_{j+1}})$.
            \item We have $|V(P_{f_{j}g_{0}})|\ge 2^{j}$ and $|V(P_{f_{j}g_{d_{k-j}}})|\ge 2^{j}$.
        \end{enumerate}
        The first and second claims are clear from the fact that the described sets of vertices are precisely the sets
        of vertices with the lowest/largest value on the first coordinate (that is, the coordinate that corresponds to the
        exponent of $x_{k-j}$), respectively. For the third claim, note that by \Cref{number of vertices is monotone}
        and the induction hypothesis we get
        \begin{equation*}
        |V(P_{f_{j}g_{0}})|\ge |V(P_{f_j})|\ge 2^{j}, \quad \text{and} \quad |V(P_{f_{j}g_{d_{k-j}}})|\ge |V(P_{f_j})|\ge 2^{j} \ .
        \end{equation*}
        Therefore, we showed that overall there are at least $2^{j+1}$ elements in $V(P_{f_{j+1}})$, completing the proof of inductive step.
    \end{proof}

\begin{remark}
    The above theorem is easily seen to be tight, for example by considering $f(x_1,\ldots,x_n) = \prod_{i=1}^{\log s}({{x_i}^d - 1})$
    over $\mathbb{C}$.
\end{remark}

As corollary we obtain the following upper bound on the number of sparse divisors of a product of factors of $(n,s,d)$-polynomials.

    \begin{corollary}\label{general bound on the number of sparse divisors}
        Let $f$ be a product of $\ell$ elements of $\cCF(n,s,d)$.
        The number of $(n,s,d)$-sparse divisors of $f$ that are not divisible by a monomial is at most 
        \begin{equation*}
             \notag
             \sum_{i=0}^{d\log s}\binom {\ell d\log s}{i}\leq s^{d(2+\log \ell)}.
        \end{equation*}
    \end{corollary}
    \begin{proof}
        We assume without loss of generality that $f$ is not divisible by any monomial.

        Suppose that $f=\prod_{i=1}^{\ell}{\phi_i}$ where $\phi_i\in \cCF(n,s,d)$, and denote by $g_i$ an $(n,s,d)$-sparse polynomial that $\phi_i$ divides.
        
        By \Cref{thm:bound on the number of irreducible factors of sparse polynomials}, each $g_i$ has at most $d\log s$ irreducible
        factors (we repeat a factor several times according to its multiplicity), so this must also be the case for $\phi_i$. Therefore, $f$ has at most $\ell d\log s$ irreducible factors, which we denote by $\{\psi_j\}$. We note that this is a multiset where factors may appear several times according to their multiplicity in the factorization of $f$.
          
        We claim that any $(n,s,d)$-sparse divisor of $f$ is a product of at most $d\log s$ such factors $\psi_j$.
        Indeed, any divisor of $f$ is a product of a subset of the $\psi_j$s, and by \Cref{thm:bound on the number of irreducible factors of sparse polynomials}, any such $s$-sparse divisor cannot  be a product of more than $d\log s$ of the $\psi_j$s. Thus, the number of
        $s$-sparse factors of bounded individual degree $d$ of $f$ is bounded from above by        \begin{equation}\label{bound for product}
            \sum_{i=0}^{d\log s}\binom{\ell d\log s}{i}.
        \end{equation}
        Now, if $\ell=1$, \eqref{bound for product} equals $2^{d\log s}=s^d$. If $\ell \ge 2$, we have that $\binom{\ell d\log s}{d\log s}$ is 
        the largest summand in \eqref{bound for product}, and it can be easily seen that our expression \eqref{bound for product} is bounded above by $d\log s\binom{\ell d\log s}{d\log s}$.
        Either way, we get that the number of relevant factors is upper bounded by
        \begin{align*}
            \max\bc{s^d, d\log s \binom{\ell d\log s}{d\log s}} &\leq \max\bc{s^d, d\log s \left(\frac{e\ell d\log s}{d\log s}\right)^{d\log s}}\\ &= \max\bc{s^d, d\log s (e\ell)^{d\log s}}\\
            &=  (d\log s) {s}^{d(\log \ell+\log e)} \\
            &\leq s^{d(2+\log \ell)},
        \end{align*}
        where the last inequality follows from the fact that for every $c>0$ and $x \geq 1$ it holds that $(ce)\ln x \leq x^c$, which
        can be verified by noting that $\frac{\ln x}{x^c}$ is maximized at $x=e^{1/c}$. Applying  this inequality for $c=d(2-\log e)$ and using the fact that $e\ln 2(2 - \log e) > 1$ we obtain the claimed result.
    \end{proof}

\subsection{Finding Sparse Divisors of Sparse Polynomials with Bounded Individual Degree}

In this section we prove \Cref{finding sparse divisors of sparse polynomials}. We recall its statement.

\SparseDivisorFinding*

\begin{proof}
    As in the proof of \Cref{finding multiquadratic factors of a product of sparse polynomials}, we follow the steps of the meta-algorithm of \Cref{sec:meta}, first studying fields such that  $\Char(\F)=0$ or $\Char(\F)>2d$, and then  arbitrary fields. For the sake of presentation, we shall assume that $f\in \cP(n+1,s,d)\subset \mathbb{F}[x_0,\vx]$.
    
    \paragraph{Preprocessing:} This step is performed as in the proof of \Cref{finding multiquadratic factors of a product of sparse polynomials}. We first learn the monomial divisors of $f$. We thus assume $x_0\nmid f$.
    Write $f=\sum_{i=0}^{d}{f_ix_0^i}$, where $f_i\in \mathbb{F}[\vx]$, and note that $f_0\neq 0$ by our assumption.
    Furthermore, $f_0\in \cP(n,s,d)$. Let $\tilde{f}$ denote the normalization of $f$ and $\hat{f}:=f(y,G)$, for $G=G_{(m)}$ with $m=(2nsd)^2$, as given by \Cref{lem:first property of G} for $(n,s,2d)$-sparse polynomials.  Using \Cref{Interpolation of polynomials} and \Cref{thm:polynomial factorization}, obtain the factorization of $\hat{f}=f(y,G)$
        into irreducible polynomials. As $\tilde{f}$ is reversed-monic with respect to $y$, so is $\hat{f}$. By \Cref{cla:reversed-monic bound} it follows that $\hat{f}$ has at most $d$ irreducible factors, and therefore at most $2^d$ divisors overall.
        Since $\hat{f}$ has $O(1)$ variables and degree $\poly(m)$, this step runs in time $\poly(m)$.

    \paragraph{Recursive call:}

    Since $f_0$ is clearly an $(n,s,d)$-sparse polynomial,  we can, by applying one recursive call, find the set of all $s$-sparse divisors of it. Denote this set by $\cF(f_0)$. By \Cref{thm:bound on the number of irreducible factors of sparse polynomials}, we have that $|\cF(f_0)|\leq s^d$.

    \paragraph{Generating, and pruning, a list of candidate solutions: }
    
    By \Cref{cla:normalization of a polynomial}, every divisor $h=\sum_{i=0}^{d}c_ix_0^i$ of $f$ corresponds
    to a divisor $\tilde{h}=\sum_{i=0}^{d}\frac{c_i}{c_0}{f_0}^i{y}^i$ of $\tilde{f}$, and therefore to a divisor
    \[\hat{h} :=\tilde{h}(y,G)=\sum_{i=0}^{d}{\frac{c_i(G)}{c_0(G)}f_0(G)^iy^i}\] 
    of $\hat{f}$.
    As before, we run over all divisors $\hat{h}=\sum_{i=0}^{d}e_iy^i$ of $\hat{f}$. For each such divisor we produce a list $\cF_{\hat{h}}$ of at most $|\cF(f_0)|$ potential candidates for $s$-sparse divisors of $f$ as follows: For each $c\in\cF(f_0)$, let $\hat{e}_i=\left((x_1\cdot \ldots \cdot x_n)^d c\right)(G)\frac{e_i}{f_0^i(G)}$. Using \Cref{lem:first property of G} we sparse interpolate $\hat{e}_0,\hat{e}_1,\ldots, \hat{e}_d$ to $(n,s,2d)$- polynomials $\tilde{e}_0,\tilde{e}_1,\ldots, \tilde{e}_d$. If the interpolation failed, for any reason, then we move to the next $c\in\cF(f_0)$. If the interpolation succeeded then we compute the largest monomial $M$ that divides each of the polynomials $\tilde{e}_0,\tilde{e}_1,\ldots \tilde{e}_d$. If $h=\frac{\sum_{i=0}^{d}{\tilde{e}_ix_0^i}}{M}$ is an $(n,s,d)$-sparse polynomial then we run   \Cref{general divisibility testing} to check whether it divides $f$, and in that case
    we add it to $\cF_{\hat{h}}$.

    We let $\cF'(f_0)$ be those elements of $\cF(f_0)$ that divide $f$ (again, using \Cref{general divisibility testing} to check that).
    
    We return the set $\cF(f)=\cF'(f_0)\cup \left(\cup_{\hat{h}\mid \hat{f}} \cF_{\hat{h}}\right)$, which is the union of all the sets produced by the algorithm.

    \begin{claim}
       $\cF(f)$ is the set of all the divisors $h=\sum_{i=0}^{d}{c_ix_0^i}$ of $f$. 
    \end{claim}
    \begin{proof}
        Since $\cF(f)$ only contains divisors of $f$, we need to show that it contains all the divisors of $f$.
        
        Recall that every divisor $h=\sum_{i=0}^{d}{c_ix_0^i}$ of $f$ gives rise to a divisor        $\hat{h}=\sum_{i=0}^{d}{\frac{c_i(G)}{c_0(G)}f_0^i(G)y^i}$ of $\hat{f}$. Moreover, $c_0$ is an $(n,s,d)$-sparse polynomial satisfying $c_0 \mid f_0$. Therefore, for if $M$ is the largest monomial dividing $c_0$ then $c_0=Mc$ for some  $c \in \cF$.
        
        Denote $e_i=\frac{c_i(G)}{c_0(G)}f_0^i(G)$, so that $\hat{h}=\sum_{i=0}^{d}e_iy^i$. Observe that for $\tilde{M}=(x_1\cdot \ldots \cdot x_n)^d/M$ we have
        \[\hat{e}_i=\left((x_1\cdot \ldots \cdot x_n)^d c\right)(G)\frac{e_i}{f_0^i(G)}=\tilde{M}(G)\cdot \frac{c_0(G) \frac{c_i(G)}{c_0(G)}f_0^i(G)}{f_0^i(G)}=\tilde{M}(G)c_i(G).\]
        Since $c_i$ is an $(n,s,d)$-sparse polynomial and $\tilde{M}$ is a monomial of individual degree bounded by $d$, $\tilde{M}c_i$
        is an $(n,s,2d)$-sparse polynomial.
        Therefore, \Cref{lem:first property of G} will return the polynomial $\tilde{M}c_i$ when run on $\hat{e}_i$. Consequently, we will add $h=\sum_{i=0}^{d}c_i x_0^i$ to $\cF_{\hat{h}}$ and hence to $\cF(f)$.
    \end{proof}

\paragraph{Arbitrary fields:}

   The assumption on the field is used only to invoke \Cref{general divisibility testing} for pairs of $(n+1,s,d)$-sparse polynomials. \cite[Lemma 42]{volkovich2017some} gives a deterministic algorithm for divisibility testing of pairs of $(n+1,s,d)$-sparse polynomials whose  running time is $\poly(n,d!,s^d)$. Therefore,
    by replacing \Cref{general divisibility testing} with this result, we can remove the assumption on the field.

\paragraph{Running time:}
    As before, the preprocessing step requires $\poly(nsd)$ time. \Cref{lem:first property of G} requires $\poly(nsd)$ time.
    \Cref{general divisibility testing} and \cite[Lemma 42]{volkovich2017some}  take $\poly(n, d!, s^d)$ time. Since we make one recursive call the total running time is $\poly(n, d!, s^d)$. 
    \end{proof}

\begin{remark}
    The algorithm above solves another rational interpolation problem. Suppose we have blackbox access to
    $\frac{af_0}{b}$, where $f_0$ is a known $s$-sparse polynomial of bounded individual degree $d$, and $a,b$ are unknown, relatively prime, $(n,s,d)$-sparse polynomials. Suppose further that $b|f_0$. Then it is possible to recover $a,b$ in $\poly(n,d!,s^d)$-time.
    This algorithm works over arbitrary fields as well. Indeed, we first compute all the divisors of $f_0$ using \Cref{finding sparse divisors of sparse polynomials}  and then proceed as above.
\end{remark}

As a corollary we obtain the following result that settles \Cref{q:factoring-sparse-irred-product} for the case of bounded individual degree polynomials. We note that this result also follows from the techniques of earlier work such as \cite{bhargava2020deterministic,volkovich2015deterministically}, but we explicitly state it as we could not find it in the literature.

\begin{corollary}[{\bf \Cref{q:factoring-sparse-irred-product} for $(n,s,d)$-sparse polynomials}]\label{cor:factoring-product-sparse}
    Let $\mathbb{F}$ be a field.
   There exists a deterministic, $\poly(n,d!,s^{d})$-time algorithm that given an $(n,s,d)$-sparse polynomial $f$ (over $\mathbb{F})$ that factors as a product of irreducible $s$-sparse polynomials $f=\prod g_i^{e_i}$, outputs all the $g_i$s. Moreover, there is a deterministic algorithm for computing the multiplicities $e_i$, where:
    If $\Char(\F)=0$ or $\Char(\F)>2d$, the algorithm runs in time $\poly(n,d!,s^d)$.
    If $\mathbb{F}$ is arbitrary, the algorithm runs in time $\poly(n,(d^2)!,s^{d^3})$.
\end{corollary}

\begin{proof}
    Computing the multiplicities $k_i$ of $x_i$ as factors of $f$ is trivial as $f$ is given to us (it is $(n,s,d)$-sparse). Therefore, we may assume that $f$ is not divisible by any monomial.
    By \Cref{finding sparse divisors of sparse polynomials} we can find the set $\cF$ of all the $s$-sparse divisors of $f$. By \Cref{thm:bound on the number of irreducible factors of sparse polynomials}, we have that $|\cF|\leq s^d$. As all irreducible factors $g_i$ of $f$ are $s$-sparse, $g_i\in \cF$.
    Thus, we are left with the problem of identifying the $g_i$s from $\cF$. To that end, for every pair of two (non-constant) elements of $\cF$, use the divisibility testing of \cite{volkovich2017some} to check if one divides the other. All non-constant elements of $\cF$ that are not divisible by any other member of $\cF$ are clearly the desired irreducible factors $g_i$.

    As for the "moreover" part, we compute multiplicities using \Cref{Recovering multiplicity of general factors} if $\Char(\F)=0$ or $\Char(\F)>2d$; and \Cref{lem:Recovering multiplicity of primitive divisors} for arbitrary fields.
\end{proof}

\subsection{Factoring a Product of Sparse Irreducible Polynomials}\label{sec:factor-product-sparse-irred}

In this subsection we show our partial solution to the open problem from \cite{DuttaST24} (\Cref{q:dst-factor-product-nsd}).
As shown in the previous sections, it is natural to consider products of factors of sparse polynomials, as well as to seek a larger class of sparse divisors beyond the irreducible ones. Accordingly, we deduce our algorithm for solving \Cref{q:dst-factor-product-nsd} from the following more general theorem.

\generalfactorizationalgorithm*

\begin{proof}

    We again use the meta-algorithm of \Cref{sec:meta}. 
    Let $f\in \mathbb{F}[x_0,\vx]$ be a product of $\ell$ polynomials in $\cCF(n+1,s,d)$.

    \paragraph{Preprocessing:} This step is identical to the corresponding step in the proof of \Cref{finding multiquadratic factors of a product of sparse polynomials}, where we use $G=G_{(m)}$ for $m=(2nds)^2$, in the definition of $\hat{f}$.  Note that $\tilde{f}$ has at most $\ell d$ irreducible factors (as polynomials in $y$). Hence, it has $O((\ell d)^d)$ divisors of degree at most $d$ in $y$.

\paragraph{Recursive call:}

Since $f_0$ is a product of $\ell$ polynomials in $\cCF(n,s,d)$ (by \Cref{lem:closure under taking free terms}), we can find the set of all $(n,s,d)$-sparse divisors of $f_0$ not divisible by a monomial by applying one recursive call. Denote this set by $\cF(f_0)$. By \Cref{general bound on the number of sparse divisors}, we have that $|\cF(f_0)|\leq \poly(s^{d\log \ell})$.

\paragraph{Generating, and pruning, a list of candidate solutions: }

This step is identical to the corresponding step in the proof of \Cref{finding sparse divisors of sparse polynomials}. Indeed, generating the set of candidates from $\cF(f_0)$ can be done in a similar manner, when we consider divisors $\hat{h}$ of $\hat{f}$ of degree at most $d$ in $y$, as we are interested only in divisors of bounded individual degree $d$. Pruning can be done again by invoking \Cref{general divisibility testing}.

\paragraph{Arbitrary fields:}

    The only problem with extending the proof to arbitrary fields is the use of \Cref{general divisibility testing}. Therefore, for fields with small characteristic, we assume that  $f$ is a product of $s$-sparse polynomials. In this case we can apply the "moreover" part of \Cref{general divisibility testing} and the argument goes through.
    
    \paragraph{Running time:}

    As in the previous proofs, the running time is determined by the size of the lists and the running time of \Cref{general divisibility testing}. The claim bound follows.
\end{proof}

We finish with the following straightforward corollary of the last theorem, which gives the promised partial solution to \Cref{q:dst-factor-product-nsd}.
\OpenProblemTheorem*

\begin{proof}
    By \Cref{General factorization algorithm} we can find in deterministic $\poly(n,d^d,s^{d\log \ell},\ell^d)$-time
    all the $s$-sparse divisors of $f$ of bounded individual degree $d$. Applying \Cref{general divisibility testing}, we can determine, for any two candidates $\phi$ and $\psi$, whether $\psi \mid \phi$. Since every irreducible factor of $f$ is $s$-sparse with individual degree at most $d$, it follows that the irreducible factors are precisely those candidates $\phi$ that are not divisible by any other candidate $\psi$. To compute the multiplicity of each obtained irreducible factor as a factor of $f$, just invoke \Cref{Recovering multiplicity of general factors}.

    As for the "moreover" part, use the "moreover" part of \Cref{General factorization algorithm} to get all the $(n,s,d)$-sparse divisors of $f$ in $\poly(n,{(d^2)}!,s^{d^3},\ell^d)$-time. Replace  \Cref{general divisibility testing} by its "moreover" part (or by the divisibility test of     \cite{volkovich2017some}) to check for every two candidates if one divides the other, and replace the use of \Cref{Recovering multiplicity of general factors} by the "moreover" part of \Cref{lem:Recovering multiplicity of primitive divisors} to compute multiplicities.
\end{proof}

\section{Factoring $(n,s,d)$-Sparse Polynomials}\label{sec:BSV based results}
In this section we apply our methods to obtain factorization results for  $(n,s,d)$-sparse polynomials, without any extra assumptions on the sparsity of their factors. For that, we rely on the sparsity upper bound of \Cref{thm: BSV bound}.

As discussed in \Cref{sec:BSV overview}, we give two algorithms: one for factoring a single $(n,s,d)$-sparse polynomial, and one for factoring
a product of such polynomials. The  first algorithm runs in time  $\poly(n,S)$ (where $S$ is the bound of \Cref{thm: BSV bound}), while the second runs in $\poly(n,d!,S^d)$-time. As we shall shortly see, the main reason for this difference is that it is easy to check if a given $S$-sparse
polynomial divides the input polynomial in case it is $s$-sparse, since in this case the quotient should be $S$-sparse as well
and therefore can be interpolated and reconstructed, but this is no longer the case for divisors of products of $s$-sparse polynomials. 

\subsection{Factorization of Sparse Polynomials}
In this subsection we prove the following theorem, which is an improvement to the algorithm of \Cref{thm: BSV factorization algorithm}.

\SparseFactorization*

Observe that, given only the bound in \Cref{thm: BSV bound}, this is the best one can hope for.
Furthermore, as can be easily seen from the proof, an improvement to \Cref{thm: BSV bound} automatically implies an improvement to the running time of our algorithm: if the bound in \Cref{thm: BSV bound} is improved to a smaller bound $M$, then the running time of our algorithm will become
$\poly(n,M,s^d)$.

\begin{proof}

    The proof follows the lines of the proof of \Cref{finding sparse divisors of sparse polynomials}, with some small modifications that we shall explain. The algorithm is the same for any field $\F$, regardless of the characteristic.
    
    Denote by $S$ an explicit sparsity bound that follows from \Cref{thm: BSV bound} for $(n+1,s,d)$-sparse polynomials, $S=\poly(n,s^{d^2\log n})$.
    
    \paragraph{Preprocessing step:}

    We perform the same preprocessing step as in the proof of \Cref{finding sparse divisors of sparse polynomials}, where we define $\hat{f}$ with respect to the generator $G=G_{(m)}$ for $m=(2ndS)^2$.
    
    \paragraph{Recursive call:} Since $f_0$ is an $(n,s,d)$-sparse polynomial, a recursive call to the algorithm will return a list $\cF'(f_0)$ containing its irreducible factors. Note that these factors may have sparsity $S$. By \Cref{thm:bound on the number of irreducible factors of sparse polynomials} the list has size at most $d\log s$, including multiplicities, when removing the variables it might contain (that is, we consider $\cF'(f_0)\setminus \{x_i\}_{1\leq i\leq n}$ instead of $\cF'(f_0)$). Therefore, $f_0$ has at most $s^d$ divisors not divisible by a monomial, that can be enumerated, and interpolated, in time $s^d\poly(S)$, by going over all subsets of irreducible divisors (taking multiplicities into account). We denote the list of all divisors not divisible by a monomial of $f_0$ by $\cF(f_0)$.

    \paragraph{Generating, and pruning, a list of candidate solutions: }
    
    Generating the list of candidates is done as in the proof of \Cref{finding sparse divisors of sparse polynomials}. Indeed,  $\cF(f_0)$ contains all the divisors (not divisible by a monomial) of $f_0$, so we can proceed in exactly the same way to find a list of candidates for divisors of $f$ (not divisible by a monomial). Regarding pruning, we note that in order to decide if a candidate polynomial $g$ is a divisor of $f$, we don't need to use \Cref{general divisibility testing}. Indeed, if $g\mid f$, then $f/g$ is $S$-sparse. Therefore, we can interpolate $f(G)/g(G)$ as an $S$-sparse polynomial $h$ and check whether $g\cdot h=f$. 
    
    This shows how to find all the divisors of $f$ not divisible by a monomial. To conclude the proof, we need to show how to use this information to get the factorization of $f$ into irreducible polynomials.
    We first show how to identify the irreducible factors among the divisors that we have   found. 
    Observe that if a divisor $g$ is reducible, then its factors are also divisors of $f$. We can therefore check whether there exists a pair of divisors $h_1,h_2$ such that $g=h_1\cdot h_2$, and return that $g$ is irreducible if no such $h_1,h_2$ were found.
    
    \paragraph{Computing the multiplicities: }
    So far, we have the list of irreducible factors of $f$ that are not monomials. Since the largest monomial dividing $f$ was computed in the preprocessing step, we are left with the problem of computing the multiplicity of each of the irreducible factors of $f$. This can be done as in the previous step. Going over all values of $k\geq 0$ one by one we check whether there exists a divisor $h$ of $f$ such that  $h\cdot g^k=f$. We stop at the first $k$ where such a divisor does not exist and return that the multiplicity of $g$ is $k-1$. Since, for $i<k$, all $g^i$ are divisors of $f$, their sparsities are all $\poly(S)$, hence this algorithm runs in $\poly(S)$-time.
    
    \paragraph{Running time:}
    It is clear that the running time is polynomial in the sparsity $S$, the number of variables $n$, and the size of the lists which is upper bounded by $s^d$.    
\end{proof}

\subsection{Factorization of Products of Sparse Polynomials}
We now prove the following theorem. Note that this theorem is different than \Cref{General factorization algorithm}; while \Cref{General factorization algorithm} finds only the $s$-sparse divisors of a product of polynomials in $\cCF(n,s,d)$, \Cref{improved BSV for products} computes the factorization of such products into irreducible factors, even if these are not $s$-sparse.

\SparseProductFactorization*

\begin{proof}

    Denote by $S=s^{O(d^2\log (n+1))}$ the sparsity bound given by \Cref{thm: BSV bound} for $(n+1,s,d)$-sparse polynomials.
    
    We follow the lines of the proof of \Cref{finding multiquadratic factors of a product of sparse polynomials}, but combine it the sparsity bound of \Cref{thm: BSV bound}.

    \paragraph{Preprocessing:} This is the same as in the proof of \Cref{finding multiquadratic factors of a product of sparse polynomials}, where $\hat{f}$ is defined as $\hat{f}=\tilde{f}(y,G)$, where $G=G_{(m)}$. Here, $m$ is as given by \Cref{thm:multi rational interpolation theorem} for $(n,S,d)$-sparse polynomials (that is, $m=\poly(n,d!,S^d)=\poly(n,d!,s^{d^3\log n})$). By \Cref{cla:reversed-monic bound}, $\hat{f}$ has at most $\ell d$ irreducible factors (as polynomials in $y$), and hence at most $O((\ell d)^d)$ divisors of degree at most $d$ in $y$.

    \paragraph{Recursive call:} By \Cref{lem:closure under taking free terms}, $f_0$ is a product of $\ell$ elements of $\cCF(n,s,d)$. Thus, by making one recursive call, we obtain its factorization into irreducible polynomials. Denote by $\cF(f_0)$ the set of irreducible factors of $f_0$.

    \paragraph{Generating, and pruning, a list of candidate solutions: }

    Since we have all the irreducible factors of $f_0$, we can proceed as in the proof of \Cref{finding multiquadratic factors of a product of sparse polynomials} and obtain a list of $(n,S,d)$-sparse  divisors of $f$, using  \Cref{thm:multi rational interpolation theorem} and \Cref{general divisibility testing}.

    \paragraph{Pruning the list:}

    We now have a set of candidates, all of which are $S$-sparse divisors of $f$. Note that if $g$ is such a divisor that is not irreducible, then every irreducible factor of $g$ must also be in the list. Therefore, to identify the irreducible factors, we invoke \Cref{general divisibility testing} (for $S$-sparse polynomials) to each pair of candidates $\phi, \psi$ and check if $\psi|\phi$. The candidates $\phi$ for which there are no other candidates that divide them are the desired irreducible factors.

    \paragraph{Computing the multiplicities:}

    We compute the multiplicities of the irreducible factors using  \Cref{Recovering multiplicity of general factors}.

    \paragraph{Arbitrary fields:}

    To obtain the result we replace the use of \Cref{thm:multi rational interpolation theorem} and \Cref{general divisibility testing} by their respective “moreover” parts. We also use \Cref{lem:Recovering multiplicity of primitive divisors} instead of \Cref{Recovering multiplicity of general factors}.

    \paragraph{Running time:}

    First consider fields of zero or large characteristic.
    The running time is determined by the number of divisors of $\hat{F}$, and the application of \Cref{thm:multi rational interpolation theorem},  \Cref{general divisibility testing}, and \Cref{Recovering multiplicity of general factors}
    for $(n+1,S,d)$-sparse polynomials. This yields the claimed $\poly(n,d!,S^d,(d\ell)^d)=\poly(n,s^{d^3\log n},(d\ell)^d)$ running time. 

    For arbitrary fields, the "moreover" parts of  \Cref{thm:multi rational interpolation theorem} and \Cref{general divisibility testing} and \Cref{lem:Recovering multiplicity of primitive divisors} imply a running time of $\poly(n,(d^2)!,S^{d^3},(d\ell)^d)=\poly(n,(d^2)!,s^{d^5\log n},\ell^d)$
\end{proof}

\begin{remark}
    In \Cref{thm:factor-of-product-of-nsd} we gave a $\poly(n,d^d,s^{d\log \ell},\ell^d)$-time algorithm for this problem, with the extra assumption that the irreducible
    factors of $f$ are in fact in $\cP(n,s,d)$. Hence, \Cref{improved BSV for products} is more general, and in fact gives a $\poly(n, s^{d^3\log n},(d\ell)^d)$ algorithm for \Cref{thm:factor-of-product-of-nsd}. That is, we have two algorithms for solving \Cref{q:dst-factor-product-nsd}: one is quasi-polynomial in the length of
    the product and polynomial in the number of variables, and the other is quasi-polynomial in the number of variables and polynomial in the number of terms in the product.
\end{remark}

\section{Future Work and Open Problems}\label{Section:open}
We conclude with a discussion of possible directions for future research.

First, it is somewhat disappointing that our main technical lemma,
\Cref{lem:G always preserves coprimality of factors},
is only valid over fields of zero or sufficiently large characteristic.
It would be interesting to extend this result to arbitrary fields.
As discussed in \Cref{Section 4}, some results are known in this direction;
the remaining challenge is to eliminate the need for the additional
non-degeneracy condition required in small characteristic.

Next, it would be extremely interesting to identify the irreducible factors
in \Cref{finding sparse divisors of sparse polynomials}.
In other words, one would like an efficient algorithm that, given an
$(n,s,d)$-sparse polynomial $f$, decides whether $f$ is irreducible
in time $\poly(n,s^d)$.
The techniques developed in this paper appear insufficient for addressing
this problem, as they only provide information about the coprimality of
$\phi(G_{(m,i)})$ and $\psi(G_{(m,i)})$ for $\phi,\psi \in \cCF(n,s,d)$,
and do not yield information about the number of irreducible factors.

It would also be interesting to generalize our rational interpolation theorem.
For instance, even the seemingly basic problem of deterministically
interpolating a quotient of coprime sparse polynomials in polynomial time
remains wide open.

Obtaining a polynomial-time algorithm for
\Cref{q:dst-factor-product-nsd}, thereby improving
\Cref{thm:factor-of-product-of-nsd}, also appears to require new ideas.
As observed throughout the paper, the main difficulty is that $c_0$ may
contain non-sparse irreducible factors, forcing the recursive call on $f_0$
to enumerate all possibilities and leading to a quasi-polynomial running time.
By contrast, the polynomial-time algorithm of
\Cref{finding multiquadratic factors of a product of sparse polynomials}
relies on the fact that $c_0$ (or $b$ in that setting) factors only into
sparse polynomials, due to
\Cref{factors of multiquadratic sparse polynomials are sparse}.
This structure, together with the ability to determine multiplicities via
\Cref{G preserves coprimality of factors when one factor is sparse},
suffices to obtain a polynomial-time algorithm.

Finally, any improvement to the bound in \Cref{thm: BSV bound} would directly
strengthen the results of \Cref{sec:BSV based results}, and would 
be of significant interest.

\bibliographystyle{alpha}
\bibliography{sources}

\appendix
\section{An Alternative Proof of a Result from \Cref{Section 4}}\label{appendixA}
In this section we describe an alternative proof of
\Cref{G preserves coprimality of factors when one factor is sparse}. While this lemma is valid over arbitrary fields, the proof presented here works only for zero characteristic fields or when the characteristic is larger than the individual degree. Thus, both \Cref{G preserves coprimality of factors when one factor is sparse} and \Cref{lem:G always preserves coprimality of factors} are stronger than what we present here. However, the proof method is slightly different, and we include it here for the possibility it would be useful for other applications.

While the proofs of \Cref{G preserves coprimality of factors when one factor is sparse} and \Cref{lem:G always preserves coprimality of factors} use the rank of the Sylvester matrix to study the factorization of a polynomial, our proof here uses an operator we call $\Delta_k$. The well known discriminant $\Delta$ of a polynomial (defined in \Cref{discriminant}) is an operator that can be used to detect the existence of irreducible factors of multiplicity greater than one. Its generalization $\Delta_k$, that we soon define, is an operator that detects irreducible factors of multiplicity at least $k+1$, for small values of $k$.
The second important property of $\Delta_k$ is that, like the discriminant, it does not  blow-up the sparsity of its input by too much.

We now give the definition of $\Delta_k$.

\begin{definition} \label{higher order discriminant}
    Let $f \in \mathbb{F}[x]$ be a non-constant polynomial. Define
    \begin{equation}
        \Delta_{k}(f)(\lambda_1,\ldots,\lambda_k) = \res(f, \lambda_1f^{(1)} + \ldots + \lambda_{k}{f^{(k)}}) \in \mathbb{F}[\lambda_1,\ldots,\lambda_k]
    \end{equation}
    where $f^{(i)}$ is the $i$-th derivative of $f$.
    
    For $f \in \mathbb{F}[x_0, x_1,\ldots,x_n]$, $\Delta_{k, x_0}(f) \in \mathbb{F}[x_1,\ldots,x_n,\lambda_1,\ldots,\lambda_k]$ is
    just $\Delta_k(f)$ when one views $f$ as an element of $\mathbb{F}(x_1,\ldots,x_n)[x_0]$. 
\end{definition}

\begin{example}
    For $k=1$, one sees that $\Delta_{1} = \lambda_1\Delta$. Thus, for example, $\Delta_1(f) = 0$ if and only
    if $\Delta(f)=0$.
\end{example}

As discussed, one should think of $\Delta_k$ as some generalization of $\Delta$ that captures the existence of factors of multiplicity
at least $k+1$, the same way $\Delta$ captures the existence of factors of multiplicity at least two.
The example justifies this statement for $k=1$; the general statement is given by the following lemma.

\begin{lemma} \label{property of higher discriminants}
    Let $f\in \mathbb{F}[x]$ be a non-constant polynomial of at most degree $d$, and suppose that $\operatorname{char}(\mathbb{F})$ is zero or larger than $d$.
    Then $f$ has an irreducible factor of multiplicity at least $k+1$ if and only if $\Delta_k(f) = 0$.
    In particular, if $f\in \mathbb{F}[x_0,x_1,\ldots,x_n]$ is such that $x_0 \in \textrm{var}(f)$ and $\deg_{x_0}(f)\leq d$, 
    under the same assumption on $\mathbb{F}$,
    then $f$ has an irreducible factor of multiplicity at least $k + 1$ that depends on $x_0$ if and only if $\Delta_{k, x_0}(f)=0$.  
\end{lemma}

\begin{proof}
    Note that as $\operatorname{char}(\mathbb{F})\in \{0\}\cup [d+1,\infty)$, and $\deg f\leq d$, every irreducible factor $\phi$ of $f$ is separable, that is, it factors to a product of different linear polynomials over $\overline{\mathbb{F}}$. One way to see it is to note that $\phi'\neq 0$ (because of the assumption on $\mathbb{F}$), and therefore $\phi'$ must be coprime to $\phi$, so by \Cref{basic resultant properties} we must have that $\Delta(\phi)=\res(\phi,\phi')\neq 0$, and hence (again by \Cref{basic resultant properties}) $\phi$ is square-free.
    
    Therefore, $f$ has an irreducible factors of multiplicity at least $k+1$ if and only if it has a root of multiplicity at least $k+1$ over $\overline{\mathbb{F}}$. 

    Suppose, without loss of generality, that $f$ is monic. In this case, by \Cref{basic resultant properties}, we know that:
    \begin{equation}
        \Delta_{k}(f) = \res(f, \lambda_1f^{(1)} + \ldots + \lambda_{k}{f^{(k)}}) = 
        \prod_{\alpha: f(\alpha)=0}{\left(\lambda_1f^{(1)}(\alpha) + \ldots + \lambda_{k}{f^{(k)}(\alpha)}\right)}    
    \end{equation}
    That is, $\Delta_{k}(f)$ is a product of linear functionals in $\overline{\mathbb{F}}[\lambda_1,\ldots,\lambda_k]$,
    and therefore vanishes if and only if one of these functionals vanishes.

    To this end, note that for a root $\alpha$ of $f$,
    a functional $\lambda_1f^{(1)}(\alpha) + \ldots + \lambda_{k}{f^{(k)}(\alpha)}$ vanishes
    if and only if $f^{(i)}(\alpha)=0$ for $0\leq i \leq k$,
    that is, if and only if $\alpha$ is a root of multiplicity at least $k+1$ of $f$. This proves the lemma.

    As for the "in particular" part, note that by Gauss' lemma (\Cref{lem:gauss}), irreducible factors of $f\in \mathbb{F}[x_0,x_1,\ldots,x_n]$ that depend on $x_0$
    correspond to irreducible factors of $f$ as an element of $\mathbb{F}(x_1,\ldots,x_n)[x_0]$.
    Therefore, the last statement of the lemma follows easily from the first by working over $\mathbb{F}(x_1,\ldots,x_n)$,
    noting that its characteristic is the characteristic of $\mathbb{F}$.
\end{proof}

\begin{remark}
    The condition on $\deg (f)$ and $\operatorname{char}(\mathbb{F})$ is necessary. For example, consider
    the polynomial $f(x,y) = x^{p}-y \in \mathbb{F}_p[x, y]$ for a prime $p$. $f$ is clearly irreducible, but
    $\Delta_{x}(f)=0$, as $f$ is a square (and, in fact, a $p$-th power) as an element of ${\overline{\mathbb{F}_p(y)}}[x]$.
\end{remark}

The lemma above provides one manner in which $\Delta_k$ generalizes $\Delta$. For our purpose, we need
to study another manner in which such a generalization holds - action on sparse polynomials.

\begin{lemma} \label{higher discriminants preserve sparsity}
    Let $f\in \mathbb{F}[x_0,x_1,\ldots,x_n]$ be an $s$-sparse polynomial with bounded individual degree $d$.
    Then, $\Delta_{k}{f} \in \mathbb{F}[x_1,\ldots,x_n,\lambda_1,\ldots,\lambda_k]$ is a $(2d)!{(ks)}^{2d}$-sparse
    polynomial of individual degree at most $2d^2$.
\end{lemma}

\begin{proof}
    The lemma follows immediately by the \cref{basic resultant properties}, by noting that $\lambda_1f^{(1)}(\alpha) + \ldots + \lambda_{k}{f^{(k)}(\alpha)}$
    is a polynomial of sparsity $ks$ and bounded individual degree $d$.
\end{proof}

We now show how to use $\Delta_k$ to obtain a similar result to \Cref{G preserves coprimality of factors when one factor is sparse}. To do that, we start with the following lemma.

\begin{lemma}\label{G does not  increase multiplicity of largest factor}
    Suppose $\mathbb{F}$ is a field for which $\operatorname{char}(\mathbb{F})\in \{0\} \cup [d+1,\infty)$.
    Let $f$ be an $(n,s,d)$-sparse polynomial defined over $\mathbb{F}$, and fix some coordinate $i$.
    Suppose that $f$ does not  have an irreducible factor of multiplicity at least $k+1$ as a polynomial in $x_i$.
    Then there exists an explicit $m=\poly(n,d!,{(ks)}^{d})$ for which $f(G_{(m,i)})$ also does not 
    have an irreducible factor of multiplicity at least $k+1$ as a polynomial in $u$.
\end{lemma}
    
\begin{proof}
    By \Cref{property of higher discriminants}, as $f$ does not  have an irreducible factor of multiplicity at least $k+1$ as a polynomial in $x_i$, we have that $\Delta_{k, x_i}(f) \neq 0$. By \Cref{higher discriminants preserve sparsity}, $\Delta_{k, x_i}(f)$ is a $\poly(d!{(ks)}^{d})$-sparse polynomial of bounded individual degree $2d^2$ on $n+d-1$ variables.
    Therefore, by \Cref{lem:first property of G}, for $m=((n+d)(d!{(ks)}^{d})(2d^2))^2=\poly(n,d!,{(ks)}^d)$, we know that $\Delta_{k, x_i}(f)({G}_{(m,i)})\neq 0$, where we substitute every coordinate $x_j$ for $j\neq i$ according to $G_{(m,i)}$ (leaving $\lambda_1,\ldots,\lambda_k$ alive).
    It is easy to see that $\Delta_{k,u}(f(G_{(m,i)})) = \Delta_{k, x_i}(f)(G_{(m,i)})$, so it follows
    that $\Delta_{k,u}(f(G_{(m,i)})) \neq 0$. The claim then follows from  the other direction of \Cref{property of higher discriminants}.
\end{proof}

We now apply the last lemma to prove the following weaker version of \Cref{G preserves coprimality of factors when one factor is sparse}.

\begin{lemma} \label{G preserves coprimality of factors when one factor is sparse using delta_k}
    Suppose $\mathbb{F}$ is a field for which $\operatorname{char}(\mathbb{F})\in \{0\} \cup [d+1,\infty)$.
    Let $\phi$ be an irreducible $s$-sparse polynomial of bounded individual degree $d$,
    and let $\psi$ be an irreducible factor of some $s$-sparse polynomial
    $g$ of bounded individual degree $d$, different than $\phi$, both defined over $\mathbb{F}$. Fix some $i$ for which $x_i\in \textrm{var}(\phi)$. Then, for 
    some explicit value of $m=\poly(n,d^{d^2},s^{d^3})$, we have that
    $\phi(G_{(m,i)})$ is coprime to $\psi(G_{(m,i)})$, when viewed as polynomials in $u$.
    Moreover, if $\phi$ is multiquadratic, taking $m=\poly(n,d^d,s^{d^2})$ actually suffices.
\end{lemma}

\begin{proof}
    Let $g=\prod{\psi} \cdot \phi^{k}$, each $\psi$ is irreducible not associate to $\phi$.
    From now on we view these polynomials as polynomials in $x_i$.

    Consider the polynomial $\tilde{g} = \phi^{d+1}g$. Its sparsity is at most $s^{d+2}$,
    its degree is at most $O(d^2)$, and $\phi$ divides it exactly $k+d+1$ times.
    As $\deg_{x_i}(g) \leq d$, every other irreducible factor of $\tilde{g}$ is of multiplicity at most $d-k < d+1$, so
    $\phi$ must be the factor of largest multiplicity of $\tilde{g}$.
    That is, $\tilde{g}$ has no factors of multiplicity more than $k+d+1$.
    By \cref{G does not  increase multiplicity of largest factor}, if
    $m=\poly(n,(d^2)!,{((k+d)s^{d+2})}^{d^2})=\poly(n,d^{d^2},s^{d^3})$
    the same is true for $\tilde{g}(G_{(m,i)})$ (as a polynomial in $u$).
    Thus, it must be the case that $\psi(G_{(m,i)})$ is prime to $\phi(G_{(m,i)})$ - 
    otherwise we would have had a factor of $\phi(G_{(m,i)})$ that appears
    at least $k+d+2$ times in the factorization of $\tilde{g}(G_{(m,i)})$ ($k+d+1$
    times that come from the $\phi$ factors of $\tilde{g}$, and one extra coming from the $\psi$ factor).

    As for the "moreover" part, note that if $\phi$ is multiquadratic then $\tilde{g}$ is of degree $O(d)$, so taking
    $m=\poly(n,d!,{((k+d)s^{d+2})}^{d})=\poly(n,d^{d},s^{d^3})$
    is enough.
\end{proof}

\begin{remark}
    In fact, the proof implies that $\phi(G_{(m,i)})$ is square free, since the generator $G$ preserves the non-vanishing of the discriminant $\Delta(\phi)$.
\end{remark}

\end{document}